\documentclass[11pt,english,american]{article}

\usepackage[english, american]{babel}
\usepackage{lmodern}
\usepackage{amsmath}
\usepackage{amsthm}
\usepackage{amssymb}
\usepackage{geometry}
\usepackage{xcolor}
\usepackage{latexsym}
\usepackage{prettyref}
\usepackage{mathtools}
\PassOptionsToPackage{normalem}{ulem}
\usepackage{ulem}
\usepackage{setspace}
\usepackage{comment}
\usepackage{dsfont}

\usepackage[bookmarks]{hyperref}

\usepackage{enumitem}

\numberwithin{equation}{section}
\numberwithin{figure}{section}

\theoremstyle{plain}\newtheorem{theorem}{Theorem}[section]
\theoremstyle{plain}\newtheorem{thm}{Theorem}[section]
\theoremstyle{plain}\newtheorem{lem}[theorem]{Lemma}
\theoremstyle{plain}\newtheorem{cor}[theorem]{Corollary}
\theoremstyle{plain}\newtheorem{ass}[theorem]{Assumption}
\theoremstyle{plain}\newtheorem{prop}[theorem]{Proposition}
\theoremstyle{definition}\newtheorem{defn}[theorem]{Definition}
\theoremstyle{remark}\newtheorem{rem}{Remark}

\newtheorem*{remarks}{Remarks}

\title{\huge  Derivation of the time-dependent Hartree equations for strongly interacting dense fermionic systems}

\AtBeginDocument{

}

\begin{document}

\global\long\def\P{\mathbf P}

\global\long\def\eN{\varepsilon}
\global\long\def\lsp{\left\langle}%
\global\long\def\rsp{\right\rangle}%
\global\long\def\re{\textnormal{Re}}%
\global\long\def\im{\textnormal{Im}}%
\global\long\def\ti#1{\tilde{#1}}%
\global\long\def\wti#1{\widetilde{#1}}%
 
\global\long\def\d{\text{d}}%
\global\long\def\RR{\mathbb{R}}%
 
\global\long\def\LaR{L_{\textrm{a}}^{2}(\mathbb{R}^{3N})}%
\global\long\def\HaR{H_{\textrm{a}}^{2}(\mathbb{R}^{3N})}%

\global\long\def\Hg{H^{\textnormal{g}}}%
\global\long\def\Ug{U^{\textnormal{g}}}%
\global\long\def\hg{h^{\textnormal{g}}}%
\global\long\def\uN{u^{(N)}}%
\global\long\def\us{\uN_{\leq}}%
\global\long\def\ubar{\overline{u_{\leq}}}%
\global\long\def\vN{v^{(N)}}%
\global\long\def\fN{f^{(N)}}%
\global\long\def\pd#1{p_{#1}^{\nabla}}%
\global\long\def\id{\textnormal{id}}%

\global\long\def\tr{\textnormal{Tr}}%
\global\long\def\op{\textnormal{op}}%
\global\long\def\as{\textnormal{as}}%
\global\long\def\id{\mathds{1}}

\global\long\def\red#1{\textcolor{red}{#1}}%

\global\long\def\green#1{\textcolor{olive}{#1}}%

\global\long\def\blue#1{\textcolor{blue}{#1}}%

\author{Duc Viet Hoang \thanks{Fachbereich Mathematik, University of Tübingen, Auf der Morgenstelle
10, 72076 Tübingen, Germany.\\ E-mail: \texttt{viet.hoang@uni-tuebingen.de}}
\and
David Mitrouskas \thanks{Institute of Science and Technology Austria (ISTA), Am Campus 1, 3400 Klosterneuburg, Austria.\\
E-mail: \texttt{mitrouskas@ist.ac.at}}
\and
Peter Pickl \thanks{Fachbereich Mathematik, University of Tübingen, Auf der Morgenstelle
10, 72076 Tübingen, Germany.\\ E-mail: \texttt{p.pickl@uni-tuebingen.de}}
}

\maketitle

\frenchspacing

\begin{spacing}{1.025}
\begin{abstract}
The time-dependent Hartree and Hartree--Fock equations provide effective mean-field descriptions for the dynamics of large fermionic systems and play a fundamental role in many areas of physics. In this work, we rigorously derive the time-dependent Hartree equations as the large-$N$ limit of the microscopic Schrödinger dynamics of $N$ fermions confined to a volume of order one and interacting via strong pair potentials. A central step in our analysis is the implementation of time-dependent gauge transformations, which eliminate the dominant contribution from the interaction potential in both the Schrödinger and Hartree evolutions.
\end{abstract}

\tableofcontents

\section{Introduction and main result}

\subsection{Introduction}

Understanding the dynamics of interacting many-body quantum systems is a central challenge in quantum physics and chemistry. However, direct solutions of the Schrödinger equation---whether exact or numerical---are typically intractable due to the exponential complexity that arises with increasing particle number. Fortunately, in many regimes, effective theories allow for significant simplification of the dynamics. They capture the essential behavior of the system while reducing its complexity, thereby enabling both analytical insight and efficient numerical implementation.

Prominent examples of effective theories are the Hartree and Hartree--Fock equations, which have become indispensable tools for modelling large fermionic systems, from atoms and molecules \cite{Fock1930,Szabo2012} to complex nuclei \cite{Gogny1986,bender2003}. In recent years, substantial progress has been made in rigorously deriving the Hartree and Hartree--Fock equations from first principles, that is, from the many-body Schrödinger dynamics. A particularly well-studied case is the semiclassical regime, where interactions between the fermions are weak and Planck’s constant is effectively small  \cite{Erdoes2003,Benedikter2014,Jaksic2016,
Porta2017,Fresta2023,Fresta2024,
Leopold2024}. In this regime, the particles typically occupy a volume of order one, and the total density varies on a macroscopic scale. To leading order, the dynamics are effectively governed by the Vlasov equation for a classical phase-space distribution \cite{Lions1993,Markowich1993,Saffirio2019,
Saffirio2020,Chen2021,Chong2023}. Hartree and Hartree--Fock theory recover the classical evolution in the appropriate limit, while also incorporating quantum corrections beyond the Vlasov dynamics. This regime captures many relevant physical systems, including electrons in large neutral atoms, degenerate Fermi gases in slowly varying external potentials, and dense astrophysical objects such as white dwarfs.

However, Hartree and Hartree--Fock theory are also widely used beyond the semiclassical regime, in systems where quantum effects are more dominant. This includes, for example, many-body systems in quantum chemistry and models of electronic structure in large molecules \cite{Kohn1999,helgaker,Szabo2012}, where the density can vary on microscopic length scales. In such settings, the system no longer admits a semiclassical approximation, yet mean-field theory often continues to provide an accurate description. Compared to the semiclassical regime, the mathematical derivation of fermonic mean-field equations in non-semiclassical settings has received somewhat less attention. Existing results have focused on weakly interacting systems, where the interaction is effectively subleading \cite{bardos2003,bardos2004,knowles,Petrat2016,Bach2016,Petrat2017}. As discussed, for instance, in \cite[Appendix B]{Bach2016} and \cite{Petrat2017}, the dynamics in these models are dominated by kinetic effects and are largely governed by free evolution, with only small corrections due to interactions.

The goal of this work is to rigorously derive the fermionic Hartree equations (sometimes called the reduced Hartree--Fock equations) from the $N$-body Schrödinger dynamics in a non-semiclassical regime with strong interactions—that is, when both quantum effects and inter-particle forces contribute at leading order. To this end, we focus on \emph{dense systems of strongly interacting fermions}, a notion recently introduced in \cite{Ruba2024}. From a physical perspective, these systems are natural, as they do not rely on any $N$-dependent weak coupling assumption. Instead, the effective dynamics emerge through a rescaling of the time variable. To our knowledge, the present work provides the first derivation of mean-field dynamics in this setting.

The starting point for our analysis is the microscopic description of a fermionic $N$-particle system, whose evolution is governed by the Schrödinger equation (we set $\hbar = 1$  throughout)
\begin{align}\label{eq:Schroedinger-eq}
i \partial_\tau \Phi_\tau = H \Phi_\tau,
\end{align}
for an antisymmetric wave function $\Phi_\tau \in L^2_a(\mathbb{R}^{3N})$. Here,
\begin{align}
L^2_a (\mathbb{R}^{3N}) = \left\{ \Phi \in L^2(\mathbb R^{3N})\ \middle|\ \Phi(x_{\pi(1)},\ldots, x_{\pi(N)}) = \text{sign}(\pi)\, \Phi(x_1, \ldots , x_N)\ \forall\, \pi \in S_N \right\},
\end{align}
with $S_N$ denoting the permutation group of $N$ elements. We consider Hamiltonians of the form
\begin{equation}
H = \sum_{1\le i\le N } (-\Delta_{i})\ + \sum_{1 \le i < j \le N } v(x_i - x_j), \label{eq:Hamiltonian}
\end{equation}
where $v$ is a sufficiently regular, $N$-independent two-body interaction potential. As usual, $\Delta_i$ denotes the Laplacian acting on the $i$-th particle. For notational simplicity, we set the mass of the fermions equal to $m = \frac{1}{2}$.

The initial $N$-particle state $\Phi_0$ is assumed to be close to a Slater determinant $\varphi_1^0 \wedge \ldots \wedge \varphi_N^0$ for some orthonormal one-particle orbitals $\varphi_1^0, \ldots, \varphi_N^0 \in L^2(\mathbb{R}^3)$. Our goal is to show that this Slater structure is approximately preserved over time, i.e., that $\Phi_\tau \approx \varphi_1^\tau \wedge \ldots \wedge \varphi_N^\tau$ in an appropriate sense, where the orbitals evolve according to the fermionic Hartree equations
\begin{equation}
i \partial_\tau \varphi_k^\tau = h(\tau)  \varphi_k^\tau \qquad \text{with} \qquad h(\tau) = -\Delta + v \ast \rho_\tau  \qquad \text{and} \qquad \rho_\tau = \sum_{k=1}^N |\varphi_k^\tau|^2 \label{eq:mf-eq}
\end{equation}
for $k = 1, \ldots, N$. Here, $v \ast \rho_\tau(x) = \int_{\mathbb R^3} dy\, v(x - y) \rho_\tau(y)$ denotes the mean-field potential, given by the convolution of the interaction potential $v$ with the Hartree density $\rho_\tau$. We note that the Hartree orbitals remain orthonormal for all $\tau > 0$.

Since no mean-field scaling is imposed on the potential, the interaction term in \eqref{eq:Hamiltonian} is much larger compared to the kinetic term. More precisely, for $N$ fermions confined to a volume of order one, the kinetic energy scales as $O(N^{5/3})$, whereas the interaction energy scales as $O(N^2)$. For such strongly interacting systems, the Hartree approximation is not expected to be accurate for times $\tau = O(1)$. Instead, the natural time scale is $\tau = \eN t$ with $\eN = N^{-2/3}$ and $t = O(1)$. The rescaling of the time variable can be justified by a simple Ehrenfest-type argument: When $N$ fermions are confined to a volume of order one, the Pauli exclusion principle forces them into high-momentum states, with typical momenta being of order $O(N^{1/3})$. At the same time, each particle experiences a force from all other particles of order $O(N)$. The change in momentum for each particle is thus heuristically given by
\emph{time} $\times$ \emph{force} = $O(\tau) \times O(N) = O(N^{1/3})$, which matches the scale of the initial momenta. This indicates that inter-particle interactions have a nontrivial effect on the motion of each particle. Over such time scales, particles travel distances of order $O(N^{-1/3})$, which---despite being small---can lead to macroscopic changes in the particle density, in particular for initial densities that vary on the same microscopic length scale (see the discussion of initial states in Section \ref{sec:del:vs:loc}). As a result, macroscopic observables, such as the spatial density, exhibit nontrivial variations on the rescaled time scale.

The above heuristics are implemented by rescaling $\tau = \eN t $ in \eqref{eq:Schroedinger-eq} and studying 
\begin{align}\label{eq:S:eq:rescaled}
i \partial_t \Phi_t = \eN H \Phi_t \qquad \text{for}\qquad  \eN = N^{-2/3} \qquad \text{and} \qquad  t=O(1).
\end{align}\vspace{-2mm}
We shall compare this with the evolution of the rescaled Hartree equations
\begin{align}\label{eq:mf-eq:rescaled}
i\partial_t \varphi_k^t = \eN  h(t)  \varphi_k^t  \qquad \text{with} \qquad h( t ) = -\Delta + v \ast \rho_t  \qquad \text{and} \qquad \rho_t = \sum_{k=1}^N |\varphi_k^t|^2
\end{align}
for $k=1,\ldots, N$.

It is important to note that even after rescaling the Hamiltonian to $\eN H$, the interaction experienced by each particle is still of order $O(N^{1/3})$, which is significantly stronger than in the typical mean-field regime, where it is of order one. In fact, this presents the central challenge in the derivation of the Hartree equations \eqref{eq:mf-eq:rescaled} from the microscopic evolution \eqref{eq:S:eq:rescaled}. To address this problem, we implement a suitable gauge transformation that removes the large potential term from the wave function, at the cost of introducing new magnetic-type interaction terms. Although these new terms are less regular, they are effectively of order one (per particle) and thus better suited for the derivation of the mean-field equations. To the best of our knowledge, this approach is novel and may also prove useful for other models of strongly interacting systems. A detailed discussion of our  strategy is provided in Section~\ref{sec:strategy:proof}.

\subsection{Main result}

\label{sec:main:results}

We state the assumptions on the pair potential and the initial data for the Hartree equations. These will be taken as our standing assumptions for the remainder of the paper.

\begin{ass}\label{ass1}
The potential $v$ is real-valued, radial and belongs to $C^2(\mathbb{R}^3)$.
\end{ass}
Under this assumption, the operator $H$ is self-adjoint with domain $D(H) = H^2(\mathbb{R}^{3N}) \cap L^2_a(\mathbb{R}^{3N})$, where $H^2(\mathbb{R}^{3N})$ denotes the second Sobolev space. Thus, the solution to the Schrödinger equation \eqref{eq:S:eq:rescaled} is given by $\Phi_t = e^{-i t \eN  H} \Phi_0$.

\begin{ass}\label{ass2}
We consider a sequence of orthonormal sets $(\varphi_1^0, \dots, \varphi_N^0) \subseteq L^2(\mathbb{R}^3)$ satisfying
\begin{align}
\max\bigg\{ \, N^{-5/3} \sum_{k=1}^N \| \nabla \varphi_k^0 \|^2 \,  , \,  N^{-7/3} \sum_{k=1}^N \| \Delta \varphi_k^0 \|^2 \,  \bigg\} \le  C
\end{align}
for some $N$-independent constant $C>0$ and all $N\ge 1$.
\end{ass}

The second assumption sets the scale of the average particle momenta (and their second moments). It is consistent with a system of $N$ fermions confined to a volume of order one, that is, a dense fermionic system. In this setting, the Pauli exclusion principle implies that the typical momentum of a particle is of order $O(N^{1/3})$. In Section~\ref{sec:del:vs:loc}, we describe two physically relevant classes of initial states that exhibit different spatial structures compatible with Assumption~\ref{ass2}. 

We note that, under the above assumption, global existence and uniqueness of solutions to \eqref{eq:mf-eq:rescaled} follow from \cite{Bove1976} or alternatively \cite[Theorem A.1]{Bach2016}. In particular $(\varphi_1^t, \ldots, \varphi_N^t) \subset H^2(\mathbb R^3)$ for all $t\in \mathbb R$; see Lemma \ref{prop:H1-H2-bd} for explicit estimates.

Next, we introduce the orthogonal projections
\begin{align}\label{eq:def:projection}
p^{\varphi_1,\ldots , \varphi_N} \coloneqq \sum_{j=1}^{N} |\varphi_{j}\rangle\langle\varphi_{j}| \quad \text{and} \quad q^{\varphi_1,\ldots , \varphi_N} \coloneqq  \id - p^{\varphi_1,\ldots , \varphi_N},
\end{align}
associated with an orthonormal set $(\varphi_1,\ldots, \varphi_N) \subset L^2(\mathbb{R}^3)$. Moreover, for a normalized $N$-particle wave function $\Phi \in L^2_{a}(\mathbb{R}^{3N})$, we define its one-particle reduced density matrix $\gamma^\Phi$ as the trace-class operator on $L^2(\mathbb{R}^3)$ with kernel
\begin{align}
\gamma^\Phi(x,y) \coloneqq \int dx_2 \dots dx_N, \overline{\Phi(x,x_2,\ldots, x_N)}, \Phi(y,x_2,\ldots, x_N).
\end{align}
In particular, $\text{Tr}(\gamma^\Phi) = 1$ and $0 \le \gamma^\Phi \le N^{-1}$, as a consequence of the normalization and antisymmetry of $\Phi$. Also note that for exact Slater determinants, $\gamma^{\varphi_1 \wedge \ldots \wedge \varphi_N} = N^{-1} p^{\varphi_1, \dots , \varphi_N}$.

We can now state our main result, which provides an approximation of bounded one-particle multiplication observables evaluated in the microscopic state $\Phi_t$ by the corresponding observables computed in the simpler Slater state $\varphi_1^t \wedge \dots \wedge \varphi_N^t$.

\begin{thm}\label{thm:main}
Let the interaction potential $v$ satisfy Assumption~\ref{ass1}, and let $\varphi_1^t, \dots, \varphi_N^t$ be solutions to the Hartree equations \eqref{eq:mf-eq:rescaled} with initial data $\varphi_1^0, \dots, \varphi_N^0$ satisfying Assumption~\ref{ass2}. Let $\Phi_t = e^{-it \eN H} \Phi_0$ denote the solution to the rescaled Schrödinger equation \eqref{eq:S:eq:rescaled} with normalized initial state $\Phi_0 \in L^2_a(\mathbb{R}^{3N})$, and let $\gamma^{\Phi_t}$ denote its one-particle reduced density. Suppose that for some $\delta_1 > \frac56 $ and $\delta_2 > \frac{1}{3}$, the initial data satisfy
\begin{align}\label{assumption:Phi:varphi}
\sup_{N \ge 1} \left( N^{\delta_1}  \textnormal{Tr}\left( \gamma^{\Phi_0} q^{\varphi_1^0, \ldots, \varphi_N^0} \right) + N^{\delta_2} \left| \textnormal{Tr}\left( \eN (-\Delta) \gamma^{\Phi_0} \right) - \frac{1}{N} \textnormal{Tr}\left( \eN (-\Delta) p^{\varphi_1^0, \ldots, \varphi_N^0} \right) \right| \right) < \infty.
\end{align}
Then there exists a constant $C > 0$ such that for all $t \ge 0$ and $N\ge 2$:
\begin{align}
& \sup_{\|M\| = 1} \left| \textnormal{Tr}\left( M \gamma^{\Phi_t} \right) - \frac{1}{N} \textnormal{Tr}\left( M p^{\varphi_1^t, \ldots, \varphi_N^t} \right) \right| \notag\\
& \hspace{3.9cm}
\le  \exp\left( Ce^{(1+t)^2} \right)  \max\left\{ N^{\frac{5}{24} - \frac{\delta_1}{4}}, N^{\frac{1}{12} -\frac{\delta_1}{8}}, N^{\frac{1}{12} - \frac{\delta_2}{4}}, N^{-\frac{1}{24}} \right\} , \label{eq:main:bound}
\end{align}
where the supremum is taken over all bounded multiplication operators $M$ on $L^2(\mathbb{R}^3)$.
\end{thm}

\begin{remarks}\phantom{a}\\[-7mm]

\begin{itemize}[leftmargin=*]
 \item [1.] Since both traces on the left-hand side of \eqref{eq:main:bound} are of order one, the estimate is meaningful for $N \to \infty$ with $t=O(1)$, where the right-hand side vanishes.

\item[2.] The quantity $ N \textnormal{Tr} ( \gamma^{\Phi_0} q^{\varphi_1^0 ,\ldots , \varphi_N^0} )  $ represents the expected number of particles in $\Phi_0$ that are not in the Slater determinant $\varphi_1^0\wedge \dots \wedge \varphi_N^0$. Thus, Assumption \eqref{assumption:Phi:varphi} implies that the number of \textit{bad particles} is initially at most of order $O(N^{1-\delta_1})$, in particular, much smaller than the total number of particles. The second term in \eqref{assumption:Phi:varphi} quantifies the deviation in kinetic energy per particle between the microscopic state and the Hartree state. Here, the assumption implies that the deviation in kinetic energy is of order $O( \eN^{-1} N^{-\delta_2})$, thus much smaller than the typical  kinetic energy per particle, which is of order $O(\eN^{-1})$. For exact Slater determinants $\Phi_0 = \varphi_1^0 \wedge \dots \wedge \varphi_N^0$, the left-hand side of \eqref{assumption:Phi:varphi} vanishes. In this case, the assumption holds for any $\delta_1, \delta_2$, and the convergence rate in our main result becomes $N^{-1/24}$.

\item[3.] Neither the convergence rate nor the restrictions $\delta_1 > \frac{5}{6}$ and $\delta_2 > \frac{1}{3}$ are expected to be optimal. In particular, we have not aimed to optimize the proof in this regard.

\item[4.] The convergence of microscopic dynamics to mean-field solutions is typically described in terms of the difference in reduced one-particle densities, measured in the trace norm topology. This topology is closely linked to comparing expectation values of bounded one-body observables. In Theorem \ref{thm:main}, we approximate only expectation values of multiplication operators. This is related to the large potential term in the Hamiltonian, which we eliminate from the wave function by using a coordinate-dependent gauge transformation.

\item[5.] 
%From a physical perspective, the most interesting case is the Coulomb potential $v(x) = |x|^{-1}$. Our current result does not cover this case. However, w
With some technical modifications and additional assumptions on the Hartree solutions, our proof can be extended to include singular potentials of the form $v(x) = |x|^{-s}$ for $s \in (0, 5/8)$. To keep the presentation as concise as possible and to avoid introducing further technical details, we do not pursue this extension here. The singular case, potentially including the Coulomb potential $v(x) = |x|^{-1}$, shall be addressed in future work.
\end{itemize}
\end{remarks}

\subsection{Initial orbitals: delocalized vs. localized}
\label{sec:del:vs:loc}

We discuss two classes of initial data for the Hartree equations, which we call \emph{delocalized} and \emph{localized} orbitals, and then explain the implications of Theorem~\ref{thm:main} for each class. The terms delocalized and localized refer to the spatial structure on length scales of order $O(N^{-1/3})$.

Delocalized orbitals arise, for example, as ground states of $N$ fermions in slowly varying (i.e., $N$-independent) external trapping potentials \cite{Fournais2018,Gottschling2018,
Cardenas2024,
Cardenas2025}, or as so-called $\pi$-electrons in large molecules \cite{clayden}. Such orbitals are extended over a macroscopic volume of order one and do not exhibit variations on smaller scales, a property that is reflected in the spatial density
\begin{align}
\rho_0 = \sum_{k=1}^N |\varphi_k^0|^2
\end{align}
satisfying $\| \nabla \rho_0 \|_1 = {O}(N)$. The orthogonality of the orbitals results from oscillations, similarly to the case of plane waves forming a free Fermi ball on the torus.

Another class of initial data compatible with Assumption~\ref{ass2} consists of orbitals that are localized on microscopic scales—for instance, $N$ orbitals with pairwise disjoint supports, all confined within a volume of order one. Physically, such orbitals may arise as ground states of non-interacting fermions in external potentials with microscopic structure, or in molecular systems where electrons are tightly bound to different (pairs of) atoms, such as core and $\sigma$-electrons \cite{clayden}. For such orbitals, the spatial density varies on microscopic length scales, which is reflected in the larger asymptotics of the density gradient: $\| \nabla \rho_0 \|_1 = {O}(N^{4/3})$.

\begin{figure}[t!]
\centering
\includegraphics[width=0.85\textwidth]{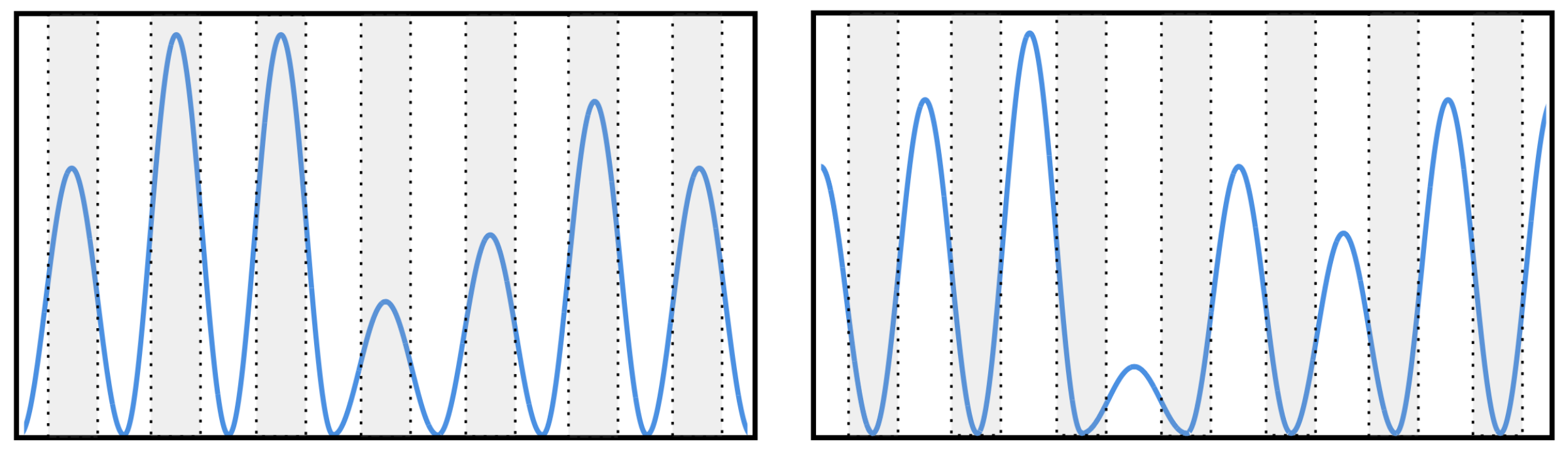} % Adjust filename as needed
\caption{\small Illustration of the mean-field evolution of the spatial density associated with $N$ localized Hartree orbitals. Initially, the orbitals have disjoint support (left). Over time they evolve over distances of order $O(N^{-1/3})$, resulting in a macroscopic change in the spatial density (right). The grey region represents the support of a multiplication observable~$M$. The illustration is conceptual and does not depict the precise evolution of the orbitals.}
\label{fig1}
\end{figure}

The qualitative behavior of the solution to the Hartree equations \eqref{eq:mf-eq:rescaled} differs substantially between the two classes of initial orbitals. To clarify this distinction, recall that on the relevant time scale, each particle effectively travels a distance of order $O(N^{-1/3})$.

In the delocalized case, since the density varies on macroscopic scales of order one, these small displacements do not significantly affect the overall density profile. Consequently, the observable density remains close to its initial configuration, and we expect that to leading order,
\begin{align} 
\operatorname{Tr}(M p^{\varphi_1^t, \ldots, \varphi_N^t}) \approx \operatorname{Tr}(M p^{\varphi_1^0, \ldots, \varphi_N^0}).
\end{align}

However, the situation is different for localized orbitals. Here, the initial density varies on the same scale as the typical particle displacement, i.e., on length scale ${O}(N^{-1/3})$. As a result, even small shifts can produce a macroscopic change in the system. The evolution of the orbitals under the Hartree dynamics therefore leads to an observable change in the density profile, and thus in the measured observable $M$. This behavior is illustrated in Figure~\ref{fig1}, which depicts the evolution of initially disjoint orbitals.

We emphasize that both cases are covered by Theorem \ref{thm:main}. The distinction affects the observable consequences of the dynamics: while the Hartree evolution may be nearly static for delocalized orbitals, it encodes genuine macroscopic quantum transport for localized ones.\medskip

\noindent\textit{Physical context and application.}
The distinction between localized and delocalized orbitals is conceptually relevant for the modeling of large molecules, where both types typically coexist. For instance, in large fullerene molecules, electrons can be categorized into core electrons, $\sigma$-bonding electrons localized between nuclei, and delocalized $\pi$-electrons spread over the molecule \cite{clayden}. Our model captures aspects of the early out-of-equilibrium dynamics for both localized and delocalized orbitals, while neglecting the presence of the nuclei for simplicity. In molecular physics, the non-equilibrium dynamics are often triggered by external perturbations such as light pulses or high-energy collisions \cite{Zewail2000,Leone2014}. While this highlights the conceptual applicability of our results, we stress that time-dependent mean-field theory is limited to the initial response and does not fully account for long-time processes like bond breaking and bond formation or correlation effects, for which more refined approaches such as time-dependent density functional theory are required \cite{Kohn1999,ullrich2012,Szabo2012}.

\subsection{Comparison with semiclassical regime}

To conclude this section, we compare our model to the well-studied semiclassical regime describing $N$ fermions in a combined weak-coupling and semiclassical limit. The corresponding Hamiltonian takes the form
\begin{align}\label{semi-classical}
H^{\rm sc} = \sum_{1 \le i \le N} (-\Delta_i) + N^{-1/3} \sum_{1 \le i < j \le N} v(x_i - x_j),
\end{align}
with relevant time scales of order $O(N^{-1/3})$.  In this model, typical paticle momenta are of order $O(N^{1/3})$ and kinetic and potential energy both scale as $O(N^{5/3})$. Similarly as in our model,  the short time scale can be justified by noting that the change in momentum per particle, estimated as \emph{time} $\times$ \emph{force} $=O(N^{1/3})$, matches the scale of the typical initial momenta. Another important observation is that by multiplying both sides of \eqref{semi-classical} by $N^{-2/3}$, the dynamics resemble those of a quantum system with effective Planck constant $\hbar = N^{-1/3}$ and interaction strength $\lambda = O(N^{-1})$, for times $t = O(N^{1/3})$. The initial mean-field data in this regime is typically
 assumed to satisfy semiclassical conditions, formulated in terms of commutator bounds:
\begin{align}\label{eq:commutator:bound}
\text{Tr}\big|[x, p^{\varphi_1^0, \dots, \varphi_N^0}]\big| = O(N^{2/3}), \qquad
\text{Tr}\big|[\nabla, p^{\varphi_1^0, \dots, \varphi_N^0}]\big| = O(N).
\end{align}
As explained, for instance, in \cite{Cardenas}, the second condition implies that the associated spatial density satisfies $\| \nabla \rho_0 \|_1 = O(N)$. Thus, the initial configurations considered in the semiclassical regime fall into the class of delocalized orbitals, as discussed in the previous section, and exclude more localized orbitals such as the ones illustrated in Figure~\ref{fig1}. For a detailed discussion of the semiclassical conditions and their interpretation, we refer to \cite{Benedikter2014, Cardenas}. Under the above assumptions, it is well known \cite{Narnhofer1981, Spohn1981} that the Wigner transform associated with the many-body quantum dynamics converges to a solution of the classical Vlasov equation, which describes the evolution of a one-body phase-space distribution. Subsequent works \cite{Erdoes2003, Benedikter2014, Jaksic2016, Porta2017} have shown that the fermionic Hartree and Hartree--Fock equations provide refined approximations by incorporating quantum corrections to the Vlasov dynamics.

\begin{table}[t!]
\renewcommand{\arraystretch}{1.35}
\centering
{\small
\begin{tabular}{c|c|c}
\textbf{} & \textbf{Semiclassical} & \textbf{Strongly interacting} \\
\hline
Hamiltonian & \eqref{semi-classical} & \eqref{eq:Hamiltonian}  \\
\hline
Time scale & $O(N^{-1/3})$ & $O(N^{-2/3})$ \\
\hline
Typical momentum & $O(N^{1/3})$ & $O(N^{1/3})$ \\
\hline
Force on each particle & $O(N^{2/3})$ & $O(N)$ \\
\hline
Change in momentum & $ O(N^{1/3})$ & $ O(N^{1/3})$ \\
\hline
Travelled distances & $O(1)$ & $O(N^{-1/3})$ \\
\hline
Structure of the spatial density & Macroscopic: $O(1)$ & Microscopic: $O(N^{-1/3})$ \\
\hline
Effective dynamics & Vlasov$+$Hartree(-Fock)& Hartree 
\end{tabular}
}
\caption{\small Comparison between the semiclassical and strongly interacting regimes. In both cases, the change in momentum, estimated heuristically as $\emph{time} \times \emph{force}$, is of the same order as the typical momenta.}
\label{table:comparison}
\end{table}

Our model differs from the semiclassical setting in several important aspects; a side-by-side comparison is provided in Table~\ref{table:comparison}. While the kinetic term in \eqref{eq:Hamiltonian} also scales as $O(N^{5/3})$, the potential term is larger by a factor of $N^{1/3}$. Accordingly, we consider shorter time scales of order $O(N^{-2/3})$. Moreover, the present model does not feature a small Planck constant, and we do not impose any semiclassical structure on the initial data. In particular, our assumptions allow for more general configurations, including localized initial orbitals characterized by a strongly varying initial density, with $\| \nabla \rho_0 \|_1 = O(N^{4/3})$, as discussed in the previous section. As a result, we expect quantum effects to contribute at leading order, and the evolution of the Wigner transform is no longer approximated by the classical Vlasov equation. We believe that our results are physically relevant, especially in regimes where electron localization and strong interactions prevent a classical approximation---for instance, in the modeling of electronic dynamics in large molecules.

% The remainder of the article is organized as follows. In the next section, we introduce the so-called counting functional, then discuss the gauge transformation, the associated gauged evolution, and an auxiliary dynamics that will play a central role in our analysis. In Section~\ref{sec:strategy:proof}, we present a detailed outline of the proof strategy, and then conclude Section~\ref{sec:Preliminaries} with an a priori estimate on the Hartree solutions. The main part of the proof is carried out in Sections~\ref{sec:auxiliary} and~\ref{sec:Norm-approximation}. Finally, in Sections~\ref{sec:toolbox:1} and~\ref{sec:toolbox:2}, we state and prove technical estimates that are used repeatedly throughout the argument.

\section{Preparations}\label{sec:Preliminaries}

\subsection{Counting functional}

In this section, we introduce the so-called counting functional. For an orthonormal set $\varphi_1,\ldots, \varphi_N \in L^2(\mathbb R^3)$, we recall Definition \eqref{eq:A:i:notation}. For a one-particle operator $A$ acting on $L^{2}(\RR^{3})$ and particle index $i\in\{1,\ldots,N\}$, we define its extension to the $N$-particle space $L^2_{a}(\mathbb{R}^{3N})$ as
\begin{align}\label{eq:A:i:notation}
A_{i}\coloneqq \id_1\otimes\ldots\otimes \id _{i-1} \otimes {A}_{i} \otimes \id_{i+1} \otimes\ldots\otimes \id_N.
\end{align} 
Next, we consider the family of orthogonal projections $P^{(N,k)}$ for $k = 0, \dots, N$, acting on $L^2_{a}(\mathbb{R}^{3N})$, defined by
\begin{align}
P^{(N,k)} \coloneqq  \left(\prod_{m=1}^{k}q_{m}^{\varphi_1,\ldots, \varphi_N} \prod_{m=k+1}^{N}p_{m}^{\varphi_1,\ldots, \varphi_N}\right)_{\text{sym}}
\end{align}
where \textit{sym} denotes the symmetrization over all $N$ variables. One straightforwardly shows the relations $P^{(N,k)}  P^{(N,\ell)}   = \delta_{k,\ell}$ and $\sum_{k=0}^{N} P^{(N,k)}  = \id $. For a detailed discussion, we refer to \cite{Petrat2014,Petrat2016}.

\begin{defn}\label{def:counting:functional}
Let $\{\varphi_j\}_{j=1}^{N} \subset L^{2}(\mathbb{R}^{3})$ be an orthonormal set and $f: \{0, \dots, N\} \to [0,1]$. We call $f$ a weight function and define the associated weight operator
\begin{align}
 \hat f^{\varphi_1,\ldots, \varphi_N} \coloneqq \sum_{k=0}^{N} f(k) P^{(N,k)} .
\end{align}
Moreover, for $\Psi \in L_a^{2}(\mathbb{R}^{3N})$, we define the counting functional
\begin{align}\label{def:counting:functional}
\alpha_{f}\left(\Psi,\varphi_{1},\ldots,\varphi_{N}\right)\coloneqq\lsp\Psi,  \hat f^{\varphi_1,\ldots, \varphi_N} \Psi\rsp = \sum_{k=0}^N f(k) \| P^{(N,k)} \Psi \|^2
\end{align}
where $\langle\cdot,\cdot\rangle$ denotes the scalar product on $L^{2}(\RR^{3N})$. 
For notational convenience, we write $\hat f \equiv \hat f^{\varphi_1, \ldots, \varphi_N}$ whenever the reference states are clear from the context. For integer $|d| \le N$, we further define the shifted operator
\begin{align}
\label{eq:shifted:weight}
\hat{f}_{d}\coloneqq\sum_{k=0}^{N} f_d(k) P^{(N,k)}  \quad \text{with} \quad f_{d}(k) \coloneqq \chi_{[0,N]}(k+d)  f(k+d) 
\end{align}
where $\chi_{S}$ denotes the indicator function of the set $S$. 
\end{defn}

In the present paper, the relevant choices of weight functions are
\begin{align}\label{eq:relevant:weight:functions}
n(k) & = \frac{k}{N}, \quad \ell(k) = \sqrt{\frac{k}{N}}, \quad  m^{(\gamma)}(k)   = \min\Big\{1,\frac{k}{N^{\gamma}}\Big\}, \quad w^{(\gamma)}(k) = 1 - m^{(\gamma)}(k)
\end{align}
for $\gamma \in (0,1]$. Note that $n = m^{(1)}$.  If $\Psi$ is normalized, then $\|P^{(N,k)} \Psi\|^2$ is the probability of finding exactly $k$ particles in $\Psi$ that are not in one of the orbitals $\varphi_1, \ldots, \varphi_N$. Accordingly, the quantity $\alpha_f(\Psi, \varphi_1, \ldots, \varphi_N)$ has the physical interpretation of an $f$-weighted expectation value of the number of such excitations. For the choice $f = n$, the functional $\alpha_n$ measures the average fraction of excitations When $f = m^{(\gamma)}$, the weight saturates once the number of excitations exceeds $N^\gamma$, thus it distinguishes between states with few and many excitations. Conversely, the choice $f = w^{(\gamma)}$ emphasizes configurations with a small number of excitations, assigning zero weight to states with more than $N^\gamma$ excitations.

We conclude this section by noting the following straightforward but important properties. A discussion of further properties of the counting functional is given in Section \ref{sec:toolbox:1}.
\begin{itemize}
\item[i)] For any weight functions $f$ and $g$, we have
\begin{align}
\widehat{fg} = \hat{f} \hat{g} \qquad \text{and} \qquad \hat f\, p_j^{\varphi_1,\ldots, \varphi_N} = p_j^{\varphi_1,\ldots, \varphi_N}  \hat f \qquad \text{for all}\qquad  j=1,\ldots, N.\label{eq:algebra:property}
\end{align}
\item[ii)] If $f(k) > 0$ for all $k \geq 1$, then $\hat f^{-1} := \widehat{f^{-1}}$ is defined on $\text{Ran}(\id - P^{(N,0)})$ and 
\begin{align}\label{eq:inverse:identity}
\hat f \hat {f}^{-1} \restriction  \text{Ran}(\id - P^{(N,0)} ) \ = \ {\id}_{\text{Ran}(\id - P^{(N,0)} )} .
\end{align}
\item[iii)] For any $\Psi \in L^2_a(\mathbb{R}^{3N})$, we obtain the sequence of identities
\begin{align}\label{eq:q:n:identity}
\alpha_{n} (\Psi, \varphi_1, \dots, \varphi_N) &
= \sum_{k=0}^{N} \frac{1}{N} \lsp \Psi, \sum_{m=1}^{N} q_{m}^{\varphi_1,\ldots, \varphi_N} P^{( N,k ) } \Psi \rsp
= \lsp \Psi, q_{1}^{\varphi_1,\ldots, \varphi_N} \Psi \rsp.
\end{align}
\item[iv)] The shifted operator arises in commutation relations, such as
\begin{align}\label{eq:shift:example}
\hat f_{-1} \left( q_1^{\varphi_1,\ldots, \varphi_N}  A_{1}^{\phantom{()}} p_1^{\varphi_1,\ldots, \varphi_N}  \right) = 
\left( q_1^{\varphi_1,\ldots, \varphi_N}  A_{1}^{\phantom{()}} p_1^{\varphi_1,\ldots, \varphi_N}  \right)  \hat f
\end{align}
for one-body operators $A_1$, and similarly for two- and three-body operators; see Lemma \ref{lem:Projection-shift}.
\end{itemize}

\subsection{Gauge transformation and auxiliary Hamiltonian}

\textit{Gauged microscopic dynamics.} A key ingredient in our analysis is the removal of the large potential term from the solution of the Schrödinger equation \eqref{eq:S:eq:rescaled}. This is achieved by introducing a suitable gauge factor. Concretely, let $\Phi_t= e^{-it \eN H}\Phi_0$ denote the solution to the Schrödinger equation with initial condition $\Phi_0$ and Hamiltonian $H$ as given in \eqref{eq:Hamiltonian}. We then consider the transformed solution
\begin{align}\label{eq:guaged:wf}
\Psi_t := \exp\bigg( i t \eN  \sum_{1\le i < j \le N} v(x_i-x_j)  \bigg) \Phi_t .
\end{align}
The gauged solution $\Psi_t$ satisfies the Schrödinger equation $i\partial_t \Psi_t = \eN H^g(t) \Psi_t $ with Hamiltonian
\begin{align}
\Hg(t) =   \sum_{i=1}^{N}\bigg(  i \nabla_{i}+t\eN  \sum_{j=1 : j\neq i }^{N}f_{ij} \bigg)^{2} \label{eq:Hg}
\end{align}
where we introduce the notation for the force field
\begin{align}\label{eq:force:field:notation} 
f_{ij} \coloneqq f (x_{i}-x_{j}) \qquad \text{with} \qquad 
  f  \coloneqq-\nabla v \ .
\end{align}
While the transformed Hamiltonian has a more complicated (magnetic-like) structure compared to the original operator $H$, the important advantage is that the new interaction terms---due to their coupling with the kinetic energy---acquire additional factors of $\eN$. As we explain in detail in Section \ref{sec:strategy:proof}, this plays a crucial role in our analysis.\medskip

\noindent \textit{Gauged Hartree dynamics}. We now introduce the corresponding gauge transformation at the level of the Hartree equations. Given $(\varphi_1^t, \ldots, \varphi_N^t)$ as the solutions of \eqref{eq:mf-eq:rescaled} with initial data $(\varphi_1^0, \ldots, \varphi_N^0)$, we define the transformed functions
\begin{align}\label{eq:guaged:Hartree:sol}
\psi_k^t := \exp \Big( i  t \eN (v\ast \rho_t) \Big) \varphi_k^t, \qquad k=1,\ldots, N,
\end{align} 
where $ \rho_t = \sum_{k=1}^N |\varphi_k^t|^2 = \sum_{k=1}^N |\psi_k^t|^2$. These satisfy the gauged Hartree equations
\begin{align}\label{eq:gauged:Hartree}
i\partial_{t}\psi_{k}^{t} = \eN \hg(t)\psi_{k}^{t},  \qquad k=1,\ldots, N,
 \end{align} with generator
\begin{align} 
\hg(t)= \Big(  i \nabla+t\eN \overline{f}\Big)^{2} +  t\eN  \Big( \overline{(i\nabla\cdot f)}+ \overline{(f\cdot i\nabla)}+2t \eN \overline{f\cdot\overline{f}} \Big). 
\label{eq:hg} 
\end{align} Here, we have introduced the notation for the mean-field force terms:
\begin{alignat}{2}\label{eq:mean:field:force:1}
\overline{f}(x) & \coloneqq f \ast\rho_{t}(x), \qquad \qquad \qquad \qquad \quad &  \overline{(f\cdot i\nabla)}(x) & \coloneqq\sum_{j=1}^{N}\lsp \psi_{j}^t,f(\cdot-x)\cdot i\nabla\psi_{j}^t\rsp,\\
\overline{f\cdot\overline{f}}(x) &
 \coloneqq\sum_{j=1}^{N}\lsp\psi_{j}^t,f(\cdot-x)\cdot\overline{f}(\cdot)\psi_{j}^t\rsp, &
 \overline{(i\nabla\cdot f)}(x) & \coloneqq\sum_{j=1}^{N}\lsp f(\cdot-x)\cdot i\nabla\psi_{j}^t,\psi_{j}^t\rsp.
 \label{eq:mean:field:force:2}
\end{alignat}
\begin{rem}
The second term in \eqref{eq:hg} follows from the identity
\begin{align}
-\partial_{t}( v\ast\rho_{t})  =  t\eN  \Big( \overline{(i\nabla\cdot f)}+ \overline{(f\cdot i\nabla)}+2t \eN \overline{f\cdot\overline{f}} \Big)
\end{align}
which is a direct consequence of the continuity equation $\partial_{t}\rho_{t}=-2 \eN \nabla\cdot \sum_{j=1}^{N}\im \{ (\varphi^t_{j})^{\ast}\nabla\varphi_{j}^t \}$.
\end{rem}
\begin{rem}
The mean-field forces are of course time-dependent, but we omit this dependence for ease of notation.
\end{rem}

\noindent \textit{Auxiliary Hamiltonian}. While the gauged operator $H^g(t)$ exhibits more favorable scaling factors than $\eN H$, its main disadvantage is that the interaction involves gradient operators, making it more difficult to control. In large parts of our analysis, we shall address this difficulty by working with an auxiliary Hamiltonian obtained via a suitable quadratic approximation of $H^g(t)$. Specifically, let $\{\psi_{k}^{t}\}_{k=1}^{N}$ be the (orthonormal) solutions of the gauged Hartree equations \eqref{eq:hg}. We then consider the (now time-dependent) orthogonal projections
\begin{align}
p = p^{\psi_{1}^{t},\ldots,\psi_{N}^{t}}\coloneqq\sum_{j=1}^{N}|\psi_{j}^{t}\rangle\langle\psi_{j}^{t}|  \qquad \text{and}\qquad 
q & \coloneqq \id -p ,
\end{align}
and use these to rewrite $H^g(t)$ as
\begin{align}\label{eq:re-writing:H:g}
\Hg(t) & = \sum_{i=1}^{N}(p_{i} +q_{i} )\bigg(  i\nabla_{i}+ t \eN \sum_{j=1:j\not=i}^{N}(p_{j}+q_j )f_{ij}(p_{j} +q_{j} )\bigg)^2 (p_i+q_i).
\end{align}
After expanding the products, we define a new auxiliary operator by discarding all terms that contain three or more $q$ operators. This operator has the advantage that we can always move a free gradient in the interaction onto a $p$ operator, which in turn can be controlled through the regularity of the (gauged) Hartree solutions. We refer to the resulting auxiliary Hamiltonian as $\widetilde{H}^g(t)$ and denote by $\widetilde{\Psi}_t$ the solution to the corresponding Schrödinger equation $i\partial_t \widetilde{\Psi}_t = \eN \widetilde{H}^g(t) \widetilde{\Psi}_t$ with $\wti \Psi_0 = \Phi_0$. More precisely, $\wti \Psi_t = \wti U(t,0) \Phi_0$ with $ \wti U(t,s) $ the two-parameter evolution generated by $\eN \widetilde H^g(t)$. For a precise definition of the latter and well-posedness of the dynamics, we refer to Section \ref{sec:aux:dynamics}.

\begin{rem} 
The idea of discarding terms with three or more $q$ operators in \eqref{eq:re-writing:H:g} is reminiscent of the Bogoliubov approximation for many-body Bose systems; see, e.g., \cite{Lewin,Napiorkowski,Petrat2019}. In the present setting, it serves as a  tool to control otherwise singular terms in the Hamiltonian $H^g(t)$.
\end{rem}

\subsection{Strategy of the proof}\label{sec:strategy:proof}

\textit{The old approach}. A way to quantify the closeness of the microscopic evolution $\Phi_t = e^{-i \eN t H} \Phi_0$ to the mean-field solutions $\{\varphi_j^t\}_{j=1}^N$ of the rescaled equations \eqref{eq:mf-eq:rescaled} is to control the counting functional
\begin{align}
\alpha_n (\Phi_t, \varphi_1^t, \ldots, \varphi_N^t ) =  \Big\langle \Phi_t, \hat n^{\varphi_1^t, \ldots, \varphi_N^t} \Phi_t \Big\rangle  = \Big\langle \Phi_t, q_1^{\varphi_1^t, \ldots, \varphi_N^t} \Phi_t \Big\rangle 
\end{align}
with weight function $n(k) = \frac{k}{N}$. The quantity $N \alpha_n$ represents the expected number of particles in $\Phi_t$ that are not in the Hartree states. Thus, if $\alpha_n$ remains small over time—compared to one—it indicates that most of the particles follow the mean-field evolution. Mathematically, the smallness of $\alpha_n$ implies that the reduced one-particle density $\gamma^{\Phi_t}$ is close to the mean-field projection $N^{-1} p^{\varphi_1^t, \ldots, \varphi_N^t}$ in trace norm distance. The usual approach to control $\alpha_n$ is to estimate its time derivative and  then apply Grönwall's inequality. This approach has been successfully implemented for weakly interacting systems \cite{Bach2016,Petrat2016}.\medskip

\noindent \textit{Why the old approach fails.}  For strongly interacting fermions described by the Hamiltonian $\eN H$, this method runs into serious difficulties. The main problem becomes apparent when differentiating $\alpha_n$. As shown in \cite{Petrat2016}, the derivative is \allowdisplaybreaks
\begin{subequations}
\begin{align}
 - i \tfrac{d}{dt} \alpha_n ( \Phi_t,\varphi_1^t,\ldots, \varphi_N^t ) & = \eN  \lsp \Phi_t, \left[ H - h_1(t), q_1 \right]   \Phi_t \rsp \label{eq:failed:Groenwall} \\[0mm]
 & = 2 \eN (N-1)  \im \langle  \Phi_t, q_1 ( p_2 v_{12} p_2 -  \tfrac{1}{N-1} v \ast \rho_t(x_1)  ) p_1  \Phi_t \rangle  \\[0mm]
 &\quad + 2  \eN  (N-1)  \im \lsp   \Phi_t, q_1 q_2  v_{12} p_1 p_2  \Phi_t \rsp  \\[0mm]
 & \quad  +  2  \eN  (N-1)  \im \lsp \Phi_t, q_1 p_2 v_{12} p_1 p_2   \Phi_t \rsp  \label{eq:failed:Groenwall:c}
\end{align}
\end{subequations}
where $v_{12} = v(x_1 - x_2)$, and $p_i = p_i^{\varphi_1^t,\ldots, \varphi_N^t}$, $q_i = \id - p_i$ denote the projections at time $t$. While each of the inner products on the right-hand side can be controlled in terms of $\| q_1 \Phi_t \|^2 = \alpha_n(\Phi_t,\varphi_1^t,\ldots, \varphi_N^t)$ plus some error of order $O(N^{-1})$, the prefactor scales like $O(N^{1/3})$. This leads to an inequality of the form $|\frac{d}{dt} \alpha_n| \le C N^{1/3} (\alpha_n+N^{-1})$, which is clearly too large to permit any useful bound on $\alpha_n$ over time scales $t = O(1)$. Moreover, there are no cancellations among the various contributions that could mitigate this problem. In summary, this approach breaks down for the present model due to the large size of the interaction term as $N \to \infty$. 

In what follows, we present an alternative strategy to overcome this obstacle.\medskip

\noindent \textit{The new approach}. As explained, the main obstacle in the old approach is the large potential contribution in $\Phi_t = e^{-i \eN t H} \Phi_0 $. A second observation is that the above strategy fails to exploit the small prefactor in front of the kinetic terms in $\eN H$ and $\eN h(t) $, as they cancel exactly in \eqref{eq:failed:Groenwall}. To address both points, we introduce the gauged wave function $\Psi_t = \exp (i t \eN \sum_{i<j} v(x_i-x_j)) \Phi_t$ as in \eqref{eq:guaged:wf}, along with \eqref{eq:guaged:Hartree:sol} for the mean-field solutions. This transformation eliminates the dominant potential term in exchange for magnetic-type interactions.  Since the latter arise from commutators with the kinetic term, they acquire additional factors of $\eN$. In the new generator, given by \eqref{eq:Hg}, the interaction terms scale as $O(N \eN^{1/2} \times \text{gradient})$ and $O(N^3 \eN^3) =O(N)$, respectively. For the present model, one should heuristically think of $\eN^{1/2} \times \text{gradient} = O(1)$, as the typical momentum of a particle is of order $O(N^{1/3})$. By comparing to the interaction in $\eN H$, which scales as $O(N^2 \eN ) = O(N^{4/3})$, this shows that we effectively gained an additional factor of $\eN^{1/2}$ by implementing the gauge transformation.
Hence, proceeding analogously to \eqref{eq:failed:Groenwall} for the gauged evolutions $\Psi_t$ and $\psi_1^t,\ldots , \psi_N^t$, we can improve the large prefactor $O(N^{1/3})$ to a prefactor of order $O(1)$. This may suggest that the old approach could become applicable if we work with the gauged dynamics instead of the original ones, i.e., for the counting functional $\alpha_n ( \Psi_t,\psi_1^t,\ldots, \psi_N^t )$.

However, this strategy gives rise to a new problem: the two-body interaction in the gauged Hamiltonian contains gradient operators of the form $(f_{12} \cdot i\nabla_1 +h.c.)$; see \eqref{eq:Hg}. In particular, when estimating the term analogous to \eqref{eq:failed:Groenwall:c} for such a two-body interaction, this would necessarily lead to a factor $\eN^{1/2} \| \nabla_1 q_1 \Psi_t \|$, which we are not able to control by a bound that is small compared to one. This term measures the kinetic energy of particles outside the mean-field states and is difficult to estimate, even with the improved scaling factors in the gauged Hamiltonian.

To avoid this problematic term, we introduce an auxiliary Hamiltonian $\widetilde{H}^g(t)$, defined as the quadratic approximation of \eqref{eq:Hg}; see Section \ref{sec:aux:dynamics} for the precise definition. The advantage of $\widetilde{H}^g(t)$ is that, due to to the presence of two or more $p$ projections in each interaction term, the problematic term with three $q$ operators and one $p$ operator does not appear when differentiating the associated counting functional. In fact, for the auxiliary evolution $\widetilde{\Psi}_t = \widetilde U(t,0) \Psi_0$, where $ \wti U(t,s)$ denotes the two-parameter evolution associated with $\eN \wti {H}^g(t)$, this allows us to close Grönwall's inequality. Concretely, we shall establish for weight function $m^{(\gamma)}$ defined in  \eqref{eq:relevant:weight:functions} with $\gamma\in (0,1]$ that for suitable initial conditions
\begin{align}\label{eq:alpha:aux}
\alpha_{m^{(\gamma)}} (\wti \Psi_t,\psi_1^t,\ldots, \psi_N^t)  = O(N^{-\gamma}).
\end{align}
The precise statement is given in Section \ref{sec:controlloing:bad:particles}. Another way to appreciate the benefit of the auxiliary Hamiltonian is that its interaction terms are more regular than those in $H^g(t)$. This improved regularity arises because gradients can always bee regularized by a $p$-projection, effectively giving a factor of order $p \times \text{gradient} = O(N^{1/3})$. This also enables us to derive a bound of the form  $\eN^{1/2} \| \nabla_1 q_1 \widetilde{\Psi}_t \| = O(N^{-1/4})$, which, while not necessary for proving \eqref{eq:alpha:aux}, becomes relevant when relating the auxiliary evolution $\widetilde{\Psi}_t$ to the full (gauged) evolution $\Psi_t$ in the next step.

Having established the necessary bounds for the auxiliary dynamics $\widetilde \Psi_t$, we turn to controlling the counting functional for the gauged dynamics $\Psi_t$. As a first step, we derive a norm approximation of the form
\begin{align} \label{eq:norm}
\| \Psi_t - \widetilde \Psi_t \| = O(N^{-1/24}).
\end{align}
The usual starting point (motivated from the analysis of bosonic mean-field systems) to estimate the norm difference between the full evolution and an approximation in terms of a quadratic generator is the identity
\begin{align} \label{eq:normsketch}
\| \Psi_t - \widetilde \Psi_t \|^2 = 2\eN  \im \int_0^t  \lsp \Psi_s, (H^g(s) - \widetilde{H}^g(s)) \widetilde \Psi_s \rsp ds
\end{align}
together with the idea of exploiting the presence of $q$-projections  to control this expression (by definition, each term in the difference of the two Hamiltonians contains at least three $q$-projections). However, for the present model, this approach fails due to two issues: the presence of gradients in the interaction and the lack of control over the $q$-projections when acting on $\Psi_t$. We explain the problematic terms in more detail at the beginning of Section~\ref{sec:Norm-approximation}. To avoid these problems, we proceed slightly differently by first estimating
\begin{align}
\| \Psi_t - \widetilde \Psi_t \| \le
\| \hat m^{(\gamma) } \widetilde \Psi_t \| +
\| \Psi_t - \hat w^{(\gamma) } \widetilde \Psi_t \|,
\end{align}
with weight functions $w^{(\gamma) } = 1 - m^{(\gamma) }$ and $m^{(\gamma) }$ defined in \eqref{eq:relevant:weight:functions}. The first term satisfies $
\| \hat m^{(\gamma) } \wti \Psi_t  \|^2 \le \langle \wti \Psi_t , \hat m^{(\gamma) } \wti \Psi_t\rangle = \alpha_{m^{(\gamma) }}( \wti \Psi_t, \psi_1^t,\ldots, \psi_N^t) = O(N^{-\gamma})$, as established in \eqref{eq:alpha:aux}. The second term is controlled via the quantity $\delta (t) := 2-2\re \langle \Psi_t, \hat w^{(\gamma) } \wti \Psi_t \rangle$ whose time derivative is
\begin{align}
\frac{d}{dt} \delta(t) = 2 \im \lsp \Psi_t, \left( \wti H^g(t) - H^g(t) \right) \hat w^{(\gamma) } \wti \Psi_t \rsp - 2 \im \lsp \Psi_t, \Bigg[ \wti H^g(t) - \sum_{j=1}^N h_j^g(t) , \hat m^{(\gamma) }  \Bigg] \wti \Psi_t \rsp.
\end{align}
The commutator term can be bounded by the square root of $\alpha_{m^{(\gamma) }}(\widetilde \Psi_t, \psi_1^t,\ldots,\psi_N^t)$. In fact, this works in close analogy to the bound for the time-derivative of the counting functional. The first term, on the other hand, resembles the right-hand side of \eqref{eq:normsketch}, but with the advantage that the Hamiltonians are now multiplied by $\hat w^{(\gamma) }$. Since $w^{(\gamma)}$ assigns zero weight to states with more than $N^\gamma$ excitations, the presence of this operator allows us to exploit the $q$-projections more effectively, in particular also when acting on $\Psi_t$, through bounds of the form $\| q_1 q_2 \hat w^{(\gamma) } \psi \| \le C N^{\gamma -1 } \| \psi \|$ for all $\psi \in L^2_a(\mathbb{R}^{3N})$; see Section \ref{sec:toolbox:1}. In this part of the argument, we need to use the previously established bound $\eN^{1/2}\| \nabla_1 q_1 \wti \Psi_t \| = O(N^{-1/4})$. Finally, choosing an appropriate $\gamma \in (0,1]$, this  enables us to balance the different error terms and establish the desired estimate \eqref{eq:norm}.

\begin{table}[t!]
\renewcommand{\arraystretch}{1.5} % Adjust this factor as needed
\centering
\begin{minipage}{0.48\textwidth}
\small
\centering
\begin{tabular}{c|c|c}
\textbf{Dynamics} & \textbf{Equation} & \textbf{Ref.} \\
\hline
Schrödinger & $i \partial_t \Phi_t = \eN  H \Phi_t $ &  \eqref{eq:Schroedinger-eq} \\
\hline
Gauged Schrödinger & $ i \partial_t \Psi_t = \eN H^g(t) \Psi_t $ & \eqref{eq:Hg} \\
\hline
Auxiliary & $ i \partial_t \widetilde{\Psi}_t =\eN  \widetilde{H}^g(t)  \widetilde{\Psi}_t $ & \eqref{eq:Hg-tilde} 
\end{tabular}
\end{minipage}
\hspace{-0.5mm}
\begin{minipage}{0.48\textwidth}
\small
\centering
\vspace{-6.6mm}
\begin{tabular}{c|c|c}
  \textbf{Dynamics} & \textbf{Equation} & \textbf{Ref.} \\
  \hline
  Hartree & $i \partial_t \varphi_j^t = \eN h(t) \varphi_j^t$ & \eqref{eq:mf-eq:rescaled} \\
  \hline
  Gauged Hartree & $i \partial_t \psi_j^t = \eN h^g(t) \psi_j^t$ & \eqref{eq:hg}
\end{tabular}
\end{minipage}
\caption{Dynamical equations used in the proof, all taken with initial time $t = 0$ and normalized initial data $\Phi_0 = \Psi_0 = \wti \Psi_0 \in L^2_a(\mathbb R^{3N})$ and $\varphi_j^0 = \psi_j^0 \in L^2(\mathbb R^3)$ for $j = 1, \ldots, N$.}
\label{tab:doubletable}
\end{table}

Using \eqref{eq:alpha:aux} and \eqref{eq:norm}, together with the Cauchy--Schwarz inequality, we can then bound the number of bad particles in the gauged solution as
\begin{align}\label{eq:alpha:gauged}
\Big\langle \Psi_t , q_1^{\psi_1^t,\ldots, \psi_N^t } \Psi_t \Big\rangle \le 4 \Big\langle \wti \Psi_t , q_1^{\psi_1^t,\ldots, \psi_N^t } \wti \Psi_t \Big \rangle  + 4 \| \Psi_t - \wti \Psi_t \|^2 = O(N^{-1/12}).
\end{align}
In the next step, we apply the following inequality for bounded one-body operators $A$, wave functions $\Psi \in L^2_{\mathrm{a}}(\mathbb{R}^{3N})$, and orthonormal $\psi_1,\ldots,\psi_N \subset L^2(\mathbb R^3)$:
\begin{align}
\left| \text{Tr}\left( A \gamma^{\Psi} \right)  - \frac1N \text{Tr} \left( A  p^{\psi_1 ,\ldots, \psi_N }  \right) \right| \le C  \| A \| \Big\langle \Psi  , q_1^{\psi_1 ,\ldots, \psi_N } \Psi \Big\rangle^{1/2}
\end{align}
for some constant $C>0$. Since multiplication operators $M$ commute with the gauge transformation inside the trace, we obtain our main result:
\begin{align}
 \text{Tr}\left( M \gamma^{\Phi_t} \right)  - \frac1N \text{Tr} \left( M  p^{\varphi_1^t ,\ldots, \varphi_N^t }  \right) = 
 \text{Tr}\left( M \gamma^{\Psi_t} \right)  - \frac1N \text{Tr} \left( M  p^{\psi_1^t ,\ldots, \psi_N^t }  \right)   = O(N^{-1/24}).
\end{align}
This concludes the strategy of our proof. 

For the reader's convenience, in Table \ref{tab:doubletable} we summarize the different dynamical equations and generators used throughout the proof.\medskip

\subsection{Estimates on the Hartree solutions}

In this section, we introduce notation related to the mean-field evolution and establish important a priori estimates for the solutions of the gauged Hartree equations.  We begin by collecting useful estimates for the various terms that appear in the mean-field Hamiltonian.

\begin{lem}
\label{lem:aux:hartree:0}
Let $\psi_1^t,\ldots, \psi_N^t$ denote the solution to the gauged Hartree equations \eqref{eq:gauged:Hartree} with initial data satisfying Assumption \ref{ass2}, and recall the definitions in \eqref{eq:mean:field:force:1} and \eqref{eq:mean:field:force:2}. There exists a constant $C>0$ such that for all $t \geq0$, we have the bounds: 
\begin{align}
\| \bar f \|_{\infty} \  + \ \|( \nabla \bar f) \|_\infty + N^{-1} \ \| \overline{ f \cdot \bar f } \|_\infty &  \le C N  , \\[3mm]
\| \overline{(i \nabla \cdot f)} \|_\infty^2 + \| \overline{(f \cdot i \nabla )} \|_\infty^2  &  \le   C N  \sum_{j=1}^N \| \nabla \psi_j^t \|^2 ,\\
\| A  p^{\psi_1^t,\ldots, \psi_N^t} \|^2 & \le C N \sum_{j=1}^N \| A \psi_j^t \|^2 
\end{align}
where $A \in \{ \nabla , \Delta \}$.
\end{lem}
\begin{proof} The first bound is straightforward; the others follow from the Cauchy--Schwarz inequality $\sum_{k=1}^N \| A \psi_k^t \| \le N^{1/2} (\sum_{k=1}^N \| A \psi_k^t \|^2 )^{1/2}$.
\end{proof}

It is convenient to rewrite the gauged mean-field Hamiltonian \eqref{eq:hg} in the form:
\begin{equation}\label{eq:hg:2}
\hg (t) =  -\Delta  + t \eN R(t) +t^{2} \eN^2 W(t) 
\end{equation}
where the operators $R(t)$ and $W(t)$ are defined as
\begin{align}
R(t) & \coloneqq i\nabla \cdot\overline{f}+\overline{f} \cdot i\nabla +\overline{(i\nabla\cdot f)}+\overline{(f\cdot i\nabla)} \label{eq:R-def}  \quad \text{and} \quad 
W(t) \coloneqq2\overline{f\cdot\overline{f}} +\overline{f}\cdot\overline{f} 
\end{align}
For future use, we define the functions
\begin{align}\label{eq:def:rho:nabla:rho:Delta}
\rho_t^\nabla \coloneqq  \sum_{k=1}^N |\nabla \psi_k^t |^2 , \quad 
\rho_t^\Delta \coloneqq  \sum_{k=1}^N |\Delta \psi_k^t |^2 , \quad  \rho_t^{R(t)} \coloneqq  \sum_{k=1}^N |R(t) \psi_k^t |^2 , \quad \rho_t^{W(t)} \coloneqq  \sum_{k=1}^N |W(t) \psi_k^t |^2.
\end{align}
To control these, we introduce the quantity
\begin{align}\label{eq:def:D(t)}
D(t) \coloneqq  \max \left\{\, N^{-5/6} \| \rho_t^\nabla \|_1^{1/2} \, ,\,   N^{-7/6}  \| \rho_t^\Delta \|_1^{1/2} \, ,\, 1 \, \right\}.
\end{align}
which will appear frequently in our subsequent proofs.
\begin{lem}\label{lem:aux:hartree} Let $\psi_1^t,\ldots, \psi_N^t$ be as in Lemma \ref{lem:aux:hartree:0}. There is a $C>0$ such that for all $t \geq 0$
\begin{align}
\eN \| \rho_t^\nabla \|_1 \le C D(t)^2 N , \quad \eN^2 \| \rho_t^\Delta \|_1 \le C D(t)^2 N ,
 \quad \eN^2 \| \rho^{R(t)}_t \|_1 + \eN^4 \| \rho_t^{W(t)} \|_1 \le CD(t)^2 \eN N^3.
\end{align}
\end{lem}
\begin{proof} The first two bounds hold by definition. To bound $\eN^2 \| \rho_t^{R(t)} \|_1$, consider for example
\begin{align*}
\eN^2  \sum_{k=1}^N  \|  \overline{ (i\nabla \cdot f ) }   \psi_k^t \|^2 \le \eN^2  \|   \overline{ (i\nabla \cdot f ) }  \|_\infty^2 N \le C \eN^2 N^2 \| \rho_t^\nabla \|_1 \le C D(t)^2 \eN  N^3
\end{align*}
by means of Lemma \ref{lem:aux:hartree:0}. The remaining terms in $\eN^2 \| \rho_t^{R(t)} \|_1$ are estimated similarly. Likewise,
\begin{align}
\eN^4 \| \rho_t^{W(t)} \|_1 = \eN^4  \sum_{k=1}^N  \|  (2 \overline{ f \cdot \bar f } +  \bar f \cdot \bar f  ) \psi_k^t \|^2 \le C \eN^4 N^4 \| \rho_t \|_1 =  C \eN^4 N^5  = C \eN N^3.
\end{align}
This completes the proof.
\end{proof}

Thanks to Assumption~\ref{ass2}, we have $D(0) \leq C$ uniformly in $N$. The following lemma (proved in the appendix) shows that $D(t)$ remains uniformly bounded for all $t \geq 0$.
\begin{lem}\label{lem:bound:D(t)} Let $\psi_1^t,\ldots, \psi_N^t$ be as in Lemma \ref{lem:aux:hartree:0}. There exists a $C>0$ such that for all $t \geq 0$
\begin{equation}
D(t)\leq e^{C(1+t)^2}  D(0).
\end{equation}
\end{lem}

\section{Auxiliary dynamics}\label{sec:auxiliary}

In this section, we introduce and analyze the auxiliary dynamics generated by the Hamiltonian obtained from \eqref{eq:re-writing:H:g} by neglecting all terms involving three or more $q$ operators.

To simplify notation throughout this section, we abbreviate the time-dependent projections as
\begin{align}\label{projections:gauged:Hartree}
p_i := p_i^{\psi_1^t,\ldots, \psi_N^t} \qquad \text{and} \qquad q_i := \id  - p_i,
\end{align}
where $\{\psi_k^t\}_{k=1}^N$ denote the solution to the gauged Hartree equations \eqref{eq:guaged:Hartree:sol} with initial data satisfying Assumption \ref{ass2}.

\subsection{Definition and well-posedness}\label{sec:aux:dynamics}
 
For $i,j,k\in \{1,2,\ldots ,N\}$ we define the following two- and three-particle projection operators. The superscript denotes the number of $q$-projections, while the subscripts indicate the set of particles on which the operator acts:
\begin{alignat}{2}
\P_{ij}^{(0)} & \coloneqq p_{i}p_{j}, \qquad \qquad \qquad & \P_{ijk}^{(0)}  & \coloneqq p_{i}p_{j}p_{k},\label{eq:def:P0}\\
\P_{ij}^{(1)} & \coloneqq p_{i}q_{j}+q_{i}p_{j},  & \P_{ijk}^{(1)} & \coloneqq p_{i}p_{j}q_{k}+p_{i}q_{j}p_{k}+q_{i}p_{j}p_{k},\\
\P_{ij}^{(2)} & \coloneqq q_{i}q_{j}, & \P_{ijk}^{(2)} & \coloneqq p_{i}q_{j}q_{k}+q_{i}p_{j}q_{k}+q_{i}q_{j}p_{k}.\label{eq:def:P2}
\end{alignat}
Next, we introduce the relevant interaction terms:
\begin{align}
(w_{\nabla f})_{ij} & \coloneqq(i\nabla_{i})\cdot f_{ij}+f_{ij}\cdot(i\nabla_{i})+(i\nabla_{j})\cdot f_{ji}+f_{ji}\cdot(i\nabla_{j}), \label{eq:w:nabla:f:potential}\\[1mm]
(w_{f})_{ij} & \coloneqq f_{ij}\cdot f_{ij}+f_{ji}\cdot f_{ji} , \\[1mm]
(w_{ff})_{ijk} & \coloneqq 2f_{ij} \cdot f_{ik}+2f_{ji}\cdot f_{jk}+2f_{ki}\cdot f_{kj},\label{eq:w:f:f:potential}
\end{align}
where we recall the notation introduced in \eqref{eq:force:field:notation}. Note that these operators are symmetric under permutations of particle indices. Then we define for $X \in \{ \nabla f ,  f\}$ the projected interaction terms:
\begin{alignat}{2}
( \wti w_X)_{ij} & \coloneqq \sum_{a=0}^2 (w_X)_{ij}^{(a)} \quad \text{with}\quad & (w_X)_{ij}^{(a)} & \coloneqq\sum_{b=0}^{a}\P^{(b)}_{ij } ( w_{X} )_ {ij}^{\phantom{()}} \P^{(a-b)}_{ij} , \label{eq:short-two-body}\\
( \wti w_{ff})_{ijk} & \coloneqq \sum_{a=0}^2 (w_{ff})_{ijk}^{(a)} \quad \text{with}\quad &
(w_{ff})_{ijk}^{(a)} & \coloneqq\sum_{b=0}^{a}\P^{(b)}_{ijk } ( w_{ff})_{ijk}^{\phantom{()}} \P^{(a-b)}_{ijk} ,\label{eq:short-three-body}
\end{alignat}
where the superscript $(a)$ indicates the number of involved $q$-projections. We emphasize that through the time-dependence of the $P$ projections, these operators are time-dependent as well.

Finally, the auxiliary Hamiltonian, consisting of all terms in \eqref{eq:re-writing:H:g} with less than three $q$-projections, can now be written in the form 
\begin{equation}
\wti{H}^g(t) \ \coloneqq \ \sum_{1\le i \le N }  (-\Delta_{i})+  \sum_{1\leq i< j \leq N} \left( t\eN   (\wti w_{\nabla  f})_{ij}  + t^2 \eN^2   ( \wti w_{f })_{ij} \right) +  \sum_{1\leq i<j<k \leq N} t^2 \eN^2 ( \wti w_{f f})_{ijk}  \label{eq:Hg-tilde}
\end{equation}

The well-posedness of the auxiliary time evolution, generated by the Hamiltonian $\eN \wti H^g(t)$ is given by the
following statement.
\begin{lem}
\label{lem:well-posedness-aux} Let $v$ satisfy Assumption \ref{ass1}. Let $\psi_1^t, \ldots, \psi_N^t$ be the solutions to the gauged Hartree equations~\eqref{eq:gauged:Hartree} with initial data satisfying Assumption~\ref{ass2}, and let $\widetilde{H}^g(t)$ be defined as in~\eqref{eq:Hg-tilde} with respect to the projections~\eqref{projections:gauged:Hartree}. Then, for every $T>0$, there exists a unique two-parameter family of operators $(\widetilde{U}(t,s))_{t,s\in [0,T]}$ acting on ${L}_a^2(\mathbb{R}^{3N})$ with the following properties:
\begin{enumerate}
\item $\wti U(t,s)D(K) \subset D(K)$ for all $t,s\in [0,T]$, where $D(K) = H^2(\mathbb R^{3N}) \cap L^2_a(\mathbb R^{3N})$. Moreover, the map $t \mapsto\wti{U}(t,s)\Psi$ is continuously differentiable for
each $s\in[0,T]$ and $\Psi\in D(K)$ with 
\begin{align}
i\partial_{t}\wti{U}(t,s)\Psi = \eN  \wti{H}^g (t)\wti{U}(t,s)\Psi.
\end{align}
\item $\wti U(t,t)=1$ and $\wti U(t,r)\wti U(r,s)=\wti U(t,s)$ for all $r,s,t\in[0,T]$.
\item $(t,s)\mapsto\wti{U}(t,s)\Psi$ is continuous for all $\Psi\in L_a^2(\mathbb R^{3N})$.
\item The solution map $\Psi_0 \mapsto \wti \Psi_t :=\wti U(t,0) \Psi_0$ is unitary.
\end{enumerate}
\end{lem}

\begin{proof} The proof of the first three items is based on \cite[Theorem 2.5]{Schmid2016}, which requires that (i) the time-dependent part $ B(t) := \eN \widetilde{H}^g(t) - \eN \sum_{i=1}^N (-\Delta_i)$ is bounded, (ii) $t \mapsto B(t)$ is strongly continuous, and (iii) $t \mapsto B(t)\Psi$ is Lipschitz for every $\Psi \in D(K)$.

Boundedness of $B(t)$ follows from Assumption~\ref{ass1} and Lemmas~\ref{lem:aux:hartree} and \ref{lem:bound:D(t)}. For instance, for normalized $\Psi \in L^2_a(\mathbb{R}^{3N})$, consider the term
\begin{align}
\left\|  \P_{ij}^{(0)} (\nabla_i \cdot f_{ij})  \P^{(2)}_{ij} \Psi \right\| \le C \sup\left\{ |f| , | \nabla f| \right\} (1+ \| \nabla p   \|)  \le C_N(t) < \infty.
\end{align}
Since each gradient in $B(t)$ can be regularized via a $p$-projection in this way, and since gradients can be commuted with $f_{ij}$ without generating unbounded functions (by Assumption~\ref{ass1}), we find $\| B(t)\Psi \| \leq C_N(t)$, showing that $B(t)$ is indeed bounded.

Properties (ii) and (iii) follow from the stronger observation that $t \mapsto B(t)$ is Lipschitz continuous (in operator norm), which in turn follows from Lipschitz continuity of $t \mapsto p = p^{\psi_1^t, \ldots, \psi_N^t}$. The latter can be established using Duhamel’s formula:
\begin{align}
\psi_k^t - \psi_k^s = \left( e^{-i (t-s) \varepsilon (-\Delta) } - 1 \right) \psi_k^s
- i \eN \int_s^t \mathrm{d}r \left( h^g(r) - ( -\Delta) \right) \psi_k^r, 
\end{align}
where the first term is bounded by $ |t-s| \eN  \|\Delta \psi_k^s\| \leq |t-s| C_{T,N}$ for all  $t,s \in [0,T]$ by Lemma~\ref{lem:aux:hartree:0}. For the second term, the integrand is bounded by $\| \nabla \psi_k^r \| \leq C_{T,N}$ for all $r\in [0,t]$, so overall $\| \psi_k^t -\psi_k^s \| \leq C_{T,N} |t-s|$, proving the Lipschitz continuity of $t \mapsto p$. From this it is straightforward to show that $t \mapsto B(t)$ is also Lipschitz continuous. This establishes properties $1.$, $2.$ and $3.$ The last statement is a straightforward computation using self-adjointness of $\wti H^g(t)$.
\end{proof}

For later use, we express the operators \eqref{eq:R-def} as partial traces of the operators \eqref{eq:w:nabla:f:potential} and \eqref{eq:w:f:f:potential}:
\begin{align}
\label{eq:trace:formulas}
R_1(t) = \tr_2 \left( p_2 (w_{\nabla f})_{12} p_2 \right), \qquad W_1(t) = \frac12 \tr_{2,3} \left( p_2 p_3 (w_{ff})_{123} p_2 p_3 \right).
\end{align}
Here and throughout the paper, we denote by
\begin{align}
\tr_{2}(A_{12}) = \sum_{n\in \mathbb N} \langle u_n, A_{12} u_n \rangle_{2}
\end{align}
the partial trace over the second variable of a two-particle operator $A_{12}$, where $(u_n)_{n \in \mathbb N}$ denotes an orthonormal basis of $L^2(\mathbb R^3)$. The subscript in $\langle \cdot , \cdot \rangle_2$ indicates integration over the variable $x_2$, so that $\langle u_n, A_{12} u_n \rangle_2$ defines an operator acting on $x_1$ alone. An analogous convention is used for $\tr_{2,3}(A_{123})$, which denotes the partial trace over both $x_2$ and $x_3$.

\subsection{Number of bad particles}
\label{sec:controlloing:bad:particles}

We provide a bound for the number of particles in the auxiliary dynamics (see Lemma \ref{lem:well-posedness-aux}) 
\begin{align}
\wti\Psi_t \coloneqq \wti U(t,0)\Psi_0
\end{align}
that are not in the Hartree orbitals. We refer to these as \textit{bad particles} or \emph{excitations}. To this end, we study the time-dependent counting functional 
\begin{align}\label{eq:aux:counting:functional}
{\alpha}_{f}(t)\coloneqq\alpha_{f} \left( \wti \Psi_{t},\psi_{1}^{t},\ldots,\psi_{N}^{t} \right) = \lsp\wti{\Psi}_{t},\widehat{ f}^{\psi_1^t,\ldots,\psi_N^t} \wti{\Psi}_{t}\rsp
\end{align}
introduced in \eqref{def:counting:functional}. We focus on the weight function $f = m^{(\gamma) }$ with $\gamma\in (0,1]$, see \eqref{eq:relevant:weight:functions}. The main result of this section is summarized in the following proposition, whose proof is given at the end of the section. 
\begin{prop}
\label{prop:aux-GW} Let $v$ satisfy Assumption \ref{ass1}. Let $\psi_1^t, \ldots, \psi_N^t$ be the solutions to the gauged Hartree equations~\eqref{eq:gauged:Hartree} with initial data satisfying Assumption~\ref{ass2}, and let $\wti \Psi_t = \wti U(t,0) \Psi_0$ for some normalized initial state $\Psi_0 \in L^2(\mathbb R^{3N})$. Furthermore, let $\alpha_{m^{(\gamma) }}(t)$ with $\gamma \in (0,1]$ be defined as in \eqref{eq:aux:counting:functional}. Then there exists a constant $C > 0$ such that for all $t \geq 0$,
\begin{align*}
\alpha_{m^{(\gamma) }}(t) \leq \exp\left( e^{C(1+t)^2} \right) \left( \alpha_{m^{(\gamma) }}(0) + N^{-\gamma} \right).
\end{align*}
\end{prop}

The next lemma provides an explicit formula for the time derivative of $\alpha_f(t)$. The statement holds for general weight functions $f$. For notational convenience, we abbreviate $\hat f \equiv  \widehat{ f}^{\psi_1^t,\ldots,\psi_N^t} $ and recall the definition of the shifted weight operator $\hat f_d$ in \eqref{eq:shifted:weight}.
\allowdisplaybreaks
\begin{lem} 
\label{lem:Derivative-of-alpha_f} Under the assumptions as in Proposition \ref{prop:aux-GW}, we have for general weight function $f$:
\begin{align}
i \frac{d}{dt}{\alpha}_{f}(t)=(\textnormal{Ia}) + (\textnormal{Ib}) +   (\textnormal{Ic})  + (\textnormal{IIa}) + (\textnormal{IIb}) + (\textnormal{IIc}) 
\end{align}
with
\begin{align*}
(\textnormal{Ia}) & \coloneqq  t \eN^2 N   \im\lsp\wti{\Psi}_{t},\left(\hat{f}-\hat{f}_{-1}\right)\left((N-1) \P^{(1)}_{12} ( w_{\nabla f} )_{12}^{\phantom{()}} \P^{(0)}_{12} -2 q_1 R_{1}(t)  p_{1}\right)\wti{\Psi}_{t}\rsp,\\[1mm]
(\textnormal{Ib}) & \coloneqq t^2 \eN^3 N  (N-1) \im\lsp\wti{\Psi}_{t}, \left(\hat{f}-\hat{f}_{-1}\right) \left( \P^{(1)}_{12} ( w_{ f} )_{12}^{\phantom{()}} \P^{(0)}_{12} \right) \wti{\Psi}_{t}\rsp,\\[1mm]
(\textnormal{Ic}) & \coloneqq \tfrac13  t^2 \eN^3 N \im\lsp\wti{\Psi}_{t},\left(\hat{f}-\hat{f}_{-1}\right)\left((N-1)(N-2) \P^{(1)}_{123}(w_{ff})_{123}^{\phantom{()}} \P^{(0)}_{123}-6  q_{1}W_{1}(t) p_{1}\right)\wti{\Psi}_{t}\rsp.\\
(\textnormal{IIa}) & \coloneqq  t\eN^2 N(N-1) \im\lsp\wti{\Psi}_{t},\left(\hat{f}-\hat{f}_{-2}\right)\left( \P^{(2)}_{12} (w_{\nabla f})_{12}^{\phantom{()}} \P^{(0)}_{12}\right) \wti{\Psi}_{t}\rsp,\\[1mm]
(\textnormal{IIb}) & \coloneqq  t^2\eN^3 N (N-1)  \im\lsp\wti{\Psi}_{t},\left(\hat{f}-\hat{f}_{-2}\right) \left( \P^{(2)}_{12} (w_{ f})_{12}^{\phantom{()}} \P^{(0)}_{12}\right) \wti{\Psi}_{t}\rsp,\\[1mm]
(\textnormal{IIc}) & \coloneqq \tfrac13 t^2 \eN^3
N(N-1) (N-2) \im\lsp\wti{\Psi}_{t},N\left(\hat{f}-\hat{f}_{-2}\right)\left( \mathbf P^{(2)}_{123} (w_{ff})_{123}^{\phantom{()}} \P^{(0)}_{123}\right)\wti{\Psi}_{t}\rsp .
\end{align*}
\end{lem} 
\begin{proof} The identity follows from a straightforward computation. We outline the main steps. First note that by Lemma \ref{lem:well-posedness-aux} and Equation \eqref{eq:gauged:Hartree}, we have
\begin{align}\label{eq:e:o:m:counting:functional}
i\partial_{t}\wti{\Psi}_{t} & = \eN \wti{H}^g(t)\wti{\Psi}_{t} \qquad \text{and}\qquad 
i\partial_{t} \hat f  = \sum_{j=1}^N \eN \left[   \hg_{j}(t) , \hat f  \right] .
\end{align}
Using the antisymmetry of $\wti{\Psi}_{t}$, and that the kinetic terms  cancel each other, we obtain  \allowdisplaybreaks
\begin{align}
-i \frac{d}{dt} {\alpha}_{f}(t)   & =\eN \lsp\wti{\Psi}_{t},\left[\wti{H}^g(t)-\sum_{j=1}^{N}\hg_{j}(t),\hat{f}\right]\wti{\Psi}_{t}\rsp\label{eq:time-derivative-alpha} \notag \\
 & = \eN \lsp\wti{\Psi}_{t},\left[ \sum_{1\leq i<j\leq N} \left( t\eN ( \wti w_{\nabla f})_{ij} + t^2 \eN^2  ( \wti w_{\nabla f} )_{ij} \right)  -\sum_{j=1}^{N}  t \eN R_{j}(t) ,\hat{f}\right]\wti{\Psi}_{t}\rsp\nonumber \\[1mm]
 & \quad + \eN \lsp\wti{\Psi}_{t},\left[ \sum_{1\leq i<j<k\leq N} t^2 \eN^2 ( \wti w_{ff})_{ijk} -\sum_{j=1}^N  t^{2}\eN^2W_{j} (t) ,\hat{f}\right]\wti{\Psi}_{t}\rsp\nonumber \\[1mm]
 & =  t \eN^2 N \lsp \wti{\Psi}_{t} , \left[ \frac{(N-1)}{2} ( \wti w_{\nabla f} )_{12}  - R_{1}(t) ,\hat{f} \right] \wti{\Psi}_{t} \rsp +  t^2 \eN^3 \frac{N(N-1)}{2} \lsp \wti{\Psi}_{t} , \left[  ( \wti w_{ f} )_{12}   ,\hat{f} \right] \wti{\Psi}_{t} \rsp  \notag\\[2mm]
 &\quad +  t^2 \eN^3  N \lsp \wti{\Psi}_{t} , \left[ \frac{(N-1)(N-2)}{6} ( \wti w_{f f} )_{123}  - W_{1}(t) ,\hat{f} \right] \wti{\Psi}_{t} \rsp \notag\\[3mm]
 & =: \frac12 A(t) + \frac16 B(t).
 \end{align}
From here, we follow essentially the same strategy as in \cite[Lemma 6.5]{Petrat2016}, that is, we insert identities $1 = p_1 + q_1$ around $R_1(t)$ and $W_1(t)$. Taking into account the shift Lemma \ref{lem:Projection-shift} (see also \eqref{eq:shift:example}) this leads to
\begin{align}
-A(t)  &= t \eN^2 N   \lsp\wti{\Psi}_{t},\left(\hat{f}-\hat{f}_{-1}\right)\left((N-1) \P^{(1)}_{12} ( w_{\nabla f } )_{12}^{\phantom{()}} \P^{(0)}_{12}-2  q_{1} R_{1}(t) p_{1}\right)\wti{\Psi}_{t}\rsp-\text{c.c.}\nonumber \\
 &\quad + t^2 \eN^3  N(N-1)   \lsp\wti{\Psi}_{t},\left(\hat{f}-\hat{f}_{-1}\right)\left( \P^{(1)}_{12} ( w_{ f } )_{12}^{\phantom{()}} \P^{(0)}_{12} \right) \wti{\Psi}_{t}\rsp-\text{c.c.}\nonumber \\
 & \quad+ t \eN^2 N(N-1) \lsp\wti{\Psi}_{t},\left(\hat{f}-\hat{f}_{-2}\right)\left( \P^{(2)}_{12}(w_{\nabla f})_{12}^{\phantom{()}} \P^{(0)}_{12}\right) \wti{\Psi}_{t}\rsp-\text{c.c.} \nonumber\\
& \quad  +  t^2 \eN^3 N(N-1) \lsp\wti{\Psi}_{t},\left(\hat{f}-\hat{f}_{-2}\right)\left( \P^{(2)}_{12}(w_{f})_{12}^{\phantom{()}} \P^{(0)}_{12}\right)\wti{\Psi}_{t}\rsp-\text{c.c.} \label{eq:A_12}
\end{align}
where $c.c.$ denotes the complex conjugate of the preceding term. Similarly, we obtain
\begin{align}
 -B(t) & = t^2 \eN^3  N \lsp\wti{\Psi}_{t},\left(\hat{f}-\hat{f}_{-1}\right)\left((N-1)(N-2) \P^{(1)}_{123} (w_{ff})_{123}^{\phantom{()}} \P^{(0)}_{123} - 6  q_1 W_1(t) p_1  \right)\wti{\Psi}_{t}\rsp-\text{c.c.}\nonumber \\
 & \quad+ t^2 \eN^3 N(N-1)(N-2) \lsp\wti{\Psi}_{t},\left(\hat{f}-\hat{f}_{-2}\right)\left( \P^{(2)}_{123} ( w_{ff} )_{123}^{\phantom{()}} \P^{(0)}_{123}\right)\wti{\Psi}_{t}\rsp-\text{c.c.}\label{eq:B_123}
\end{align}
Combining the above identities proves the claimed identity.
\end{proof}

Next, we estimate the time-derivative $\partial_t \alpha_f(t)$ for the specific weight function $f = m^{(\gamma)}$.

\begin{lem}
\label{lem:Estimates-derivative-of-alpha_m}
Under the same assumptions as in Proposition \ref{prop:aux-GW}, there exists a constant $C > 0$ such that for all $t \in \mathbb{R}$
\begin{align}
\left| \frac{d}{dt} \alpha_{m^{(\gamma)}}(t) \right| \leq C (1 + t)^2 D(t) \left( \alpha_{m^{(\gamma)}}(t) + N^{-\gamma} \right),
\end{align}
where $D(t)$ is defined in \eqref{eq:def:D(t)}.
\end{lem}

\begin{proof} By Lemma \ref{lem:Derivative-of-alpha_f}, we need to estimate $(\rm {Ia}), \ldots, (\rm {IIc})$ for $\hat f = \hat m^{(\gamma)}$. The first step is to factorize the difference of the weight functions and shift one of the square-roots to the other side of the two- and three-particle operators . To this end, let us note that for $d=1,2,3$,
\begin{align}
\hat m^{(\gamma)} - \hat m_{-d}^{(\gamma)} \restriction \text{Ran}\left( \sum_{k=d}^{N} P^{(N,k)} \right)  = \sum_{k=d}^N \left( m^{(\gamma)}(k)  - \chi_{[0,N]} (k-d)  m^{(\gamma)}(k-d) \right) P^{(N,k)}, \\
\hat m^{(\gamma)}_d - \hat m_{}^{(\gamma)} \restriction \text{Ran}\left( \sum_{k=0}^{N-d} P^{(N,k)} \right)  = \sum_{k=0}^{N-k} \left( \chi_{[0,N]} (k+d)  m^{(\gamma)}(k+d)  - m^{(\gamma)}(k) \right) P^{(N,k)}
\end{align}
are positive self-adjoint operators, so that we can define
\begin{align}
\hat{D}^{(\gamma)}_{-d} &  \coloneqq  \left( 
\hat m^{(\gamma)} - \hat m_{-d}^{(\gamma)} \right)^{1/2} \quad \text{on} \qquad \text{Ran}\left( \sum_{k=d}^{N} P^{(N,k)} \right), \label{eq:D:gamma:def}\\
\hat{E}_{-d}^{(\gamma)} & \coloneqq \left( m^{(\gamma)}_d - m^{(\gamma)}\right)^{1/2} \quad \text{on} \qquad  \text{Ran}\left( \sum_{k=0}^{N-d} P^{(N,k)} \right).\label{eq:E:gamma:def}
\end{align}
As shown in Lemma \ref{lem:Weight-estimate-m}, these operators satisfy the relation
\begin{align}\label{eq:commutator:m:diff}
\left( \hat{m}^{(\gamma)} - \hat{m}^{(\gamma)}_{-d}\right)  \left(  \P^{(d)}_{\mathcal C} A_\mathcal{C} ^{\phantom{()}} \P^{(0)}_{\mathcal{C}} \right)  =
\hat{D}_{-d}^{(\gamma)} \left(  \P^{(d)}_\mathcal{C} A^{\phantom{()}}_{\mathcal C} \P^{(0)}_{\mathcal{C}} \right) \hat{E}_{-d}^{(\gamma)}
\end{align}
where $d=\{1,2\}$, $\mathcal C \in \{  \{12\} , \{123\}\}$ or $d=1$, $\mathcal C=\{1\}$ with $\P_{1}^{(0)} := p_1$ and $\P_1^{(1)} := q_1$. Here, $A_{\mathcal C}$ is an operator acting only on particle indices $\mathcal C$.

After rewriting the operators to be estimated with the aid of \eqref{eq:commutator:m:diff}, we apply the technical bounds from Section \ref{sec:toolbox}. In what follows, we suppress the $\gamma$-dependence and write
\begin{align}
\hat m^{(\gamma)} \equiv \hat m,\qquad \hat m^{(\gamma)}_{-d} \equiv \hat m_{-d},\qquad \hat D_{-d}^{(\gamma)}\equiv \hat D_{-d} ,\qquad \hat E_{-d}^{(\gamma)} \equiv \hat E_{-d} .
\end{align}
\noindent \underline{$1q$-terms}: Using the explained strategy, together with Lemmas \ref{lem:Weight-estimate-m} and \ref{lem:1q-estimate} and $\eN = N^{-\frac23}$, we estimate
\begin{align}
| ( {\rm Ib}) | & =  \left|t^2 \eN^3 N(N-1) \im\lsp\wti{\Psi}_{t},\left(\hat{m}-\hat{m}_{-1}\right)\left( \P^{(1)}_{12}  (w_{f})_{12}^{\phantom{()}} \P^{(0)}_{12} \right) \wti{\Psi}_{t}\rsp\right|\nonumber \\[0.5mm]
 & \leq t^{2}  \left|\lsp\hat{D}_{-1}\wti{\Psi}_{t},\left( \P^{(1)}_{12} (w_{f})_{12}^{\phantom{()}} \P^{(0)}_{12} \right) \hat{E}_{-1}\wti{\Psi}_{t}\rsp \right| \nonumber \\[0.5mm]
  & \leq C t^{2}  \|q_{1}\hat{D}_{-1}\wti{\Psi}_{t}\|\|\hat{E}_{-1}\wti{\Psi}_{t}\|\nonumber \\[0.5mm]
 & \leq Ct^{2}   N^{-\frac12-\frac{\gamma}{2} } \sqrt{ {\alpha}_{m^{(\gamma) }}(t)}.
\end{align}
In more detail, we estimate the inner product in the second line using \eqref{eq:P0:P1:bound:example}, applied with $\psi = \hat D_{-1} \widetilde{\Psi}_t$ and $\varphi = \hat E_{-1} \widetilde{\Psi}_t$. In the last step (and repeatedly in what follows), we use the weight estimates 
\begin{align}\label{eq:weight:estimates:rep}
\| q_1 \hat D_{-d} \widetilde{\Psi}_t \|^2 \leq C N^{-1} \alpha_{m^{(\gamma)}}(t) \qquad \text{and} \qquad \ \| \hat E_{-d} \widetilde{\Psi}_t \|^2 \leq C N^{-\gamma}
\end{align} 
for $d \in \{1,2\}$; see Lemma~\ref{lem:Weight-estimate-m}.

We continue with ${\rm{(Ia)}}$. Here, we use that $1 =\hat \ell \hat{\ell}^{-1}$ on $\text{Ran}(\id - P^{(N,0)} )$---see Eq. \eqref{eq:inverse:identity}---with weight weight function $\ell(k) = \sqrt{k/N}$, to write
\begin{align}
\P^{(1)}_{12}\wti{\Psi}_{t} = \hat{\ell}\, \hat{\ell}^{-1} \P^{(1)}_{12}\wti{\Psi}_{t} \quad \text{and} \quad q_1 \wti{\Psi}_{t} = \hat{\ell} \, \hat{\ell}^{-1} q_1  \wti{\Psi}_{t}.
\end{align}
We then move $\hat{\ell}$ to the right side of the inner product using Lemma~\ref{lem:Projection-shift}; see also \eqref{eq:shift:example}. This step is essential to balance the number of $q$ projections on both sides of the inner product, which is necessary to close the Grönwall estimate. We obtain
\begin{align}
| (\rm {Ia}) | &= t \eN^2 N  \left|\im\lsp\wti{\Psi}_{t},\left(\hat{m}-\hat{m}_{-1}\right)\left((N-1)\P^{(1)}_{12}(w_{\nabla f})_{12}^{\phantom{()}} \P^{(0)}_{12}- 2 q _{1} R_{1}(t) p_{1}\right)\wti{\Psi}_{t}\rsp \right| \nonumber \\
 &= t \eN^2 N  \left|\im\lsp  \wti{\Psi}_{t},\hat D_{-1}\left((N-1) \P^{(1)}_{12}(w_{\nabla f})_{12}^{\phantom{()}} \P^{(0)}_{12}- 2 q_{1} R_{1}(t)p_{1}\right) \hat E_{-1} \wti{\Psi}_{t}\rsp \right| \nonumber \\
 & = t N^{-\frac13}\left| \im \lsp \hat{\ell}^{-1}\hat{D}_{-1}\wti{\Psi}_{t},\left((N-1) \P^{(1)}_{12}( w_{\nabla f})_{12}^{\phantom{()}} \P^{(0)}_{12}-2 q_{1}R_{1}(t) p_{1}\right) \hat{\ell}_{+1} \hat{E}_{-1}\wti{\Psi}_{t}\rsp \right|\nonumber .
\end{align}
Next, we recall the trace formula \eqref{eq:trace:formulas} in order to apply Lemma \ref{lem:0q-diag-estimate} with $r^{(i)}_1 = 2q_1$, $\psi = \hat{\ell}^{-1}\hat{D}_{-1}\wti{\Psi}_{t}$ and $\varphi = \hat{\ell}_{+1}  q_{1}\hat{E}_{-1}\wti{\Psi}_{t} $ to obtain
\begin{align}
|{  (\rm {Ia})} | & \leq C t N^{-\frac13}  ( N^{\frac12} + \| \rho_t^\nabla\|^{\frac12}_1)   \,  \|\hat{\ell}_{+1}\hat{E}_{-1}\ti{\Psi}_{t}\| \left( N \| \hat{\ell}^{-1} q_{2}q_{1}\hat{D}_{-1}\ti{\Psi}_{t}\| +\| \hat{\ell}^{-1} q_{1}\hat{D}_{-1}\ti{\Psi}_{t}\| \right) .
\end{align}
with $\rho^\nabla_t$ defined in \eqref{eq:def:rho:nabla:rho:Delta}. We use that $( \hat \ell_{+1})^2 =\hat  n_{+1}$, so that for every $\psi \in L_a^2(\mathbb R^{3N})$,
\begin{align}\label{eq:ell:to:n}
\| \hat \ell_{+1} \psi \|^2 \le \langle \psi , \hat n \psi \rangle + N^{-1} \| \psi \|^2 =  \langle \psi , q_1  \psi \rangle + N^{-1} \| \psi \|^2 ,
\end{align}
as well as $\| \hat {\ell}^{-1} q_1 \psi \| \le C \| \psi \|$ by Lemma \ref{lem:l-inverse conversion}. After invoking the weight 
estimates \eqref{eq:weight:estimates:rep} we find that
\begin{align}
| (\rm {Ia}) |  & \leq  C t N^{-\frac56} ( N^{\frac12} + \| \rho_t^\nabla\|^{\frac12}_1)     \left( \alpha_{m^{(\gamma)}}(t) + N^{-\gamma} \right) \notag\\
 & \leq C t  ( N^{-\frac13} + N^{-\frac56} \| \rho_t^\nabla\|^{\frac12}_1)    \left( \alpha_{m^{(\gamma)}}(t) + N^{-\gamma} \right) .
\end{align}
Using Lemma \ref{lem:0q-diag-estimate} for $s_1^{(1)} = 3 q_1$, $\psi = \hat{\ell}^{-1} \hat{D}_{-1}\wti{\Psi}_{t} $ and $\varphi = \hat{\ell}_{+1}\hat{E}_{-1}\wti{\Psi}_{t}$ we can proceed similarly for the three-body term, \allowdisplaybreaks
\begin{align}
 | ({\rm Ic })| 
&= t^2  \eN^3 N \left|\im\lsp\wti{\Psi}_{t},\big(\hat{m}-\hat{m}_{-1}\big)\left((N-1)(N-2) \P^{(1)}_{123} (w_{ff})_{123}^{\phantom{()}} \P^{(0)}_{123} -6 q_{1}W_{1}(t) p_{1}\right)\wti{\Psi}_{t}\rsp\right|\nonumber \\
& =  t^{2} N^{-1} \left| \im \lsp\hat{\ell}^{-1} \hat{D}_{-1}\wti{\Psi}_{t},\left((N-1)(N-2) \P^{(1)}_{123}(w_{ff})_{123}^{\phantom{()}} \P^{(0)}_{123}-6 q_{1}W_{1}(t)p_{1}\right) \hat{\ell}_{+1} \hat{E}_{-1}\wti{\Psi}_{t}\rsp\right |\nonumber \\
& \le C N^{\frac23} \| \hat{\ell}_{+1}\hat{E}_{-1}\wti{\Psi}_{t} \| \left( N \| q_1 q_2  \hat{\ell}^{-1} \hat{D}_{-1}\wti{\Psi}_{t} \|^2 + \|  q_1 \hat{\ell}^{-1} \hat{D}_{-1}\wti{\Psi}_{t} \|^2 \right)^{1/2} \notag\\
 & \leq Ct^{2} \left( \alpha_{m^{(\gamma)}}(t) +N^{-\gamma} \right)\ .
\end{align}
\uline{2q-terms}: Here, we apply essentially the same strategy, now using \prettyref{lem:2q-estimate-asym}. This gives
\begin{align}
| ( {\rm {IIb}}) |  & = \left| t^2 \eN^3 N (N-1) \im\lsp\wti{\Psi}_{t},\left(\hat{m}-\hat{m}_{-2}\right) \left( \P^{(2)}_{12}  (w_{f})_{12}^{\phantom{()}}  \P^{(0)}_{12} \right) \wti{\Psi}_{t}\rsp \right|\nonumber \\
 & \leq t^{2}  \left|\lsp \hat{D}_{-2}\wti{\Psi}_{t}, \left( \P^{(2)}_{12} (w_{f})_{12}^{\phantom{()}} \P^{(0)}_{12} \right) \hat{E}_{-2}\wti{\Psi}_{t}\rsp \right|\nonumber \\
 & \leq C t^{2}N^{-\frac12} \sqrt{{\alpha}_{m^{(\gamma)}}(t)}  (  \alpha_{m^{(\gamma)}}(t) + N^{-\gamma}  )^{1/2},
\end{align}
as well as
\begin{align}
| ({\rm {IIa }}) | & = \left| t \eN^2 N(N-1) \im \lsp\wti{\Psi}_{t}, \left(\hat{m}-\hat{m}_{-2}\right) \left( \P^{(2)}_{12} (w_{\nabla f})_{12}^{\phantom{()}}  \P^{(0)}_{12} \right) \wti{\Psi}_{t}\rsp\right|\nonumber \\[1mm]
 & \leq t N^{\frac23} \left|\lsp\hat{D}_{-2}\wti{\Psi}_{t}, \left( \P^{(2)}_{12} (w_{\nabla f})_{12}^{\phantom{()}} \P^{(0)}_{12} \right) \hat{E}_{-2}\wti{\Psi}_{t}\rsp \right|\nonumber \\
 & \leq Ct  N^{\frac23}  N^{-1}  ( 1 + N^{-\frac12} \| \rho_t^\nabla\|^{\frac12}_1)    \sqrt{ \alpha_{m^{(\gamma)}}(t)  } ( \alpha_{m_{(\gamma)}}(t) + N^{-\gamma} )^\frac12 
\end{align}
and 
\begin{align}
 | ({\rm {IIc}} )| & = \left| t^2 \eN^3 N(N-1) (N-2) \im\lsp\wti{\Psi}_{t}, \left(\hat{m}-\hat{m}_{-2}\right) \left( \P^{(2)}_{123} (w_{ff})_{123}^{\phantom{()}} \P^{(0)}_{123} \right) \wti{\Psi}_{t}\rsp\right|\nonumber \\
 & \leq t^{2} N \left|\lsp \hat{D}_{-2}\wti{\Psi}_{t}, \left( \P^{(2)}_{123} (w_{ff})_{123}^{\phantom{()}} \P^{(0)}_{123} \right) \hat{E}_{-2}\wti{\Psi}_{t}\rsp\right| \nonumber \\
 & \leq Ct^{2} \sqrt{ \alpha_{m^{(\gamma)}}(t)}  (   \alpha_{m^{(\gamma)}}(t) +   N^{-\gamma} )^{1/2} .
\end{align}
Adding everything together proves the lemma.
\end{proof}

\begin{proof}[Proof of Proposition \ref{prop:aux-GW}]
By Lemma~\ref{lem:Estimates-derivative-of-alpha_m} and Grönwall’s inequality, we have
\begin{align}
\alpha_{m^{(\gamma)}}(t) \leq C \exp\left( \int_0^t (1 + s)^2 D(s)\, ds \right) \left( \alpha_{m^{(\gamma)}}(0) + N^{-\gamma} \right)
\end{align}
with $D(s)$  defined in \eqref{eq:def:D(t)}. The bound now follows from Lemma~\ref{lem:bound:D(t)} and Assumption~\ref{ass2}.
\end{proof}

\subsection{Kinetic energy of bad particles}
\label{sec:bad:kinetic:energy}

Next, we show that the kinetic energy of particles outside the Hartree orbitals in $\wti \Psi_t$ remains small, for suitable initial conditions, compared to the typical kinetic energy per particle, which is of order $O(\eN^{-1})$. This is made precise in the following proposition, whose proof is postponed to the end of the section. Note that we continue with the convention that $p$ and $q$ denote the time-dependent operators defined in \eqref{eq:Hg-tilde}.

\begin{prop}
\label{prop:kinetic-energy-estimate} Let $v$ satisfy Assumption \ref{ass1}. Let $\psi_1^t, \ldots, \psi_N^t$ be the solutions to the gauged Hartree equations~\eqref{eq:gauged:Hartree} with initial data satisfying Assumption~\ref{ass2}, and let $\wti \Psi_t = \wti U(t,0) \Psi_0$ for some normalized initial state $\Psi_0 \in L^2(\mathbb R^{3N})$.
There exists a constant $C > 0$ such that for all $t \geq 0$,
\begin{align*}
\eN \| \nabla_1 q_1 \widetilde{\Psi}_t \|^2 \leq \exp\left( e^{C(1 + t)^2} \right) \left( \left| \eN \| \nabla_1 \Psi_0 \|^2 - \eN  \frac{1}{N}  \sum_{k=1}^N \| \nabla \varphi_k^0 \|^2 \right| + \alpha_n(0)^{1/2} + N^{-1/2} \right),
\end{align*}
where $\alpha_n(0)$ is defined in \eqref{eq:aux:counting:functional} with $f = n = m^{(1)}$.
\end{prop}
We introduce the gauged Hartree energy functional
\begin{equation}
E^{g}(t) \coloneqq  \text{Tr}\left( p \widetilde{h}(t) p \right) \label{eq:Eg} \qquad \text{with} \qquad \widetilde{h}(t) \coloneqq  -\Delta + \frac{1}{2} t \eN  R(t) + \frac{1}{3} t^2\eN^2 W(t)
\end{equation}
where $R(t)$ and $W(t)$ were defined in \eqref{eq:R-def}. 

In the first step we estimate the kinetic energy of the bad particles in terms of the energy difference of the microscopic system and the Hartree system. To this end, we set
\begin{align}\label{def:beta}
\beta(t)\coloneqq \eN N^{-1} \left( \lsp\wti{\Psi}_{t} ,  \wti{H}^g(t) \wti{\Psi}_{t}\rsp - E^g(t)  \right) .
\end{align}
\begin{lem}
\label{lem:Bad-kinetic-energy-decomp}
Under the assumptions as in Proposition \ref{prop:kinetic-energy-estimate} there is a $C>0$ such that for all $t\geq0$
\begin{align}
\eN \|\nabla_1 q_1\wti{\Psi}_{t}\|^{2}\leq\beta(t)+ \exp\left( e^{C(1+t)^2}  \right) \left(\alpha_{n}(0)^{1/2}+N^{-\frac{1}{2}}\right).
\end{align}
\end{lem}
\begin{proof}
We rewrite the auxiliary Hamiltonian \eqref{eq:Hg-tilde} as  
\begin{equation}
\wti{H}^g(t) = \sum_{i=1}^{N} (-\Delta_{i})+W^{(0)}+W^{(1)}+W^{(2)},
\end{equation}
where $W^{(n)}$ consists of the terms with exactly $n$ $q$-operators, that is
\begin{align}
W^{(n)} & := \sum_{1\leq i< j \leq N} \left( t\eN   (\wti w_{\nabla  f})_{ij}^{(n)}  + t^2 \eN^2   ( \wti w_{f })_{ij}^{(n)}   \right) +  \sum_{1\leq i<j<k \leq N} t^2 \eN^2 ( \wti w_{f f})_{ijk} ^{(n)} 
\end{align}
for $n=0,1,2$. With this we can write
\begin{align}
 \sum_{i=1}^{N}q_{i}(-\Delta_{i})q_{i} & = \sum_{i=1}^{N}\left\{ (-\Delta_{i})-p_{i}(-\Delta_{i})p_{i}-p_{i}(-\Delta_{i})q_{i}-q_{i}(-\Delta_{i})p_{i}\right\} \nonumber \\
 & =\wti{H}^g(t) -\left(  W^{(0)}+W^{(1)}+W^{(2)}\right) - \sum_{i=1}^{N}\left\{ p_{i}(-\Delta_{i})p_{i}+p_{i}(-\Delta_{i})q_{i}+q_{i}(-\Delta_{i})p_{i}\right\} 
\end{align}
and thus:
\begin{equation}\label{eq:bad:kinetic:proof}
\eN \|\nabla_1q_1\wti{\Psi}_{t}\|^{2} = \eN N^{-1} \sum_{i=1}^{N}\| \nabla_{i}q_{i}\wti{\Psi}_{t}\|^{2}=\beta(t)+D^{(0)} + D^{(1)}  + D^{(2)}
\end{equation}
with 
\begin{align}
D^{(0)}  & = \eN \lsp\wti{\Psi}_{t},\left(  \tfrac1N \tr \left( p \wti{h} (t)  p \right) - p_1( - \Delta_{1})p_{1}-\tfrac1N W^{(0)}\right)\wti{\Psi}_{t}\rsp,\\[1mm]
D^{(1)}  & = -  \eN   \lsp\wti{\Psi}_{t},\left( p_{1}(-\Delta_{1})q_{1}+q_{1}(-\Delta_{1})p_{1}  +  \tfrac1N W^{(1)}\right)\wti{\Psi}_{t}\rsp, \label{eq:W:1:definition}\\[1mm]
D^{(2)} & = -  \eN  \lsp\wti{\Psi}_{t}, \tfrac1N W^{(2)}\wti{\Psi}_{t}\rsp.
\end{align}
The next step is to show that $D^{(0)} $, $D^{(1)} $ and $D^{(2)}$ can be
bounded in terms of the counting functional and suitable errors.\smallskip

\noindent \underline{$D^{(0)} $-term}: To bound this term, we add and subtract $p_1\wti{h}_1(t) p_1$ such
that
\begin{align}
D^{(0)}  =  \eN \lsp\wti{\Psi}_{t},\left(  \tfrac1N \tr \left( p \wti{h} (t)  p \right) - p_1\wti{h}_1(t) p_1 \right) \wti \Psi_t \rsp + \eN \lsp \wti \Psi_t,  \left(  p_1 \left( \wti{h}_1(t) + \Delta_{1}\right) p_{1}-\tfrac1N W^{(0)}\right)\wti{\Psi}_{t}\rsp.
\end{align}
For the first summand, we use Lemma \ref{lem:hg-0q-estimate} for $h = \wti h(t)$ such that
\begin{align}
 \left|\eN \lsp\wti{\Psi}_{t},\left(p_{1}\wti{ h}_{1}(t) p_{1}-\tfrac1N \text{Tr}\left(p \wti{h}  (t) p \right)\right)\wti{\Psi}_{t}\rsp\right| & \leq  \eN N^{-1/2} \| \rho^{\wti h(t)}_t \|_1^{1/2} \left( \alpha_n(t) + N^{-1} \right)^{1/2} \notag\\
 & \le C (1+t)^2 D(t)  \left( \alpha_n(t) + N^{-1} \right)^{1/2}
\end{align} 
with $\rho^{\wti h(t)}_t= \sum_{k=1}^N |\wti h(t) \varphi_k^t|^2$, where we further invoked the estimate
\begin{align}\label{eq:1:norm:h:tilde}
\tfrac12 \| \rho_t^{\wti h(t)} \|_1 \le \| \rho_t^\Delta \|_1 + t^2 \eN^2 \| \rho_t^{R(t)} \|_1 + t^4 \eN^4 \| \rho_t^{W(t)} \|_1  \le C(1+t)^4 D(t)^2 \eN^{-2} N
\end{align}
by means of Lemma \ref{lem:aux:hartree}. We also recall that $\| q_1 \wti \Psi_t\|^2 = \langle \wti \Psi_t , \hat n \wti \Psi_t \rangle = \alpha_n(t)$; see \eqref{eq:q:n:identity}.

For the second teerm in $D_{0}$, we need to use cancellations between the microscopic Hamiltonian and the mean-field terms:\allowdisplaybreaks
\begin{align}
 & \left| \eN \lsp\wti{\Psi}_{t}, \left( \tfrac1N W^{(0)}-p_{1} \left(\wti{h}_{1}(t) + \Delta_{1}\right)p_{1} \right) \wti{\Psi}_{t} \rsp\right|\nonumber \\[1mm]
 & \leq \frac{1}{2} t \eN^{2} \left|\lsp\wti{\Psi}_{t},\left((N-1) \P^{(0)}_{12} (w_{\nabla f})_{12}^{\phantom{()}} \P^{(0)}_{12} -p_{1}R_{1}(t)p_{1}\right)\wti{\Psi}_{t}\rsp \right| \nonumber \\
 & \quad+\frac{1}{2}t^2 \eN^{3}\left|\lsp\wti{\Psi}_{t},(N-1) \left( \P^{(0)}_{12}(w_{f})_{12}^{\phantom{()}} \P^{(0)}_{12}\right) \wti{\Psi}_{t}\rsp\right|\nonumber \\
 & \quad+ \frac{1}{6}t^2\eN^{3}\left|\lsp\wti{\Psi}_{t},\left((N-1)(N-2) \P^{(0)}_{123}(w_{ff})_{123}^{\phantom{()}} \P^{(0)}_{123}-2 p_{1}W_{1}(t)p_{1}\right)\wti{\Psi}_{t}\rsp\right|\nonumber \\
 & \leq C(1+t)^{2}D(t) \left(\alpha_{n}(t) +N^{-1}\right)^{1/2}
\end{align}
where we employed Lemma \ref{lem:0q-diag-estimate} with $r^{(0)}_1 =2  p_1$ for the first line and with $s_1^{(0)} = 3 q_1 $ for the third line, and Lemma \ref{lem:1q-estimate} for the second line. \smallskip

\noindent \underline{$D^{(1)} ${-term}}: In contrast to the proof of \prettyref{lem:Estimates-derivative-of-alpha_m} where
we used cancellations between microscopic terms and mean-field terms for the 1-$q$ contributions, here we estimate
these terms directly by ${\alpha}_{n}(t)^{1/2}$. To this end, we estimate using Lemma \ref{lem:1q-estimate},
\begin{align}
 \eN^2  (N-1)\left| \lsp\wti{\Psi}_{t}, \left( \P^{(0)}_{12} (w_{\nabla f})_{12}^{\phantom{()}} \P^{(1)}_{12} \right) \wti{\Psi}_{t}\rsp \right| &  \leq C D(t)\alpha_{n}(t)^{1/2},\\[1mm]
  \eN^3 (N-1)\left| \lsp\wti{\Psi}_{t},\left( \P^{(0)}_{12} (w_{f})_{12}^{\phantom{()}} \P^{(1)}_{12} \right) \wti{\Psi}_{t}\rsp  \right|  
 & \leq C N^{-1/3} \alpha_{n}(t)^{1/2} \\[1mm]
  \eN^3 \frac{(N-1)(N-2)}{3}\left| \lsp\wti{\Psi}_{t},\left( \P^{(0)}_{122} (w_{ff})_{123}^{\phantom{()}} \P^{(1)}_{123}\right) \wti{\Psi}_{t}\rsp \right| & \le C \alpha_{n}(t)^{1/2}.
\end{align}
Moreover, the kinetic terms in \eqref{eq:W:1:definition} are estimated by Cauchy-Schwarz via
\begin{align}
 \left|\eN \lsp \wti{\Psi}_{t}, \left\{ p_{1}(-\Delta_{1})q_{1}+q_{1}(-\Delta_{1})p_{1}\right\} \wti{\Psi}_{t}\rsp  \right| & \le 2 \eN \left| \lsp \wti{\Psi}_{t},q_{1}(-\Delta_{1})p_{1}\wti{\Psi}_{t} \rsp \right|\nonumber \\
 & \leq 2  \alpha_{n}(t)^{1/2} \eN \|\Delta_{1}p_{1}\wti{\Psi}_{t}\| \leq 2 D(t) \alpha_{n}(t)^{1/2}
\end{align}
where we used the first statement of \prettyref{lem:diagonalization-estimate} for $A=\Delta$, in order to estimate 
\begin{equation}
\eN^2 \|\Delta_{1}p_{1}\wti{\Psi}_{t}\|^{2}\leq  \eN^{2} N^{-1} \| \rho^{\Delta}_t \|_1 = N^{-7/3} \| \rho^\Delta_t \|_1 \leq D(t)^2 .
\end{equation}
Adding everything up, this gives
\begin{align}
|D^{(1)}|  \le C (1+t)^2 D(t) \alpha_n(t)^{1/2}.
\end{align}
\noindent \underline{$D^{(2)} $-term}: We decompose 
\begin{subequations}
\begin{align}
& \frac1N \lsp\wti{\Psi}_{t},W^{(2)}\tilde{\Psi}_{t}\rsp \notag\\
& =\frac{(N-1)}{2}\left( t\eN \lsp\wti{\Psi}_{t} , \left( \P^{(1)}_{12} (w_{\nabla f})_{12}^{\phantom{()}} \P^{(1)}_{12} \right) \wti{\Psi}_{t}\rsp + t^2 \eN^2  \lsp\wti{\Psi}_{t}, \left( \P^{(1)}_{12} (w_{f})_{12}^{\phantom{()}} \P^{(1)}_{12} \right) \wti{\Psi}_{t}\rsp \right) \label{eq:D2-a}\\
 & \quad+ (N-1)\left( \re\lsp\wti{\Psi}_{t},\left( \P^{(0)}_{12} ( w_{\nabla f})_{12}^{\phantom{()}} \P^{(2)}_{12} \right) \wti{\Psi}_{t}\rsp + t^2 \eN^2 \re\lsp\wti{\Psi}_{t}, \left( \P^{(0)}_{12}( w_{f})_{12}^{\phantom{()}} \P^{(2)}_{12}\right) \wti{\Psi}_{t}\rsp \right)  \label{eq:D2-b}\\
 & \quad+\frac{ (N-1)(N-2)}{6} t^2 \eN^2 \lsp\wti{\Psi}_{t},\left( \P^{(1)}_{123}(w_{ff})_{123}^{\phantom{()}} \P^{(1)}_{123}\right) \wti{\Psi}_{t}\rsp\label{eq:D2-c}\\
 & \quad+\frac{ (N-1)(N-2)}{3} t^2 \eN^2 \re\lsp\wti{\Psi}_{t}, \left( \P^{(0)}_{123}(w_{ff})_{123}^{\phantom{()}} \P^{(2)}_{123}\right) \wti{\Psi}_{t}\rsp\label{eq:D2-d} .
\end{align}
\end{subequations}
We then estimate using Lemmas \ref{lem:2q-estimate-sym} and \ref{lem:2q-estimate-asym}
for the symmetric 2$q$-terms
\begin{align}
\eN |\eqref{eq:D2-a}| & \leq C (1+t)^2 D(t) \alpha_{n}(t) +Ct \alpha_{n}(t)^{1/2}\|\eN^{1/2} \nabla_{1}q_{1}\wti{\Psi}_{t}\| ,\\
\eN |\eqref{eq:D2-c}| & \leq Ct^{2} \alpha_{n}(t)
\end{align}
and for the antisymmetric 2$q$-terms
\begin{align}
\eN |\eqref{eq:D2-b}| & \leq C(1+t)^2 D(t) \alpha_{n}(t)^{1/2}\left(\alpha_{n}(t)+\eN^2 \right)^{1/2}  , \\
\eN |\eqref{eq:D2-d}| & \leq Ct^{2}\alpha_{n}(t)^{1/2}\left(\alpha_{n}(t)+N^{-1}\right)^{1/2}.
\end{align}
Hence, we obtain
\begin{align}
|D^{(2)}| & \le C (1+t)^2 D(t) \left( \alpha_n(t) + N^{-1} \right) +  \frac12 \eN \| \nabla_1 q_1 \wti \Psi_t \|^2 .
\end{align}

Using \eqref{eq:bad:kinetic:proof} and adding all estimates together, we arrive at
\begin{align}
\eN \|\nabla_1 q_1\wti{\Psi}_{t}\|^{2}\leq\beta(t)+C(1+t)^{4}D(t)^{2} \left(\alpha_{n}(t)^{1/2}+N^{-\frac{1}{2}}\right).
\end{align}
The claimed bound now follows from Lemma \ref{lem:bound:D(t)} and Assumption \ref{ass2}, together with Proposition \ref{prop:aux-GW} for $\gamma=1$. This completes the proof of the lemma.
\end{proof}

Since the Hamiltonians $\wti H^g(t)$ and $\wti{h}(t)$ are time-dependent, the energy difference $\beta(t)$ is not conserved in time. In the next lemma, we estimate its time-derivative.
\begin{lem}
\label{lem:beta-deriv} Under the assumptions as in Proposition \ref{prop:kinetic-energy-estimate} there is a $C>0$ such that for all $t\geq0$  
\begin{align}
|\partial_{t}\beta(t)| \leq e^{C(1+t)^2}  \beta(t)+  \exp\left( e^{C(1+t)^2}  \right)\left( \alpha_n(0)^{1/2} +N^{-\frac{1}{2}} \right).
\end{align}
\end{lem}
\begin{proof}
Recalling Lemma \ref{lem:well-posedness-aux} and the definition of the time-dependent projections \eqref{projections:gauged:Hartree}, we have the equations of motion
\begin{align}\label{eq:eq:motion}
i\partial_{t}\wti{\Psi}_{t} & = \eN \wti{H}^g(t)\wti{\Psi}_{t},\qquad \text{and}\qquad 
i\partial_{t}p   = \eN  [ \hg (t),p ] .
\end{align}
For convenience, we also recall the definitions of the Hamiltonians\allowdisplaybreaks
\begin{align}
\hg (t) & =  - \Delta +t\eN R (t) +t^{2}\eN^{2}W (t) = \wti{h}(t) +\frac{1}{2} t \eN R (t)+\frac{2}{3}t^{2}\eN^{2}W (t), \label{eq:recall:hg}
\\[2mm]
\wti{H}^g(t) & = \sum_{1\le i \le N }  (-\Delta_{i})+  \sum_{1\leq i< j \leq N} \left( t\eN   (\wti w_{\nabla  f})_{ij}  + t^2 \eN^2   ( \wti w_{f })_{ij} \right) +  \sum_{1\leq i<j<k \leq N} t^2 \eN^2 ( \wti w_{f f})_{ijk} . \label{eq:recall:Hg-tilde}
\end{align}
Note that the operators
$(\wti w_{\nabla f})_{ij}$, $(\wti w_{f})_{ij}$ and $(\wti w_{ff})_{ij}$, as defined in \eqref{eq:short-two-body} and \eqref{eq:short-three-body}, are time-dependent  through the time-dependence of the projections $\P_{i j}^{(a)}$, $\P^{(a)}_{ijk}$

Next, we calculate 
\begin{align}
 i\partial_{t}\beta(t) & =N^{-1} \eN \left(   \lsp\wti{\Psi}_{t},\left(  i\partial_{t}\wti{H}^g(t) \right)\wti{\Psi}_{t}\rsp+   \text{Tr} \left(p \left(\eN \left[  \hg (t) ,\wti{h} (t)\right]-(i\partial_{t}\wti{h} (t))\right) p  \right) \right) \label{eq:beta-t-derivative}
\end{align}
and note that there are three sources of $t$-dependence in $\wti{H}^g(t)$ and $\wti{h}(t)$:
\begin{enumerate}
\item[(a)] The derivatives of $R (t)$ and $W (t)$ in $\partial_t \wti h(t)$.  In combination with the commutator
term $[\hg (t),\wti{h} (t)]$, these add up to zero.
\item[(b)] The derivative of the $t$-dependent prefactors in \eqref{eq:recall:hg} and \eqref{eq:recall:Hg-tilde}. These terms can be treated analogously to the proof of \prettyref{lem:Bad-kinetic-energy-decomp}.
\item[(c)]  The derivatives of $(\wti w_{\nabla f})_{ij}$, $(\wti w_{f})_{ij}$ and $(\wti w_{ff})_{ijk}$ in \eqref{eq:recall:Hg-tilde}. 
\end{enumerate}
More explicitly,  to estimate \eqref{eq:beta-t-derivative}
we decompose
\begin{align}\label{eq:derivative:identity}
 &  i\partial_{t}\wti{H}^g(t)  + \eN  \tr \left(p \left[ \hg (t) ,\wti{h} (t)\right]p \right)  - \tr \left(p (i\partial_{t}\wti{h} (t) )p \right)  = (\textnormal{a})+(\textnormal{b})+(\textnormal{c})
\end{align}
with
\begin{align}
(\textnormal{a}) & \coloneqq \eN \tr \left(p \left[\hg (t),\wti{h}(t) \right] p \right) - \tr \left(p \left(\frac{1}{2}t \eN (i\partial_{t}R(t))+\frac{1}{3}t^2 \eN^{2}( i\partial_{t}W(t) ) \right)p \right),\\
(\textnormal{b}) & \coloneqq  i\eN \left( \sum_{1\leq i<j \leq N}   (w_{\nabla f})_{ij}^{(0)} - \frac{1}{2} \tr \left(p R(t) p \right) \right) + i 2 t  \eN^{2} \sum_{1\leq i<j \leq N} (w_{f})_{ij}^{(0)} \nonumber \\
 & \quad+i2t\eN^{2} \left(  \sum_{1\leq i<j<k\leq N}(w_{ff})_{ijk}^{(0)} - \frac{1}{3} \tr \left( p W(t) p \right) \right) \nonumber \\
 & \quad+ \sum_{1\leq i<j\leq N}\left( i \eN (w_{\nabla f})_{ij}^{(1)}+ i 2t\eN^{2}(w_{f})_{ij}^{(1)}\right)+i2t\eN^{2}  \sum_{1\leq i < j < k \leq N}(w_{ff})_{ijk}^{(1)}\nonumber \\
 & \quad+ \sum_{1\leq i < j \leq N}\left( i \eN (w_{\nabla f})_{ij}^{(2)}+ i 2t\eN^{2}(w_{f})_{ij}^{(2)}\right)+i2t\eN^{2}  \sum_{1\leq i < j < k \leq N}(w_{ff})_{ijk}^{(2)}\ ,\\
%(\textnormal{c}) & \coloneqq  \sum_{1\leq i < %j \leq N}\left(t\eN  (w_{\nabla f})_{ij}^{(2,%%\partial)}+t^2\eN^{2}(w_{f})_{ij}^{(2,\partial)}%\right)+t^{2}\eN^{2}  \sum_{1\leq i < j < k \leq %N}(w_{ff})_{ijk}^{(2,\partial)}\ ,\\
(\textnormal{c}) & \coloneqq  \sum_{1\leq i < j \leq N}\left(t\eN  i\partial_t (\wti  w_{\nabla f})_{ij} +t^2\eN^{2} i\partial_t (\wti  w_{f})_{ij} \right)+t^{2}\eN^{2}  \sum_{1\leq i < j < k \leq N} i \partial_t (\wti  w_{ff})_{ijk}\ .
\end{align}
Next, we  estimate each part of the time derivative separately.\medskip

\noindent \underline{New Term (a)}: Using the trace formulas \eqref{eq:trace:formulas}, we compute
\begin{align}
i \partial_t R_1(t) = i\partial_t \tr_2\left(p_2 (w_{\nabla f })_{12} p_2 \right) = - \eN \tr_2 \left( p_2 [h^g_2(t) , (w_{\nabla f })_{12} ] p_2 \right) 
\end{align}
and thus, using $(w_{\nabla f })_{12}  = (w_{\nabla f })_{21} $, we get
\begin{align*}
\tr \left( p ( i \partial_t R(t)) p \right) = 
 \tr_1 \left( p_1 ( i \partial_t R_1(t) ) p_1 \right) 
& = - \eN  \tr_{1,2}\left( p_1 p_2 [h^g_2 (t) , (w_{\nabla f })_{12} ] p_1 p_2 \right) \notag\\[1mm]
& = - \eN \tr_{2}\left( p_2 [h^g_2(t) , R_2(t)] p_2\right) = - \eN \tr \left( p [h^g(t),R(t) ] p\right).
\end{align*}
Similarly,
\begin{align}
i \partial_t W_1(t) = \frac12  i\partial_t \tr_{2,3}\left(p_2p_3 (w_{f f })_{123} p_2 p_3  \right) = - \frac12  \eN \tr_{2,3} \left( p_2 p_3 [h^g_2(t) + h^g_3(t) , (w_{ f f })_{123} ] p_2 p_3  \right) .
\end{align}
Using symmetry of $(w_{f f })_{123}$ under exchange of indices, we obtain
\begin{align*}
\tr \left( p ( i \partial_t W(t)) p \right) = 
 \tr_1 \left( p_1 ( i \partial_t W_1(t)) p_1 \right) 
& = -  \eN  \tr_{1,2,3}\left( p_1 p_2 p_3 [h^g_2(t)    , (w_{f f })_{123} ] p_1 p_2 p_3 \right) \notag\\[1mm]
& = - 2 \eN \tr_{2}\left( p_2 [h^g_2(t) , W_2(t)] p_2\right) = - 2 \eN \tr \left( p [h^g(t), W (t) ] p\right).
\end{align*}
Hence, we arrive at
\begin{align}
(\textnormal{a}) & =  \eN \tr \left(p \left[\hg (t),\wti{h}(t) + \frac12 t \eN R(t) + \frac23 t^2 \eN^2 W(t) \right] p \right)  = \eN \tr \left(p \left[\hg (t) , \hg(t) \right] p \right) = 0,
\end{align}
where we used \eqref{eq:recall:hg}.\medskip

\noindent \underline{Term (b)}: These terms are estimated analogously to the $D^{(0)}$-term, $D^{(1)}$-term
and $D^{(2)}$-term from the proof of \prettyref{lem:Bad-kinetic-energy-decomp}. The result is
\begin{equation}
N^{-1} \eN \left|  \lsp \wti{\Psi}_{t},(\textnormal{b})\wti{\Psi}_{t}\rsp \right| \leq C(1+t)^{4}D(t)^{2} \left( \eN \| \nabla_1 q_1 \wti \Psi_t \|^2 +\alpha_{n}(t)^{1/2}+N^{-\frac{1}{2}}\right).
\end{equation}
\noindent \underline{Term (c)}: Recalling the definition  of the projections in \eqref{eq:def:P0}--\eqref{eq:def:P2} as well as \eqref{eq:short-two-body} and \eqref{eq:short-three-body}, we observe that the time derivatives  $\partial_t (\wti  w_{\nabla f})_{ij}$,  $ \partial_t (\wti  w_{f})_{ij} $ and $ \partial_t (\wti  w_{ff})_{ijk}$ are  non-zero only  if they act on a $p$-operator in $(w_{\nabla f})^{(2)}_{ij}$, $(w_{f})^{(2)}_{ij}$ and $(w_{ff})^{(2)}_{ijk}$. All other derivatives cancel out, as for instance in 
\begin{align}
(\partial_{t}q_{i})q_{j} (w_{\nabla f})^{(2)}_{ij} p_{i}p_{j} + (\partial_{t}p_{i})q_{j} (w_{\nabla f})_{ij}^{(1)} p_{i}p_{j} =0.
\end{align}
Analogous to this example, each derivative of a $q$-operator in $(w_{\nabla f})_{ij}^{(a)}$ is cancelled by a derivative of a $p$-operator in $(w_{\nabla f})_{ij}^{(a-1)}$. Similarly, for $(w_{\nabla f})_{ij}^{(a)}$ and $(w_{\nabla ff})_{ijk}^{(a)}$.\smallskip

Thus, we obtain for $X \in \{ \nabla f , f\}$:
\begin{align}
i\partial_t (\wti w_{X})_{12}  & = \left(i\partial_{t}\P^{(0)}_{12} \right)(w_X)_{12}^{\phantom{()}} \, \P^{(2)}_{12}+ \P^{(2)}_{12}\, (w_X)_{12}^{\phantom{()}} \left(i\partial_{t}\P^{(0)}_{12} \right)\nonumber \\
 & \quad+\left((i\partial_{t}p_{1})q_{2}+q_{1}(i\partial_{t}p_{2})\right)(w_X)_{12}^{\phantom{()}} \P^{(1)}_{12} + \P^{(1)}_{12}\, (w_X)_{12}^{\phantom{()}} \left((i\partial_{t}p_{1})q_{2}+q_{1}(i\partial_{t}p_{2})\right),\\[2mm]
i\partial_t (\wti w_{ff})_{123} &  = \left(i\partial_{t} \P^{(0)}_{123} \right)(w_{ff})_{123}^{\phantom{()}} \, \P^{(2)}_{123}+\P^{(2)}_{123} \, (w_{ff})_{123}^{\phantom{()}} \left(i\partial_{t} \P^{(0)}_{123 }\right)\nonumber \\
 & \quad +\left(q_{3}(i\partial_{t}p_{1}p_{2})+q_{1}(i\partial_{t}p_{2}p_{3})+q_{2}(i\partial_{t}p_{1}p_{3})\right)(w_{ff})_{123}^{\phantom{()}} \, \P^{(1)}_{123 }\nonumber \\
 & \quad + \P^{(1)}_{123}\, (w_{ff})_{123}^{\phantom{()}}\left(q_{3}(i\partial_{t}p_{1}p_{2})+q_{1}(i\partial_{t}p_{2}p_{3})+q_{2}(i\partial_{t}p_{1}p_{3})\right).
\end{align}
Using \eqref{eq:eq:motion} and $(w_X)_{12} =(w_X)_{21} $, we can simplify the different terms to
\begin{align}
  \lsp\wti{\Psi}_{t},\left(i \partial_{t} \P^{(0)}_{12}\right) (w_X)_{12}^{\phantom{()}} \, \P^{(2)}_{12}\wti{\Psi}_{t} \rsp & =2\eN \lsp\wti{\Psi}_{t},\left(\hg_{1}(t)p_{1}-p_{1}\hg_{1}(t)\right)p_{2}(w_X)_{12} \, \P^{(2)}_{12}\wti{\Psi}_{t}\rsp,\notag\\[1mm]
 \lsp\wti{\Psi}_{t},\left((i\partial_{t}p_{1})q_{2}+q_{1}(i\partial_{t}p_{2})\right)(w_X)_{12}^{\phantom{()}} \, \P^{(1)}_{12} \wti{\Psi}_{t}\rsp 
 & =2\eN \lsp\wti{\Psi}_{t},\left(\hg_{1}(t)p_{1}-p_{1}\hg_{1}(t)\right)q_{2}(w_X)_{12}^{\phantom{()}}\,   \P^{(1)}_{12}\wti{\Psi}_{t}\rsp,
\end{align}
 and, similarly, using $(w_{ff})_{123}=(w_{ff})_{213}=(w_{ff})_{321}$
\begin{align}
 \lsp\wti{\Psi}_{t},\left(i \partial_{t} \P^{(0)}_{123} \right)(w_{ff})_{123}^{\phantom{()}} \P^{(2)}_{123 } \wti{\Psi}_{t}\rsp 
 & =3\eN \lsp\wti{\Psi}_{t},\left(\hg_{1}(t)p_{1}-p_{1}\hg_{1}(t)\right)p_{2}p_{3}(w_{ff})_{123} \P^{(2)}_{123 }\wti{\Psi}_{t}\rsp , \\[1mm]
 & \hspace{-6cm} \lsp \wti{\Psi}_{t}, \left(q_{3}(i \partial_{t}p_{1}p_{2})+q_{1}(i \partial_{t}p_{2}p_{3})+q_{2}( i \partial_{t}p_{1}p_{3})\right)  (w_{ff})_{123} \P^{(1)}_{123}\wti{\Psi}_{t}\rsp  \notag\\
 & \hspace{-1.5cm}=  6\eN \lsp \wti{\Psi}_{t},\left(\hg_{1}(t)p_{1}-p_{1}\hg_{1}(t)\right) p_{2}q_{3}(w_{ff})_{123}\P^{(1)}_{123} \wti{\Psi}_{t}\rsp.
\end{align}
Thus, after multiplying with the factor $\eN$ from \eqref{eq:beta-t-derivative}, it holds that
\begin{subequations}
\begin{align}
& \eN \frac1N  \lsp\wti{\Psi}_{t},(\textnormal{c})\wti{\Psi}_{t}\rsp  =  \frac{N-1}{2}\lsp \wti{\Psi}_{t}, \left( t \eN^3 ( i \partial_t  (w_{\nabla f})_{12}) + t^2 \eN^4 ( i \partial_t (w_{f})_{12}) ) \right) \wti{\Psi}_{t}\rsp \notag\\
& \quad + \frac{(N-1)(N-2)}{6} t^2 \eN^4 \lsp \wti \Psi_t, (i\partial_t (w_{ff})_{123})) \wti \Psi_t \rsp\notag \\
& 
 =i2 t\eN^3 (N-1)\re\lsp\wti{\Psi}_{t},\left(\hg_{1}(t) p_{1}-p_{1}\hg_{1}(t) \right)p_{2}(w_{\nabla f})_{12}  \P^{(2)}_{12} \wti{\Psi}_{t}\rsp\label{eq:c-1}\\
 & \quad \quad + i2 t^2\eN^4 (N-1)\re\lsp\wti{\Psi}_{t},\left(\hg_{1}(t) p_{1}-p_{1}\hg_{1}(t) \right)p_{2}(w_{ f})_{12} \P^{(2)}_{12} \wti{\Psi}_{t}\rsp\label{eq:c-1b}\\
 & \quad  \quad+i t^2 \eN^4 (N-1)(N-2)\re\lsp\wti{\Psi}_{t},\left(\hg_{1}(t) p_{1}-p_{1}\hg_{1}(t) \right)p_{2}p_{3}(w_{ff})_{123} \P^{(2)}_{123} \wti{\Psi}_{t}\rsp\label{eq:c-2}\\
 & \quad  \quad+i2 t\eN^3  (N-1)\re\lsp\wti{\Psi}_{t},\left(\hg_{1}(t)p_{1}-p_{1}\hg_{1}(t) \right)q_{2}(w_{\nabla f})_{12} \P^{(1)}_{12} \wti{\Psi}_{t}\rsp\label{eq:c-3}\\
& \quad  \quad+i2 t^2\eN^4  (N-1)\re\lsp\wti{\Psi}_{t},\left(\hg_{1}(t)p_{1}-p_{1}\hg_{1}(t) \right)q_{2}(w_{ f})_{12} \P^{(1)}_{12} \wti{\Psi}_{t}\rsp\label{eq:c-3b}\\
 & \quad  \quad + i2 t^2 \eN^4 (N-1)(N-2)\re\lsp\wti{\Psi}_{t},\left(\hg_{1}(t) p_{1}-p_{1}\hg_{1}(t)\right) p_{2}q_{3}(w_{ff})_{123} \P^{(1)}_{123} \wti{\Psi}_{t}\rsp\label{eq:c-4}.
\end{align}
\end{subequations}
To estimate these terms, we apply Lemmas \ref{lem:2q-asym-kin-estimate} and \ref{lem:2q-sym-kin-estimate} with $h_1= \hg_1(t)$, and apply an analogous bound to \eqref{eq:1:norm:h:tilde}, namely
\begin{align}
\eN^2 N^{-1} \| \rho_t^{\hg(t)} \|_1  \le C(1+t)^4 D(t)^2.
\end{align} Concretely, by Lemma \ref{lem:2q-asym-kin-estimate} we obtain
\begin{align}
 |\eqref{eq:c-1}| & \leq C(1+t)^{3}D(t) \left( \eN \alpha_n(t) + \eN \|\nabla_{1}q_{1}\wti{\Psi}_{t}\|^2 \right)^\frac12 \left(\alpha_{n}(t)+N^{-1}\right)^{\frac12}  \\[1mm]
  |\eqref{eq:c-1b}|+  |\eqref{eq:c-2}| & \leq C(1+t)^{4}D(t) \alpha_{n}(t) ^{1/2}\left(\alpha_{n}(t) +N^{-1}\right)^{1/2}.
\end{align}
Similarly, by using \prettyref{lem:2q-sym-kin-estimate}, we find
\begin{align}
 |\eqref{eq:c-3}| 
 & \leq C(1+t) D(t) \alpha_{n}(t) ^{1/2}\left((1+t)^{2}D(t)\alpha_{n}(t) ^{1/2}+\eN^{1/2}\| \nabla_{1}q_{1}\wti{\Psi}_{t}\|\right)\notag \\[1mm]
|\eqref{eq:c-3b}| +  |\eqref{eq:c-4}| & \leq C(1+t)^{4}D(t) \alpha_{n}(t) .
\end{align}
In total, we  have shown that
\begin{align}
N^{-1} \eN \left |\lsp\ti{\Psi}_{t},(\textnormal{c})\ti{\Psi}_{t}\rsp \right| & \leq    C(1+t)^{4}D(t)^{2} \left(  \eN \| \nabla_1 q_1 \wti \Psi_t \|^2  +   \alpha_{n}(t)^{1/2}+\alpha_{n}(t)+N^{-\frac{1}{2}}\right).\label{eq:c-total-bound}
\end{align}

Using \eqref{eq:beta-t-derivative} and \eqref{eq:derivative:identity} and adding all estimates together, we arrive at
\begin{align}
|\partial_{t}\beta(t)|\leq    C(1+t)^{4}D(t)^{2}\left( \eN \| \nabla_1 q_1 \wti \Psi_t\|^2 +\alpha_{n}(t)^{1/2}+N^{-\frac{1}{2}}\right).
\end{align}
Invoking Lemma~\ref{lem:bound:D(t)} and Assumption~\ref{ass2}, as well as Proposition \ref{prop:aux-GW} for $\gamma=1$, thus
leads to
\begin{align}
|\partial_{t}\beta(t)| \leq e^{C(1+t)^2}  \beta(t)+  \exp\left( e^{C(1+t)^2}  \right)\left( \alpha_n(0)^{1/2} +N^{-\frac{1}{2}} \right).
\end{align}
This completes the proof of the lemma.
\end{proof}
We can now give the proof of Proposition \ref{prop:kinetic-energy-estimate}.
\begin{proof}[Proof of Proposition \ref{prop:kinetic-energy-estimate}] Lemma \ref{lem:beta-deriv} together with Grönwall's inequality imply
\begin{align}
|\beta(t) | \le \exp \left( e^{C(1+t)^2} \right) \left( \beta(0) + \alpha_n(0)^{1/2} + N^{-\frac12} \right).
\end{align}
Using Lemma \ref{lem:Bad-kinetic-energy-decomp} together with 
\begin{align}
\beta(0) = \eN \| \nabla_1 q_1 \Psi_ 0 \|^2 - \frac1N \eN \sum_{k=1}^N \| \nabla \varphi_k^0 \|^2 ,
\end{align}
where we used $\psi_k^0 = \varphi_k^0$, concludes the proof of the proposition.
\end{proof}

\subsection{Summary of auxiliary bounds}

We now summarize the bounds established in the previous two sections in a form suitable for use in the proof of our main theorem. Recall that Assumption~\eqref{assumption:Phi:varphi} implies
 \begin{align} 
 \alpha_{m^{(\gamma)}}(0) \le C N^{1 - \gamma - \delta_1}, \qquad \left | \eN \|  \nabla_1 \Psi_0 \|^2 - \eN \frac{1}{N}  \sum_{k=1}^N \| \nabla \psi_k^0 \|^2 \right | \le C N^{-\delta_2}.
\end{align}
Here, we used $\mathrm{Tr}(\gamma^{\Phi_0} q^{\varphi_1^0, \ldots, \varphi_N^0}) = \langle \Psi_0, q^{\varphi_1^0, \ldots, \varphi_N^0} \Psi_0 \rangle = \alpha_{m_1}(0)$, and that $\alpha_{m^{(\gamma)}}(0) \le N^{1 - \gamma} \alpha_{m^{(1)}}(0)$, since $m^{(\gamma)}(k) \le N^{1 - \gamma} m^{(1)}(k)$; see \eqref{eq:relevant:weight:functions}. Combining these bounds at initial time $t=0$ with Propositions~\ref{prop:aux-GW} and~\ref{prop:kinetic-energy-estimate}, we obtain the following corollary.

\begin{cor}\label{Cor:apha:beta:bound:delta} 
Under the same assumption as in Proposition \ref{prop:aux-GW}, and for initial states satisfying Condition \eqref{assumption:Phi:varphi}, we have that for all $\gamma \in (0,1]$ and $t\ge 0$
\begin{align}
\alpha_{m^{(\gamma)}}(t) & \le \exp\left( e^{C(1+t)^2} \right) \left(N^{1-\gamma -\delta_1 } +N^{-\gamma}\right) ,\\ 
\eN \|\nabla_{1}q_{1}\wti {\Psi}_{t}\|^2 &  \leq  \exp\left( e^{C(1+t)^2} \right) \left(  N^{-\delta_2 } + N^{-\delta_1/2} +   N^{-1/2} \right)  .\label{eq:bound:kin:energy:delta:2}
\end{align}
\end{cor}

\section{Norm approximation of the gauged dynamics\protect\label{sec:Norm-approximation}}

We establish a norm approximation between the gauged dynamics $\Psi_t$, generated by $\eN H^g(t)$, and the auxiliary dynamics $\widetilde{\Psi}_t$, generated by $\eN \widetilde{H}^g(t)$. At the end of the section, we combine this approximation with the estimates from Corollary \ref{Cor:apha:beta:bound:delta} to prove Theorem~\ref{thm:main}.

Before stating the result, let us briefly explain why a direct application of Duhamel’s formula,
\begin{align}
\| \Psi_t - \widetilde{\Psi}_t \|^2 = 2 \eN \left| \int_0^t  ds  \, \im \lsp \Psi_s, \left( H^g(s) - \widetilde{H}^g(s) \right) \widetilde{\Psi}_s  \rsp \right|, \label{eq:naive-Duhamel}
\end{align}
is not sufficient to establish smallness of the norm approximation with the bounds obtained so far. One of the resulting terms on the right-hand side is
\begin{align}
T := \eN^2 N(N-1) \int_0^t  ds \, s  \, \im \lsp \Psi_s, q_1 p_2 \left( f_{12} \cdot i\nabla_1 + i \nabla_1 \cdot f_{12} \right) q_2 q_1 \widetilde{\Psi}_s \rsp. \label{eq:problematic-term}
\end{align}
As explained in Section~\ref{sec:strategy:proof}, we do not control expressions such as $\| q_1 \Psi_s \|$ for the gauged dynamics. Therefore, we would have to apply the two free $q$-projections to the auxiliary dynamics. This would lead to the estimate $|T| \le C N^\frac23  \int_0^t ds \, s  \| \nabla_1 q_1 \Psi_s\| \, \|q_1 q_2 \wti \Psi_s \|$, which involves the problematic term $\| \nabla_1 q_1 \Psi_s \|$. Clearly, this term is even harder to control than $\| q_1 \Psi_s \|$. Moreover, if we use the bound $\| \nabla_1 q_1  \Psi _s \| \le \| \nabla_1 \Psi _s \| + \| \nabla_1 p_1 \Psi_s \|$, we find at best $\| \nabla_1 q_1 \Psi_s \| = O(N^{1/3})$. Thus, even under the optimal estimate for $\| q_1 q_2 \widetilde{\Psi}_s \| = O(N^{-1})$, this would still yield $|T| = O(1)$, which is not sufficient to prove the desired norm approximation. This shows that a different strategy is needed.

Our new strategy is based on the observation that for proving norm convergence it is essentially sufficient to bound $\|\Psi_t -  \hat w^{(\gamma)} \wti \Psi_t \|$ with weight function $w^{(\gamma)}$ defined in \eqref{eq:relevant:weight:functions}. Since $w^{(\gamma)}$ assigns zero weight to states with more than $N^\gamma$ excitations, the presence of $\hat w^{(\gamma)}$ allows us to exploit multiple $q$ operators more efficiently than in \eqref{eq:problematic-term}. The precise results are contained in the next two statements. Throughout this section, all weight operators are defined again in terms of the projections $p$ and $q$ introduced in \eqref{projections:gauged:Hartree}.
\begin{lem}
\label{lem:Norm-approx-lemma} Let $v$ satisfy Assumption~\ref{ass1}, and let $\psi_1^t, \ldots, \psi_N^t$ be the solutions to the gauged Hartree equations~\eqref{eq:gauged:Hartree} with initial data satisfying Assumption~\ref{ass2}. Let $\Psi_t$ denote gauged dynamics~\eqref{eq:guaged:wf} with normalized initial state $\Psi_0 \in L^2(\mathbb{R}^{3N})$, and let $\wti \Psi_t = \wti U(t,0) \Psi_0$ denote the auxiliary dynamics.  Furthermore, assume that Condition~\eqref{assumption:Phi:varphi} holds for $\delta_1>5/6$ and $\delta_2>1/3$, and consider
\begin{align*}
S^{(\gamma)}(t) & \coloneqq \eN \lsp\Psi_{t},\left(\Hg(t)-\wti H^g (t)\right)\hat{w}^{(\gamma)}\wti{\Psi}_{t}\rsp,\\
T^{(\gamma)}(t) & \coloneqq - \eN \lsp\Psi_{t},\Bigg[\wti H^g (t)-\sum_{j=1}^{N}\hg_{j}(t),\hat{w}^{(\gamma)}\Bigg]\wti{\Psi}_{t}\rsp
\end{align*}
for weight function $w^{(\gamma)}$  defined in \eqref{eq:relevant:weight:functions}.
Then, there exists a constant $C>0$ such that for all $t\geq0$:
\begin{align}
|S^{(1/6)}( t) | + | T^{(1/6)} (t) |  \le C \exp\left( e^{C(1+t)^2} \right)  \left(   N^{ 5/12 -  \delta_1/2 } +  N^{1/6-\delta_2/2}  + N^{1/6 - \delta_1 / 4 }  + N^{-1/12} \right) .
\end{align}
\end{lem}

Before we prove the lemma, we use it to obtain the norm approximation. 
\begin{prop}
\label{prop:Norm-approx-result} Under the same assumptions as in Lemma \ref{lem:Norm-approx-lemma}, there
exists a $C>0$ such that
\begin{align}
\| {\Psi}_{t}- \wti \Psi_{t}\|\leq C \exp\left( e^{C(1+t)^2} \right)  \left(   N^{ 5/24 -  \delta_1/4 } +  N^{1/12-\delta_2/4} + N^{1/12 - \delta_1 / 8 } + N^{-1/24} \right) .
\end{align}
\end{prop}
 
\begin{proof} Let $\gamma\in (0,1]$. With $\hat w^{(\gamma)}  = \id - \hat m^{(\gamma)} $, we estimate
\begin{align}
\|\wti{\Psi}_{t}-\Psi_{t}\| & \leq\| \hat{m}^{(\gamma)}  \wti{\Psi}_{t} \|+\|\hat{w}^{(\gamma)} \wti{\Psi}_{t}-\Psi_{t}\|\nonumber 
\end{align}
and use $\hat m^{(\gamma)}  \le 1$ so that $\| \hat{m}^{(\gamma)}  \wti{\Psi}_{t} \|^2 \le \alpha_{m^{(\gamma)}} (t)$, with the latter defined in \eqref{eq:aux:counting:functional}. Thus, the first term can be directly bounded by Corollary \ref{Cor:apha:beta:bound:delta}. Using $\hat w^{(\gamma)} \le 1$, the second term can be bounded by $\|\hat{w}^{(\gamma)}  \wti{\Psi}_{t}-\Psi_{t}\|^2 \le 2 - 2 \re \langle \Psi_t, \hat w^{(\gamma)}  \wti \Psi_t \rangle =: \delta(t)$. With \eqref{eq:e:o:m:counting:functional} for $\hat f = \hat w^{(\gamma)} $, we can compute the time derivative
\begin{align*}
-i \frac{d}{dt} \delta(t) &  = - 2 \eN \im \lsp \Psi_t, \left(  H^g(t) \hat w^{(\gamma)} - \hat w^{(\gamma)}\wti  H^g(t)  \right) \wti \Psi_t \rsp  - 2\eN  \im \lsp \Psi_t, \left[  \sum_{j=1}^N h_j^g(t) , \hat w^{(\gamma)} \right] \wti \Psi_t \rsp  \notag\\
& = 2 \eN \im \lsp \Psi_t, \left(  H^g(t) -\wti H^g(t) \right) \hat w^{(\gamma)}  \wti \Psi_t \rsp \notag\\
& \qquad + 2\eN  \im \lsp \Psi_t, \left[ \wti H^g(t) - \sum_{j=1}^N h_j^g(t) , \hat w^{(\gamma)} \right] \wti \Psi_t \rsp =  T^{(\gamma)} (t) + S^{(\gamma)} (t).
\end{align*}
Since $\delta(0) = 2  \langle \Psi_0, (1-\hat w^{(\gamma)} ) \Psi_0 \rangle =  2 \alpha_{m^{(\gamma)} } (0) \le C N^{ 1  - \gamma -  \delta_1} \le C N^{5/6-\delta_1}$ for $\gamma=1/6$, the statement now follows from Lemma \ref{lem:Norm-approx-lemma}.
\end{proof}

\begin{proof}[Proof of \prettyref{lem:Norm-approx-lemma}]
We bound $T^{(\gamma)} (t)$ and $S^{(\gamma)} (t)$ separately. Note that $T^{(\gamma)} (t)$ looks similar to the time derivative of $\alpha_{m^{(\gamma)} }(t) = \langle \wti \Psi_t, \hat m^{(\gamma)}  \wti \Psi_t\rangle$ computed in \eqref{eq:time-derivative-alpha}, but with two different
states on the left and right side of the inner product.\medskip

\noindent \underline{Bound for $T^{(\gamma)} (t)$:} Following the steps in \eqref{eq:time-derivative-alpha}--\eqref{eq:B_123}, and abbreviating $\hat m^{(\gamma)}  \equiv \hat m$, we obtain
\begin{align}
 S^{(\gamma)}  (t) = \frac{1}{2} \widetilde A(t) +  \frac{1}{6} \widetilde B(t) 
\end{align}
with
\begin{subequations}
\begin{align}
 - \widetilde A(t) & = t\eN^2 N \lsp\Psi_{t},\big(\hat{m}-\hat{m}_{-1}\big)\left((N-1)\P^{(1)}_{12} (w_{\nabla f})_{12}^{\phantom{()}} \P^{(0)}_{12}-2  q_{1} R_{1}(t) p_{1}\right)\wti{\Psi}_{t}\rsp\label{eq:I_bx}\\
& \quad - t \eN^2 N\lsp\Psi_{t},\left((N-1)\P^{(0)}_{12} (w_{\nabla f})_{12}^{\phantom{()}} \P^{(1)}_{12}-2 p_{1} R_{1}(t) q_{1}\right)\big(\hat{m}-\hat{m}_{-1}\big)\wti{\Psi}_{t}\rsp\label{eq:I_ax}\\
& \quad + t^2\eN^3 N \lsp\Psi_{t},\big(\hat{m}-\hat{m}_{-1}\big)\left((N-1)\P^{(1)}_{12}(w_{f})_{12}^{\phantom{()}} \P^{(0)}_{12}-2  q_{1}\ R_{1}(t) p_{1}\right)\wti{\Psi}_{t}\rsp\label{eq:I_bx:2}\\
& \quad - t^2 \eN^3 N\lsp\Psi_{t},\left((N-1)\P^{(0)}_{12} (w_{f})_{12}^{\phantom{()}} \P^{(1)}_{12} -2 p_{1} R_{1}(t) q_{1}\right)\big(\hat{m}-\hat{m}_{-1}\big)\wti{\Psi}_{t}\rsp\label{eq:I_ax:2}\\
 & \quad + t\eN^2  N(N-1)\lsp\Psi_{t},\big(\hat{m}-\hat{m}_{-2}\big) \left( \P^{(2)}_{12} (w_{\nabla f})_{12}^{\phantom{()}} \P^{(0)}_{12}\right) \wti{\Psi}_{t}\rsp\label{eq:II_bx:b} \\
& \quad - t\eN^2 N(N-1)\lsp\Psi_{t}, \left( \P^{(0)}_{12} (w_{\nabla f})_{12}^{\phantom{()}} \P^{(2)}_{12} \right)  \big(\hat{m}-\hat{m}_{-2}\big)\wti{\Psi}_{t}\rsp\label{eq:II_ax}\\
 & \quad + t^2\eN^3  N(N-1)\lsp\Psi_{t},\big(\hat{m}-\hat{m}_{-2}\big)\left( \P^{(2)}_{12} (w_{f})_{12}^{\phantom{()}} \P^{(0)}_{12}\right) \wti{\Psi}_{t}\rsp\label{eq:II_bx:2} \\
& \quad - t^2\eN^3 N(N-1)\lsp\Psi_{t},\left( \P^{(0)}_{12} (w_{ f})_{12}^{\phantom{()}} \P^{(2)}_{12} \right) \big(\hat{m}-\hat{m}_{-2}\big)\wti{\Psi}_{t}\rsp\label{eq:II_ax:2}
\end{align}
\end{subequations}
and
\begin{subequations}
\begin{align}
 - \widetilde B(t) & =t^2 \eN^3 N\lsp\Psi_{t},\left(\hat{m}-\hat{m}_{-1}\right)\left((N-1)(N-2)\P^{(1)}_{123} (w_{ff})_{123}^{\phantom{()}} \P^{(0)}_{123}-6  q_{1}W_{1}(t) p_{1}\right)\wti{\Psi}_{t}\rsp \label{eq:IV_bx}\\ 
 &  - t^2 \eN^3 N\lsp \Psi_{t},\left((N-1)(N-2) \P^{(0)}_{123} (w_{ff})_{123}^{\phantom{()}} \P^{(1)}_{123} -6  p_{1}W_{1}(t) q_{1}\right)\left(\hat{m}-\hat{m}_{-1}\right)\wti{\Psi}_{t}\rsp \label{eq:IV_ax}\\
 &  +t^2 \eN^3  N(N-1)(N-2)\lsp \Psi_{t},\left(\hat{m}-\hat{m}_{-2}\right) \left( \P^{(2)}_{123} (w_{ff})_{123}^{\phantom{()}} \P^{(0)}_{123} \right) \wti{\Psi}_{t}\rsp.\label{eq:V_bx}\\
&  - t^2 \eN^3 N(N-1)(N-2)\lsp \Psi_{t},\left( \P^{(0)}_{123} (w_{ff})_{123}^{\phantom{()}} \P^{(2)}_{123}\right)  \left(\hat{m}-\hat{m}_{-2}\right)\wti{\Psi}_{t}\rsp\label{eq:V_ax}.
\end{align}
\end{subequations}
From here we proceed similarly as explained in the proof of Lemma \ref{lem:Estimates-derivative-of-alpha_m}, that is we factorize $\hat{m}-\hat{m}_{-d} = \hat D_{-d} \hat D_{-d}$ with $\hat D_{-d} \equiv \hat D^{(\gamma)}_{-d}$, and shift one of the operators to the other side of the two- and tree-particle operator. However, note that unlike in the proof of Lemma \ref{lem:Estimates-derivative-of-alpha_m}, here we do not aim at closing a Grönwall argument. We continue by estimating each term separately.

We start with the second and fourth terms, where we apply Lemma \ref{lem:0q-diag-estimate} (with $r_1^{(1)} = 2q_1$) and Lemma \ref{lem:1q-estimate} with $\psi\equiv\hat{D}_{-1}\wti{\Psi}_{t}$
and $\varphi\equiv\hat{E}_{-1}\Psi_{t}$. This gives
\begin{align}
 |\eqref{eq:I_ax}| +  |\eqref{eq:I_ax:2}| & \le t \eN^2 N \left | \lsp \hat E_{-1} \Psi_t ,  \left((N-1) \P^{(0)}_{12 }( w_{\nabla f })_{12}^{\phantom{()}} \P^{(1)}_{12}-2  p_{1} R_{1}(t) q_{1}\right)  \hat D_{-1} \wti \Psi_t \rsp \right| \notag\\
 & \quad +  t^2 \eN^3 N^2  \left| \lsp \hat E_{-1} \Psi_t ,  \left( \P^{(0)}_{12} (w_f)_{12}^{\phantom{()}}  \P^{(1)}_{12}\right)    \hat D_{-1} \wti \Psi_t \rsp \right| \notag\\
 & \le C t N^{-\frac13} (N^\frac12 + \| \rho_t^\nabla \|_1^\frac12 )\| \hat E_{-1} \Psi_t \|\big( N  \| q_1 q_2 \hat D_{-1} \ti \Psi_t \|^2 + \| q_1 \hat D_{-1} \wti \Psi_t \|^2 \big)^{1/2} \notag\\
 & \quad + C t^2   \| \hat E_{-1} \Psi_t \| \, \| q_1 \hat D_{-1} \wti \Psi_t \|.
\end{align}
Invoking $\| \hat E_{-1} \Psi_t \|^2 \le C N^{-\gamma } $,  $\| q_2 \hat D_{-1} \wti \Psi_t \|^2  \le N^{- 1 } \alpha_{m^{(\gamma)} }(t)$ and $\| q_1 q_2 \hat D_{-1} \wti \Psi_t \|^2 \le N^{\gamma-2} \alpha_{m^{(\gamma)} }(t)$, see \eqref{eq:weight:estimates:rep} and Lemma \ref{lem:Weight-estimate-m}, we obtain
\begin{align}
 |\eqref{eq:I_ax}| +  |\eqref{eq:I_ax:2}| & \le C (1+t)^2 (N^{-\frac13} + N^{-\frac56}\| \rho^\nabla \|_1^\frac12 )\|  \sqrt{\alpha_{m^{(\gamma)}  }(t)}
  + C t^2 N^{-\frac{\gamma}{2} - \frac12} \sqrt{\alpha_{m^{(\gamma)} }(t) } \ .
\end{align} 

The first and third lines are the same with $\Psi_t$ and $\wti \Psi_t$ replaced. However, since we cannot exploit the $q$ operator on $\Psi_t$, we need to insert $\hat \ell \hat {\ell}^{-1} = 1$, with $\ell$ defined in \eqref{eq:relevant:weight:functions}, and shift $\hat \ell$ to $\wti \Psi_t$ via Lemma \ref{lem:Projection-shift}. With the aid of Lemma \ref{lem:0q-diag-estimate} (with $r_1^{(1)} = 2q_1$) and Lemma \ref{lem:1q-estimate} applied to $\psi\equiv \hat {\ell}^{-1} \hat{D}_{-1} {\Psi}_{t}$
and $\varphi\equiv \hat \ell_{+1} \hat{E}_{-1}\wti \Psi_{t}$, this gives
\begin{align}
  |\eqref{eq:I_bx}| +   |\eqref{eq:I_bx:2}|  &   \le t N^{-\frac13} \left| \lsp \hat \ell_{+1} \hat E_{-1} \wti \Psi_t ,  \left((N-1)\P^{(0)}_{12} ( w_{\nabla f })_{12}^{\phantom{()}} \P^{(1)}_{12} -2   p_{1} R_{1} (t) q_{1}\right) \hat {\ell}^{-1} \hat D_{-1} \Psi_t \rsp \right| \notag \\
 & \qquad \quad \quad +  t^2 \left| \lsp \hat E_{-1} \wti \Psi_t ,  \left( \P^{(0)}_{12}(w_f)_{12}^{\phantom{()}} \P^{(1)}_{12} \right)  \hat D_{-1} \Psi_t \rsp \right| \notag\\
& \le C t N^{-\frac13} (N^\frac12 + \| \rho^\nabla \|_1^\frac12 )\| \hat \ell_{+1} \hat E_{-1} \wti \Psi_t \|\big( N  \| \hat {\ell}^{-1} q_1 q_2 \hat D_{-1} \wti \Psi_t \|^2 + \| \hat {\ell}^{-1} q_1 \hat D_{-1} \Psi_t \|^2 \big)^{1/2} \notag\\
 & \qquad \quad \quad + C t^2   \| \hat \ell_{+1} \hat E_{-1}\wti \Psi_t \| \, \| \hat {\ell}^{-1} q_1 \hat D_{-1} \Psi_t \|.
\end{align}
Invoking in addition \eqref{eq:ell:to:n} and Lemma \ref{lem:l-inverse conversion}, and then using again Lemma \ref{lem:Weight-estimate-m}, we can bound this term in the same way as the previous one:
\begin{align}
 |\eqref{eq:I_bx}| +   |\eqref{eq:I_bx:2}|    & \le C (1+t)^2 (N^{-\frac13} + N^{-\frac56}\| \rho^\nabla_t \|_1^\frac12 )\|  \sqrt{\alpha_{m^{(\gamma)}}(t)}.
\end{align}

We proceed with Lemma \ref{lem:2q-estimate-asym} (here we only need to use one $q$ operator, hence no additional shifting is required):
\begin{align}
 | \eqref{eq:II_ax} | +  | \eqref{eq:II_ax:2} |& \le  t N^{\frac23} \left| \lsp \hat E_{-2} \Psi_{t}, \left( \P^{(2)}_{12} (w_{\nabla f})_{12}^{\phantom{()}} \P^{(0)}_{12} \right) \hat D_{-2} \wti{\Psi}_{t}\rsp  \right| \notag\\
& \quad + t^2    \left| \lsp \hat E_{-2} \Psi_{t} , \left(  \P^{(2)}_{12} (w_{f})_{12}^{\phantom{()}} \P^{(0)}_{12} \right) \hat D_{-2} \wti{\Psi}_{t}\rsp  \right| \notag\\
& \le C t (1+ N^{-\frac12} \| \rho_t^\nabla\|_1^{\frac12}) \| q_1 \hat E_{-2} \Psi_t\| \big( N^{-1} \| \hat D_{-2} \wti \Psi_t \|^2 + \| q_1 \hat D_{-2}  \wti \Psi_t \|^2 \big)^{1/2} \notag\\
& \quad + C t^2 \| q_1 \hat E_{-2} \Psi_t \| \big( N^{-1} \| D_{-2} \wti \Psi_t \|^2 + \| q_2 D_{-2} \wti \Psi_t \|^2 \big)^{1/2}\notag\\
& \le C (1+t)^2  (1+ N^{-\frac56} \| \rho_t^\nabla\|_1^{\frac12})  ( N^{-\gamma} + \alpha_{m^{(\gamma)}  }(t)) ^{\frac12} .
\end{align}
and the same bound holds for $ | \eqref {eq:II_bx:b} | +  | \eqref {eq:II_bx:2}|$.

For the remaining two terms we estimate via Lemma \ref{lem:0q-diag-estimate} with $s_1^{(1)}=3q_1$ and $\varphi = \hat E_{-1} \Psi_t$ and $\psi  = \hat D_{-1} \wti \Psi_t$:
\begin{align}
 |\eqref{eq:IV_bx}| 
& =  t^2 \eN^3 N \left| \lsp \hat  E_{-1} \Psi_{t}, \left((N-1)(N-2) \P^{(0)}_{123}(w_{ff})_{123}^{\phantom{()}} \P^{(1)}_{123} -6 p_{1}W_{1}(t)q_{1}\right) \hat D_{-1} \wti{\Psi}_{t}\rsp \right| \notag\\
&  \le C t^2 N^{\frac12}   \| \hat E_{-1} \Psi_t \| ( N \| q_1 q_2 \hat D_{-1}  \wti \Psi_t \|^2 + \| q_1  \hat D_{-1} \wti  \Psi_t\|^2 )^{1/2} \notag\\[0.5mm]
&  \le  C t^2 N^{-\gamma/2}  ( N^\gamma +1 )^{1/2} \sqrt{  \alpha_{m^{(\gamma)}  }(t) }
\end{align}
as well as with $\varphi = \hat \ell_{+1} \hat E_{-1} \wti \Psi_t$ and $\psi  = \hat {\ell}^{-1} \hat D_{-1} \Psi_t$:
\begin{align}
 |\eqref{eq:IV_ax}| & \le  t^2 \eN^3 N \left| \lsp \hat \ell_{+1}  \hat E_{-1} \wti \Psi_{t}, \left((N-1)(N-2) \P^{(0)}_{123} (w_{ff})_{123}^{\phantom{()}} \P^{(1)}_{123 } - 6 t p_{1}W_{1}(t) q_{1}\right)  \hat {\ell}^{-1} \hat D_{-1} {\Psi}_{t}\rsp \right| \notag\\
& \le C t^2 N^{\frac12}   \| \hat \ell_{+1} \hat E_{-1} \wti \Psi_t \| \left( N \| \hat {\ell}^{-1}  q_1 q_2 \hat D_{-1}   \Psi_t \|^2 + \| \hat {\ell}^{-1} q_2  \hat D_{-1}   \Psi_t\|^2 \right)^{1/2} \notag\\[0.5mm]
&  \le  C t^2 \sqrt{  \alpha_{m^{(\gamma)} }(t) }.
\end{align}

Moreover, via Lemma \ref{lem:2q-estimate-asym} for $ \psi = \hat D_{-2} \wti \Psi_t $ and $\varphi = \hat E_{-2} \Psi_t$, we obtain
\begin{align}
 |\eqref{eq:V_ax} |  & = t^2 N \left| \lsp \hat E_{-2} \Psi_{t}, \left( \P^{(0)}_{123}(w_{ff})_{123}^{\phantom{()}} \P^{(2)}_{123} \right)  \hat D_{-2} \wti{\Psi}_{t}\rsp \right| \notag\\ 
& \le C t^2 N  \|  q_1 \hat D_{-2} \wti{\Psi}_{t} \|  ( N^{-1} \| \hat E_{-2} \Psi_t \|^2 + \| q_1 \hat E_{-2} \Psi_t \|^2 )^{1/2} \notag\\
& \le  C t^2  \sqrt{\alpha_{m^{(\gamma)} }(t)} 
\end{align}
and the same bound holds for \eqref{eq:V_bx}. 

After invoking Corollary \ref{Cor:apha:beta:bound:delta}, we can thus conclude that  
\begin{align}
|T^{(\gamma)}(t) | \le  \exp\left(e^{C(1+t)^2}\right)  \Big( N^{(1-\gamma -\delta_1 )/ 2 } + N^{-\gamma/2}\Big). \label{eq:bound:B:gamma}
\end{align}
For $\gamma = 1/6$, we obtain $|T^{(\gamma)}  (t) | \le \exp( e^{C(1+t)^2} ) ( N^{ 5/12 -  \delta_1/2 } + N^{-1/12} )$. \medskip

\noindent \underline{Bound for $S^{(\gamma)} (t)$:} We recall the definitions \eqref{eq:def:P0}--\eqref{eq:def:P2} and introduce
\begin{align} \label{eq:def:P:3}
\P_{ijk}^{(3)} := q_i q_j q_k.
\end{align}
With these, we write the difference
between the two Hamiltonians as
\begin{align}
\Hg(t)-\wti H^g(t) & =W^{(3)}+W^{(4)}+W^{(5)}+W^{(6)}
\end{align}
with 
\begin{align}
W^{(a)} & \coloneqq \sum_{1\leq i < j \leq N}\left( t\eN (w_{\nabla f})_{ij}^{(a)} + t^ 2\eN^2  (w_{ f})_{ij}^{(a)} \right) +\sum_{1\leq i < j < k \leq N}  t^ 2\eN^2  (w_{ff})_{ijk}^{(a)},
\end{align}
where we introduced for $X\in \{ \nabla f, f\}$:
\begin{subequations}
\begin{align}
 (w_X)_{ij}^{(3)} & \coloneqq \P^{(1)}_{ij} ( w_{X} )_ {ij}^{\phantom{()}} \P^{(2)}_{ij} + h.c. ,  \label{eq:short-two-body:remainder} \\ 
 (w_X)_{ij}^{(4)} & \coloneqq \P^{(2)}_{ij} ( w_{X} )_ {ij}^{\phantom{()}}  \P^{(2)}_{ij}  ,\\
  (w_X)_{ij}^{(5)} =   (w_X)_{ij}^{(6)}  & \coloneqq  0 
\end{align}
\end{subequations}
and:
\begin{subequations}
\begin{align}
(w_{ff})_{ijk}^{(3)} & \coloneqq \left( \P^{(3)}_{ijk } ( w_{ff})_{ijk}^{\phantom{()}} \P^{(0)}_{ijk} +  \P^{(2)}_{ijk} ( w_{ff})_{ijk}^{\phantom{()}} \P^{(1)}_{ijk} \right) + h.c.   ,\label{eq:short-three-body:remainder}
\\
(w_{ff})_{ijk}^{(4)} & \coloneqq \left( \P^{(3)}_{ijk} ( w_{ff})_{ijk}^{\phantom{()}} \P^{(1)}_{ijk } +  \P^{(2)}_{ijk} ( w_{ff})_{ijk}^{\phantom{()}} \P^{(2)}_{ijk} \right) + h.c.   ,
\\
(w_{ff})_{ijk}^{(5)} & \coloneqq  \P^{(3)}_{ijk} ( w_{ff})_{ijk}^{\phantom{()}} \P^{(2)}_{ijk } + h.c.  ,\\
(w_{ff})_{ijk}^{(6)} & \coloneqq  \P^{(3)}_{ijk} ( w_{ff})_{ijk}^{\phantom{()}} \P^{(3)}_{ijk} .
\end{align}
\end{subequations}

The general strategy for estimating $T^{(\gamma)}(t)$ is to apply Lemmas \ref{lem:3q-estimate} and \ref{lem:4q-estimate} for terms involving three or more $q$ operators, and to then exploit the presence of $q$ operators together with the weight functional $\hat w^{(\gamma)}$ using \prettyref{lem:Weight-estimate-w}. A key point is that the gradient $\nabla_i$ in $w_{\nabla f}$ must always be moved to the right, so that it acts on the side of $\wti \Psi_t$, while shifting $\hat w^{(\gamma)}$ to the side of $\Psi_t$. In this situation, all other $q$ operators are also moved to act on $\Psi_t$, and we estimate expressions such as $q_1 \hat w^{(\gamma)}_{\pm 1} \Psi_t$ and $q_1 q_2 \hat w^{(\gamma)}_{\pm} \Psi_t$ via Lemma \ref{lem:Weight-estimate-m}. The presence of the (shifted) weight operator $\hat w^{(\gamma)}_{\pm 1}$ is crucial at this step, since without it we have no control over $ q_1 \Psi_t $ and $q_1 q_2 \Psi_t$. For the terms involving $w_f$ and $w_{ff}$ we always shift $\hat w^{(\gamma)}$ and sufficiently many $q$'s to either of the two sides, and then use Lemma \ref{lem:Weight-estimate-w}.\medskip

\noindent \uline{Term $W^{(3)}$}: We decompose $
 \langle \Psi_{t},W^{(3)}\hat{w}^{(\gamma)} \wti{\Psi}_{t}\rangle = (\rm{IIIa}) + (\rm{IIIb}) + (\rm{IIIc}) + (\rm{IIId})  $ with
\begin{align}
(\rm {IIIa}) & := t\eN^2 \tfrac{N(N-1)}{2}\lsp\Psi_{t}, \left(  \P^{(1)}_{12} \, (w_{\nabla f})_{12}^{\phantom{()}} \,  \P^{(2)}_{12} + \P^{(2)}_{12} \, (w_{\nabla f})_{12}^{\phantom{()}} \,   \P^{(1)}_{12}  \right)  \hat w^{(\gamma)}  \wti{\Psi}_{t}\rsp\nonumber \\
(\rm {IIIb})  & : = t^2 \eN^3 \tfrac{N(N-1)}{2}\lsp\Psi_{t}, \left( \P^{(1)}_{12}\, (w_{f})_{12}^{\phantom{()}} \,   \P^{(2)}_{12} +  \P^{(2)}_{12}\, (w_{f})_{12}^{\phantom{()}}  \,   \P^{(1)}_{12}\right)  \hat w^{(\gamma)}   \wti{\Psi}_{t}\rsp\nonumber \\
(\rm {IIIc})  & : =  t^2 \eN^3 \tfrac{N(N-1)(N-2)}{6}\lsp\Psi_{t}, \left( \P^{(1)}_{123}\, (w_{ff})_{123}^{\phantom{()}}  \,   \P^{(2)}_{123}  +  \P^{(2)}_{123}\, (w_{ff})_{123}^{\phantom{()}}  \,   \P^{(1)}_{123} \right) \hat w^{(\gamma)}  \wti{\Psi}_{t}\rsp\nonumber \\
 (\rm {IIId})  & : = t^2 \eN^3 \tfrac{N(N-1)(N-2)}{6}\lsp \Psi_{t}, \left( \P^{(0)}_{123}\, (w_{ff})_{123}^{\phantom{()}}  \,   \P^{(3)}_{123} +   \P^{(3)}_{123}\, (w_{ff})_{123}^{\phantom{()}}  \,   \P^{(0)}_{123} \right) \hat w^{(\gamma)} \wti{\Psi}_{t}\rsp\notag .
\end{align}
We further estimate each term as $|({\rm {IIIx}})| \le ({\rm {IIIx1}}) + ({\rm {IIIx2}}) $ for ${\rm x} \in \{{ \rm a,b,c,d}\}$.  For the first term we apply Lemma \ref{lem:Weight-estimate-m} to shift $\hat w^{(\gamma)}$ to the left side (see also \eqref {eq:commutator:m:diff}) and then use Lemma \ref{lem:3q-estimate} Eq. \eqref{eq:3q:bound:12} for $ \varphi = \hat w^{(\gamma)}_{+1} \Psi_t $ and $\psi = \wti \Psi_t$ to estimate
\begin{align} 
({\rm {IIIa1}}) & := t\eN^2 \tfrac{N(N-1)}{2} \left|\lsp \hat w^{(\gamma)}_{+1} \Psi_{t},   \left( \P^{(1)}_{12}  \,   (w_{\nabla f})_{12}^{\phantom{()}} \,   \P^{(2)}_{12} \right) \wti{\Psi}_{t}\rsp\right| \notag \\
&\quad  \le C t N \left(  N^{-1/2} \|q_{1} \hat w^{(\gamma)}_{+1}  \Psi_t \| +  \|q_{1} q_2\hat w^{(\gamma)}_{+1}  \Psi_t \| \right)  \left(   \eN^{1/2} \|\nabla q_{1}\wti \Psi_t \| +\eN^{1/2}  \|q_{1}\wti \Psi_t \|    \right) .
\end{align}
We then use 
\begin{align}
\|q_{1} \ldots q_{n_0}\hat w^{(\gamma)}_{\pm d}  \Psi_t \| \le  \sqrt 2  N^{\gamma/2-1/2}
\end{align} 
for $n_0 \ge 1 $ and $d\in \{0,1,2,3\}$, by Lemma \ref{lem:Weight-estimate-w}, and invoke Inequality \eqref{eq:bound:kin:energy:delta:2} from Corollary \ref{Cor:apha:beta:bound:delta} and $\| q_1 \wti \Psi_t \| \le 1$. This gives
\begin{align}
| ({\rm {IIIa1}})| & \le C \exp\left( e^{C(1+t)^2} \right)  N^{\gamma}  \left(   N^{-\delta_2/2} + N^{-\delta_1/4 } + N^{-1/4} \right) \notag\\
& =  C \exp\left( e^{C(1+t)^2} \right)  \left(   N^{1/6-\delta_2/2} +  N^{1/6 - \delta_1 / 4 } + N^{-1/12} \right) 
\end{align}
where we set $\gamma=1/6$. With Lemma \ref{lem:3q-estimate} Eq. \eqref{eq:3q:bound:21}, we obtain the same bound for
\begin{align} 
({\rm {IIIa2}}) & := t\eN^2 \tfrac{N(N-1)}{2} \left|\lsp w^{(\gamma)}_{-1}   \Psi_{t},   \left( \P^{(2)}_{12}   \,   (w_{\nabla f})_{12}^{\phantom{()}} \,   \P^{(1)}_{12} \right)   \wti{\Psi}_{t}\rsp\right| \notag \\
& \le  C t N \| q_1 q_2 w^{(\gamma)}_{-1}  \| \left( \eN^{1/2} \| \nabla_1 q_1 \wti \Psi_t \| + N^{-5/6} \| \rho_t^\nabla \|_1^{1/2} \| q_1 \wti \Psi_t \|    \right) \notag\\
& \le  C \exp\left( e^{C(1+t)^2} \right)  \left(   N^{1/6-\delta_2/2} + N^{1/6-\delta_1/4 } + N^{-1/12} \right) ,
\end{align}
where we invoked $\| q_1 \wti \Psi_t \|^2 \le \alpha_{m^{(1)}}(t)$, Corollary \ref{Cor:apha:beta:bound:delta} as well as $N^{-5/6} \| \rho_t^\nabla \|_1^{1/2} \le D(t) \le C e^{(1+t)^2}$ by Lemma \ref{lem:bound:D(t)}. 

Using Lemma \ref{lem:3q-estimate}, the other terms can be bounded similarly. For instance,
\begin{align}
| (\rm {IIIc1})| & := t^2 \eN^3 \tfrac{N(N-1)(N-2)}{6} \left| \lsp\Psi_{t},  \P^{(1)}_{123} \, (w_{ff})_{123}^{\phantom{()}} \,   \P^{(2)}_{123}  \hat w^{(\gamma)}   \wti{\Psi}_{t}\rsp \right|  \notag\\
& \le C t^2 N \left( N^{-1/2} \| q_1 q_2   \hat w^{(\gamma)}   \wti{\Psi}_{t} \| + \| q_1 q_2 q_3   \hat w^{(\gamma)}   \wti{\Psi}_{t} \| \right)  \notag\\
& \le  C t^2 \left( N^{\gamma-1/2} + N^{3/2 \gamma - 1/2} \right) \le C t^2 N^{-1/12}, 
\end{align}
as well as
\begin{align}
| (\rm {IIIc1})| & := t^2 \eN^3 \tfrac{N(N-1)(N-2)}{6} \left| \lsp\Psi_{t}, \hat w^{(\gamma)}_{-1} \P^{(2)}_{123} \, (w_{ff})_{123}^{\phantom{()}} \,   \P^{(1)}_{123}    \wti{\Psi}_{t}\rsp \right|  \notag\\
& \le C t^2 N \left( N^{-1/2} \| q_1 q_2   \hat w^{(\gamma)}_{-1}   {\Psi}_{t} \| + \| q_1 q_2 q_3   \hat w^{(\gamma)}_{-1}   {\Psi}_{t} \| \right)  \le C t^2 N^{-1/12}.
\end{align}
In the same way, one finds that $|({\rm {IIIb})} | + | {(\rm {IIId})} | \le C t^2 N^{-1/12}$, which completes the bound on $W^{(3)}$.

\medskip

\noindent \underline{Terms $W^{(4)}$--$W^{(6)}$}: Using Lemma \ref{lem:4q-estimate}, these terms are estimated in the same way as described for $W^{(3)}$. In conclusion, we arrive at the bound
\begin{align}
|T^{(\gamma)} (t) | \le C \exp\left( e^{C(1+t)^2} \right)  \left(   N^{1/6-\delta_2/2} + N^{1/6-\delta_1/4 } + N^{-1/12} \right).
\end{align}
This completes the proof of the lemma.
\end{proof}

\subsection{Proof of Theorem \ref{thm:main}}

We are now ready to prove our main result. We begin by estimating
\begin{align}
\lsp \Psi_t , q_1 \Psi_t \rsp & \le \left| \lsp \Psi_t , q_1 ( \Psi_t - \wti \Psi_t) \rsp \right| + \| q_1 \Psi_t \| \| q_1 \wti \Psi_t \| \le \tfrac{1}{2} \| q_1 \Psi_t \|^2 + 2 \| \Psi_t - \wti \Psi_t \|^2 + 2 \| q_1 \wti \Psi_t \|^2.
\end{align}
Applying Corollary~\ref{Cor:apha:beta:bound:delta} (with $\gamma=1$) and Proposition~\ref{prop:Norm-approx-result}, we obtain
\begin{align}\label{eq:bound:for:q:Psi}
\| q_1 \Psi_t \| \le C \exp\left( e^{C(1+t)^2} \right) \left( N^{5/24 - \delta_1/4} + N^{1/12 - \delta_2/4} + N^{1/12 - \delta_1/8} + N^{-1/24} \right).
\end{align}

Next, we estimate the difference of one-particle expectation values:
\begin{align}
& \left| \tr \left( A \gamma^{\Psi_t} \right) - \frac{1}{N} \tr (A p) \right| = \left| \lsp \Psi_t , A_1 \Psi_t \rsp - \frac{1}{N} \tr (p A p) \right| \notag \\
&\qquad \qquad \le \left| \lsp \Psi_t , p_1 A_1 p_1 \Psi_t \rsp - \frac{1}{N} \tr (p A p) \right| + 3 \|A\| \|q_1 \Psi_t\| \le \|A\| \left( 4 \| q_1 \Psi_t \| + N^{-1/2} \right), \label{eq:trace:estimate}
\end{align}
where we inserted the identity $\id = p_1 + q_1$ to the left and right of $A_1$, used Cauchy--Schwarz to bound the terms involving $q_1$, and applied Lemma~\ref{lem:hg-0q-estimate} to control the remaining term, that is, we invoked
\begin{align}
\left| \lsp \Psi_t , A_1 \Psi_t \rsp - \frac{1}{N} \tr (p A p) \right| \le C\| A\| \left( \| q_1 \Psi_t \| + N^{-1/2} \right),
\end{align}
where we used that $\|\rho^A\|_1 \le \|A\|^2 \| \rho_t\|_1 \le \|A\|^2 N$.

Finally, recall that we are interested in approximating expectation values with respect to the ungauged dynamics $\Phi_t$ and the ungauged orbitals $\varphi_1^t, \dots, \varphi_N^t$. Using that the gauge transformations \eqref{eq:guaged:wf} and \eqref{eq:guaged:Hartree:sol}, which relate $\Psi_t$ and $\Phi_t$, as well as $p = p^{\psi_1^t, \dots, \psi_N^t}$ and $p^{\varphi_1^t, \dots, \varphi_N^t}$, commute with multiplication operators $A \equiv M$, we obtain that
\begin{align}
\tr \left( M \gamma^{\Phi_t} \right) - \frac{1}{N} \tr \left( M p^{\varphi_1^t, \dots, \varphi_N^t} \right) = \tr \left( M \gamma^{\Psi_t} \right) - \frac{1}{N} \tr (M p).
\end{align}
Theorem~\ref{thm:main} now follows directly from the estimates \eqref{eq:trace:estimate} and \eqref{eq:bound:for:q:Psi}.
\hfill $\blacksquare$

\section{Technical toolbox}
\label{sec:toolbox}

Throughout this section, we consider projections
\begin{align}\label{eq:p:toolbox}
p = p^{\varphi_{1},\ldots,\varphi_{N}} = \sum_{k=1}^N | \varphi_k \rangle \langle \varphi_k | \quad \text{and} \quad q = \id - p
\end{align}
for a general orthonormal system $\{\varphi_{i}\}_{i=1}^{N} \subset L^2(\mathbb{R}^3)$. With these projections, we set
\begin{alignat}{3}
\P^{(0)}_i & \coloneqq p_i , &  \qquad \qquad  \P_{ij}^{(0)} & \coloneqq p_{i}p_{j}, \qquad \qquad \qquad & \P_{ijk}^{(0)}  & \coloneqq p_{i}p_{j}p_{k},\label{eq:def:P0}\\
\P^{(1)}_i & \coloneqq q_i , &  \qquad \qquad   \P_{ij}^{(1)} & \coloneqq p_{i}q_{j}+q_{i}p_{j},  & \P_{ijk}^{(1)} & \coloneqq p_{i}p_{j}q_{k}+p_{i}q_{j}p_{k}+q_{i}p_{j}p_{k},\\
& & \P_{ij}^{(2)} & \coloneqq q_{i}q_{j}, & \P_{ijk}^{(2)} & \coloneqq p_{i}q_{j}q_{k}+q_{i}p_{j}q_{k}+q_{i}q_{j}p_{k},\label{eq:def:P2}\\
& &  &  & \P_{ijk}^{(3)} & \coloneqq q_{i}q_{j}q_{k}\label{eq:def:P2}.
\end{alignat}

\subsection{Counting functional estimates}
\label{sec:toolbox:1}

The purpose of this section is to summarize a collection of important properties and estimates involving the weight operator and the shifted weight operator from Definition~\ref{def:counting:functional}. The following statements are used repeatedly in the main part of the proof.

\begin{lem}[Lemma 6.4 in \cite{Petrat2016}]
\label{lem:Projection-shift} Let $A_{\mathcal C}$ be a self-adjoint operator only
acting on particle indices $\mathcal{C}\subseteq\{1,2,\ldots,N\}$
with $|\mathcal{C}| \leq3$. Then it holds for all $a,b=0,1,\ldots,3$ that
\begin{align*}
\hat{f}\left( \P^{(a)}_{\mathcal{C}}A_{\mathcal C}^{\phantom{()}} \P^{(b)}_{\mathcal{C}}\right) & =\left(\P^{(a)}_{\mathcal{C}}A_{\mathcal C}^{\phantom{()}} \P^{(b)}_{\mathcal{C}}\right)\hat{f}_{a-b}.
\end{align*}
\begin{comment}
Moreover for $\mathcal{C}_{1}\cap\mathcal{C}_{2}=\varnothing$ and $a,b=0,1,\ldots,3$
and $a',b'=0,1,\ldots,3$ 
\begin{align*}
\hat{f}\, \P^{(a)}_{\mathcal{C}_{1}}\left(\P^{(a')}_{\mathcal{C}_{2}} A_{\mathcal C}^{\phantom{()}} \P^{(b')}_{\mathcal{C}_{2}}\right)\P^{(b)}_{\mathcal{C}_{1}} & =\P_{a}^{\mathcal{C}_{1}}\left(\P^{(a')}_{\mathcal{C}_{2}} A^{\phantom{()}}_{\mathcal C} \P^{(b')}_{\mathcal{C}_{2}}\right) P^{(N,k)} \P^{(b)}_{\mathcal{C}_{1}} \hat{f}_{a+a'-b-b'}.
\end{align*}
It also holds that 
\begin{align}
\sum_{a=0}^{|\mathcal{C}|} \P^{(a)}_{\mathcal{C}}=\prod_{k=1}^{|\mathcal{C}|}(p_{k}+q_{k})=1.
\end{align}
\end{comment}
\end{lem}

\begin{comment}
\begin{proof}
The first statement was shown in  The second one  can be proved by using the decomposition
\begin{align}
P_{N,k}=\sum_{d=0}^{|\mathcal{C}_{1}|}\sum_{d'=0}^{|\mathcal{C}_{2}|} \P^{(d)}_{\mathcal{C}_{1}} \P^{(d')}_{\mathcal{C}_{2}}\P^{(k-d-d')}_{\{1,\ldots,N\}\backslash(\mathcal{C}_{1}\sqcup\mathcal{C}_{2})}
\end{align}
and shifting the $\{1,\ldots,N\}\backslash(\mathcal{C}_{1}\sqcup\mathcal{C}_{2})$-part
to see that for all $a'=0,1,\ldots,|\mathcal{C}_{1}|$ and $b'=0,1,\ldots,|\mathcal{C}_{2}|$
it holds
\[
P_{b'}^{\mathcal{C}_{2}}P_{N,k}P_{a'}^{\mathcal{C}_{1}}=P_{b'}^{\mathcal{C}_{2}}P_{a'}^{\mathcal{C}_{1}}P_{N,k}=P_{N,k}P_{b'}^{\mathcal{C}_{2}}P_{a'}^{\mathcal{C}_{1}}.
\]
\end{proof}
\end{comment}

Furthermore, we need the following statements for concrete choices
of weights $\hat{f}$.

\begin{lem}[$q$-conversion to $\hat{n}$]
\label{lem:q-conversion} Let $\psi\in L_{a}^{2}(\RR^{3N})$. For $n_{0}=1,\ldots,6$ and sufficiently large $N$ 
\begin{align}
\|\prod_{i=1}^{n_{0}+1}q_{i}\psi\|^{2} \le 2 \lsp \psi,(\hat{n})^{n_{0}+1}\psi\rsp.
\end{align}
\end{lem}

\begin{proof}
Since $P^{(N,k)}$ (defined in \eqref{def:counting:functional}) contains exactly $k$ $q$-operators it holds that $\sum_{m=1}^{N}q_{m}P^{(N,k)}=kP^{(N,k)}$
which implies 
\begin{equation}
\hat{n}=N^{-1}\sum_{k=0}^{N}kP^{(N,k)}=N^{-1}\sum_{m=1}^{N}q_{m}\sum_{k=0}^{N}P^{(N,k)}=N^{-1}\sum_{m=1}^{N}q_{m}.
\end{equation}
Thus, the statement is true for $n_{0}=1$. Since $q_1$ is a projection, i.e., $q_1^2 = q_1$, it follows recursively that
\begin{align}
0\leq\lsp\psi,\prod_{i=1}^{n_{0}}q_{i}q_{n_{0}+1}\psi\rsp & =\frac{1}{N-n_{0}}\sum_{m=n_{0}+1}^{N}\lsp\psi,\prod_{i=1}^{n_{0}}q_{i}q_{m}\psi\rsp\nonumber \\
 & =\frac{1}{N-n_{0}}\left(\lsp\psi,\prod_{i=1}^{n_{0}}q_{i}\sum_{m=1}^{N}q_{m}\psi\rsp-n_{0}\langle\psi,\prod_{i=1}^{n_{0}}q_{i}\psi\rangle\right)\nonumber \\
 & \leq\frac{N}{N-n_{0}}\lsp\psi,\prod_{i=1}^{n_{0}}q_{i}N^{-1}\sum_{m=1}^{N}q_{m}\psi\rsp\nonumber \\
 & \leq\frac{(N-n_{0}-1)!N^{n_{0}+1}}{N!}\lsp\psi,\Big(N^{-1}\sum_{m=1}^{N}q_{m}\Big)^{n_{0}+1}\psi\rsp\nonumber \\
 & =\frac{(N-n_{0}-1)!N^{n_{0}}}{(N-1)!}\lsp\psi,(\hat{n})^{n_{0}+1}\psi\rsp.
\end{align}
The desired result thus holds for large $N$, and in the particular case $n_0 \leq 6$, it holds for all $N \geq 33$.
\end{proof}
For the next statement, recall $\ell(k) = \sqrt{k/N}$ and $\hat \ell^{-1} := \widehat {\ell^{-1}}$ on $\text{Ran}(\id -P^{(N,0)})$, and that $q_1 \psi \in \text{Ran}(\id -P^{(N,0)}) $
\begin{lem}[$\ell$-conversion]
\label{lem:l-inverse conversion} Let $\psi\in L_{a}^{2}(\RR^{3N})$. It holds for $n_{0}\in\mathbb{N}$ with $n_{0}<N$ and sufficiently
large $N$ that
\begin{align}
\|\widehat{\ell}^{-1}\prod_{i=1}^{n_{0}}q_{i}\psi\|^{2}\leq2\lsp \psi,(\hat{n})^{n_{0}-1}\psi\rsp.
\end{align}
\end{lem}
\begin{proof}
Note that we can commute $P^{(N,k)}$ and $\prod_{i=1}^{n_{0}}q_{i}$
such that we can apply \prettyref{lem:q-conversion} for $\hat{\ell}^{-1}\psi$
which is still antisymmetric. Since $P^{(N,k)}$ preserves the antisymmetry
we can apply \prettyref{lem:q-conversion}:
\begin{align*}
\|\hat{\ell}^{-1} \prod_{i=1}^{n_{0}}q_{i}\psi\|^{2} & \leq2\lsp\psi,\hat{\ell}^{-1}(\hat{n})^{n_{0}}\hat{\ell}^{-1}\psi\rsp\leq2\lsp\psi,(\hat{n})^{n_{0}-1}\psi\rsp,
\end{align*}
where we used $\hat \ell^{-2} = \hat n^{-1}$ in the last step.
\end{proof}

We further recall the operators $D^{(\gamma)}_{-d}$ and $E^{(\gamma)}_{-d}$ defined in \eqref{eq:D:gamma:def} and \eqref{eq:D:gamma:def}, respectively.
 
\begin{lem}[{Estimates for $\hat{m}^{(\gamma)}$, \cite[Lemma 7.1]{Petrat2016}}]
\label{lem:Weight-estimate-m} Let $d\in \{1,2,3\}$ and let $A_{\mathcal C}$ be an operator acting only on particles with indices $\mathcal{C}\subseteq\{1,2,\ldots,N\}$ with $|\mathcal C|\le 3$. It holds for $a=0,1,\ldots,|\mathcal{C}|-d$ that 
\begin{align}
\left( \hat{m}^{(\gamma)} -\hat{m}^{(\gamma)}_{-d} \right) \P^{(a+d)}_{\mathcal{C}} A_{\mathcal C}^{\phantom{()}} \P^{(a)}_{\mathcal{C}}= \hat{D}_{-d}^{(\gamma)} \left(\P^{(a+d)}_{\mathcal{C}}A_{\mathcal C}^{\phantom{()}} \P^{(a)}_{\mathcal{C}}\right)\hat{E}_{-d}^{(\gamma)}.
\end{align}
Furthermore, the following estimates hold for all $\psi\in L_{a}^{2}(\RR^{3N})$ 
\begin{align*}
\|\hat{D}_{-d}^{(\gamma)} \psi \|^{2} & \leq dN^{-\gamma},\\
\|q_{1}\hat{D}_{-d}^{(\gamma)} \psi \|^{2} & \leq d(d+1)N^{-1}\ \alpha_{m^{(\gamma)}},\\
\|q_{1}q_{2}\hat{D}_{-d}^{(\gamma)} \psi \|^{2} & \leq d(d+1)^{2}N^{\gamma-2}\   \alpha_{m^{(\gamma)}}, 
\end{align*}
and
\begin{align*}
\|\hat{E}_{-d}^{(\gamma)} \psi\|^{2} & \leq dN^{-\gamma},\\
\|q_{1}\hat{E}_{-d}^{(\gamma)} \psi\|^{2} & \leq dN^{-1}\ \alpha_{m^{(\gamma)}},\\
\|q_{1}q_{2}\hat{E}_{-d}^{(\gamma)} \psi\|^{2} & \leq dN^{\gamma-2}\ \alpha_{m^{(\gamma)}},
\end{align*}
where we abbreviate $\alpha_{m^{(\gamma)}} \equiv \langle \psi , \hat m^{(\gamma)} \psi \rangle$.
\end{lem}

\begin{lem}[Estimates for $\hat{w}^{(\gamma)}$]
\label{lem:Weight-estimate-w}It holds for all $\psi\in L_{a}^{2}(\RR^{3N})$,
sufficiently large $N\ge1 $, $n_0 =1,\ldots,6$ and $d=0,1,2,3$ that
\begin{align*}
\| \prod_{i=1}^{n_0} q_i  \hat{w}^{(\gamma)}_{\pm d}\psi\|^{2} & \leq 2 N^{n_0(\gamma-1)}\alpha_{m^{(\gamma)}},
\end{align*}
%\begin{align*}
%\|(\hat n)^{\frac{n_0}{2}}\hat{w}^{(\gamma)}_{\pm d}\psi\|^{2} & \leq N^{n_0(\gamma-1)}\alpha_{m^{(\gamma)}},
%\end{align*}
where we abbreviate $\alpha_{m^{(\gamma)}} \equiv \langle \psi , \hat m^{(\gamma)} \psi \rangle$.
\end{lem}

\begin{proof} First note that  $\| ( q_1q_2 \dots q_{n_0} ) \hat{w}^{(\gamma)}_{\pm d}\psi\|^{2}  \le 2 \|(\hat n)^{\frac{n_0}{2}}\hat{w}^{(\gamma)}_{\pm d}\psi\|^{2} $ by Lemma \ref{lem:q-conversion}. Using \eqref{eq:algebra:property}, one finds that
\begin{align}
\|(\hat n)^{\frac{n_0}{2}}\hat{w}^{(\gamma)}_{-d}\psi\|^{2} & =\sum_{k=d}^{N^{\gamma}}\left(\frac{k}{N}\right)^{n_0}\left(1-\frac{k-d}{N^{\gamma}}\right)^{2}\langle\psi,P^{(N,k)}\psi\rangle\nonumber \\
 & =N^{n_0 (\gamma-1)}\sum_{k=d}^{N^{\gamma}}\left(\frac{k}{N^{\gamma}}\right)^{n}\left(1-\frac{k-d}{N^{\gamma}}\right)^{2}\langle\psi,P^{(N,k)}\psi\rangle\nonumber \\
 & \leq N^{n_0(\gamma-1)}\sum_{k=d}^{N^{\gamma}}\frac{k}{N^{\gamma}}\langle\psi,P^{(N,k)}\psi\rangle\leq N^{n_0(\gamma-1)}\alpha_{m^{(\gamma)}} 
\end{align}
where we used $\max\{x^{n_0-1}(1-x-c)^{2}\mid x\in[c,1],0<c<1\}\leq1$.
Similarly, one verifies
\begin{align}
\|(\hat n)^{\frac{n_0}{2}}\hat{w}^{(\gamma)}_{+d}\psi\|^{2} & =N^{n_0(\gamma-1)}\sum_{k=0}^{N^{\gamma}-d}\left(\frac{k}{N^{\gamma}}\right)^{n_0}\left(1-\frac{k+d}{N^{\gamma}}\right)^{2}\langle\psi,P^{(N,k)}\psi\rangle \notag\\
& =N^{n_0(\gamma-1)}\sum_{k=0}^{N^{\gamma}-d} \frac{k}{N^{\gamma}} \langle\psi,P^{(N,k)}\psi\rangle  \leq N^{n_0(\gamma-1)} \alpha_{m^{(\gamma)}} .
\end{align}

\end{proof}

\subsection{Diagonalization estimates}

The goal of the following sections is to establish the technical estimates used in the proofs of Sections \ref{sec:auxiliary} and \ref{sec:Norm-approximation}. We organize the estimates according to the number of $q$ operators involved. In the case of zero or one $q$ operator, special care is required to exploit cancellations between the microscopic Hamiltonian and the corresponding mean-field term.

The next three lemmas provide the preparations for estimating terms where such cancellations play a crucial role. For instance, terms such as
\begin{align}\label{example}
 \lsp\varphi , \left(  (N-1) \P_{12}^{(1)}  (w_{\nabla f })_{12}  \P_{12}^{(0)} - 2 q_1 \tr_2 (p_2 (w_{\nabla f })_{12} p_2 ) p_1  \right)\psi\rsp
\end{align}
This term, which appears for instance in the estimate for $|\partial_t \alpha_{m^{(\gamma)}}(t)|$, contains only a single $q$ operator, which is not sufficient to close the Grönwall argument. To handle such terms, we follow the strategy introduced in \cite{Petrat2014, Petrat2016}, which relies on diagonalizing operators of the form $p_2 (w_{\nabla f})_{12} p_2$. After subtracting the mean-field contribution, this diagonalization effectively yields an additional $q_2$ operator, which is crucial for obtaining the desired bounds.

Due to the magnetic-like structure of the gauged Hamiltonians, we require substantial generalizations of the diagonalization lemmas used in \cite{Petrat2014,Petrat2016}. Most importantly, we need to diagonalize not only with respect to $p_2$, but also with respect to non-symmetric operators of the form $
p_2^\nabla = \sum_{i=1}^N | \nabla \varphi_i \rangle \langle \varphi_i |_2$, and we need estimates for expressions involving additional gradient operators and three-particle interactions. This is summarized in the following statements, in particular  Lemma \ref{lemma:diaginalization:general}. In Section \ref{sec:0:1:q:diagonalization}, we apply the results from this section to bound all relevant terms of the form \eqref{example}.

\begin{lem} 
\label{lem:diagonalization} Let $h_{12}$ denote a two-body operator on $L^2(\mathbb R^3, dx_1) \otimes L^2(\mathbb R^3,dx_2)$ of the form $h_{12}= \tfrac12 ( \id_1 \otimes B_{2})g_{12}+\tfrac12 g_{12}( \id_1 \otimes B_{2})$
where $g_{12}=g(x_{1}-x_{2})$ is a multiplication operator for some measurable function $g:\mathbb R^3 \to \mathbb R $  and a self-adjoint operator $B$. Let further denote $A$ a self-adjoint operator on $L^2(\mathbb R^3)$ and let $p_{2}^{A}\coloneqq\sum_{i=1}^{N}|A\varphi_{i}\rangle\langle\varphi_{i}|_{2}$ for an orthonormal set $\{ \varphi_i\}_{i=1}^N\subset D(A)$ satisfying 
\begin{align}
\sup_{x_1} | \langle A\varphi_{i},h_{12}A\varphi_{j}\rangle_{2} & (x_{1}) | <\infty \quad \text{for all} \quad 1\le i,j \le N.
\end{align}
For  every $x_{1}\in \mathbb R^3$ there exists an orthonormal
set $\left\{ \chi_{i}^{1}\right\} _{i=1}^N \equiv \left\{ \chi_{i}^{x_{1}}\right\} _{i=1}^N\subset\textnormal{span}(\varphi_{1},\ldots,\varphi_{N})$ 
such that
\begin{align}
\left(p_{2}^{A}\right)^{\ast}h_{12}p_{2}^{A}=\sum_{i=1}^{N}\lambda_{i}(x_{1})|\chi_{i}^{1}\rangle\langle\chi_{i}^{1}|_{2} \quad \text{with} \quad \lambda_{i}(x_{1}) =\langle A\chi_{i}^{1},h_{12}A\chi_{i}^{1}\rangle_{2} (x_1) . 
\end{align}
Moreover, we have $\langle A\chi_{i}^{1},h_{12}A\chi_{j}^{1}\rangle_{2}=0$ for
$i\not=j$ and 
\begin{align}\label{eq:lambda:i:sum}
\sum_{i=1}^{N}\lambda_{i}(x_{1}) & =\sum_{i=1}^{N}\langle A\varphi_{i},h_{12}A\varphi_{i}\rangle_{2}(x_{1}).
\end{align}
\end{lem}
\begin{rem} In the special case that $B_2 = 1$, we obtain
\begin{align}\label{eq:remark:diagonalization}
\sum_{i=1}^N \lambda_i = g \ast \rho^A \quad \text{with} \quad \rho^A \coloneqq \sum_{i=1}^N | A\varphi_i|^2.
\end{align}
\end{rem}

\begin{proof}[Proof of \prettyref{lem:diagonalization}]
The proof is a generalization of \cite[Lemma 6.9]{Petrat2014} where
$B = 1 $, $A = 1$ and $p_{2}^{A} =  p_{2}$ is a projection. Abbreviating $1 \otimes B_2 = B_2$, we write 
\begin{align}
\left(p_{2}^{A}\right)^{\ast}h_{12}p_{2}^{A}
& = \tfrac12 \sum_{i,j=1}^{N}\langle A\varphi_{i},\left(B_{2}g_{12}+g_{12}B_{2} \right)A\varphi_{j}\rangle_{2}(x_{1}) |\varphi_{i}\rangle\langle\varphi_{j}|_{2} .
\end{align}
Now, for fixed $x_1\in \mathbb R^3$, the right -hand side defines a symmetric $N\times N$ matrix. Hence, it can be diagonalized through a unitary $N\times N$ matrix $U(x_{1})$. More precisely, we can write
\begin{equation}
\left(p_{2}^{A}\right)^{\ast}h_{12}p_{2}^{A}=\sum_{i=1}^{N}\lambda_{i}(x_{1})|\chi_{i}^{1}\rangle\langle\chi_{i}^{1}|_{2}
\end{equation}
for
\begin{equation}
|\chi_{i}^{1}\rangle =\sum_{k=1}^{N}U_{ik}(x_{1})|\varphi_{k}\rangle, \quad  i=1,\ldots,N.\label{eq:def:chi:i}
\end{equation}
and
\begin{align}
\lambda_{i}(x_{1}) & =\langle\chi_{i}^{1},\left(p_{2}^{A}\right)^{\ast}h_{12}p_{2}^{A}\chi_{i}^{1}\rangle_{2}(x_{1})\nonumber \\
 & =\sum_{j,k=1}^{N}U_{ij}^{\ast}(x_{1})U_{ik}(x_{1})\langle\varphi_{j}^{1},\left(p_{2}^{A}\right)^{\ast}h_{12}p_{2}^{A}\varphi_{k}^{1}\rangle_{2}(x_{1})\nonumber \\
 & =\sum_{j,k=1}^{N}U_{ij}^{\ast}(x_{1})U_{ik}(x_{1})\langle A\varphi_{j}^{1},h_{12}A\varphi_{k}^{1}\rangle_{2}(x_{1})  =\langle A\chi_{i}^{1},h_{12}A\chi_{i}^{1}\rangle_{2}(x_{1}).
\end{align}
Moreover, using $\chi_i^1 \in \text{Span}(\varphi_1,\ldots, \varphi_N)$ so that $A\chi_i^1 = \sum_{j=1}^N  \langle \varphi_j , \chi_i^1 \rangle  A \varphi_j $, we have 
\begin{align}
\langle A\chi_{k}^{1},h_{12}A\chi_{l}^{1}\rangle_{2}(x_{1}) & =\langle\chi_{k}^{1},\left(p_{2}^{A}\right)^{\ast}h_{12}p_{2}^{A}\chi_{l}^{1}\rangle_{2}(x_{1})\nonumber \\
 & =\sum_{i=1}^{N}\lambda_{i}(x_{1})\langle\chi_{k}^{1},\chi_{i}^{1}\rangle_2 \langle\chi_{i}^{1},\chi_{l}^{1}\rangle_{2} =\lambda_{k}(x_{1})\delta_{k,l}
\end{align}
for all $k,l=1,2,\ldots,N$. Thus, it holds 
\begin{equation}
\sum_{i=1}^{N}\lambda_{i}(x_{1})=\sum_{i=1}^{N}\langle A\chi_{i}^{1},h_{12}A\chi_{i}^{1}\rangle_{2}(x_{1})=\sum_{j=1}^{N}\langle A\varphi_{j},h_{12}A\varphi_{j}\rangle_{2}(x_{1})
\end{equation}
where we used $\sum_{i}U_{ij}^{\ast}(x_{1})U_{ik}(x_{1})=\delta_{jk}$.
\end{proof}

We now apply the previous lemma to estimate operators of the form $p^A_2 g_{12} p^A_2$.

\begin{lem}
\label{lem:diagonalization-estimate} Let $g :\mathbb R^3 \to \mathbb R_0^+$ be measurable, let $A$ be a self-adjoint operator on $L^2(\mathbb R^3)$ and denote $p^A_2 =\sum_{i=1}^N | A \varphi_i \rangle \langle \varphi_i |_2$  and   $\rho^A = \sum_{j=1}^N |A\varphi_j|^2$ for an orthonormal set $\{ \varphi_i\}_{i=1}^N \subset D(A)$, $g_{ij}=g(x_i-x_j)$ and $g_2=g(x_2)$. For every $\psi\in L^2_{a}(\mathbb R^{3N})$ we have
\begin{align}
\langle\psi,\left(p_{2}^{A}\right)^{\ast}g_{2}p_{2}^{A}\psi\rangle & \le  \frac1N \| g \rho^A  \|_1 \langle \psi , \psi \rangle
\end{align}
and for all $\psi\in L^2(\mathbb R^{3N})$  that are antisymmetric in $x_2,\ldots, x_N$ we have
\begin{align}
\langle\psi,\left(p_{2}^{A}\right)^{\ast}g_{12}p_{2}^{A}\psi\rangle & \leq\frac{1}{N-1}\|g \ast\rho^{A}\|_{\infty}\langle\psi,\psi\rangle, \label{eq:diag:bound:pgp} \\
\langle\psi,p_{1}\left(p_{2}^{A}\right)^{\ast}g_{12}p_{2}^{A}p_{1}\psi\rangle & \leq \frac{1}{N(N-1)}\|( g \ast\rho^{A})\ \rho\|_{1}\langle\psi,\psi\rangle. %\\
%\langle\psi,q_{3}\left(p_{1}^{A}\right)^{\ast}p_{2}g_{12}g_{13}p_{1}^{A}p_{3}q_{2}\psi\rangle & \leq\frac{1}%{(N-1)(N-2)}\|g \ast\rho\|_{\infty}\|g \ast\rho^{A}\|_{\infty}\langle\psi,q_{1}\psi\rangle
\end{align}
\end{lem}

\begin{rem}
Let us note that \prettyref{lem:diagonalization} applies to a broader
class of two-body operators $h_{12}$, whereas \prettyref{lem:diagonalization-estimate}
is valid only for non-negative multiplication operators. In practice, we will always take the absolute values before applying \prettyref{lem:diagonalization-estimate}. 
\end{rem}

\begin{proof}
The proofs work similar to \cite[Lemmas 6.10--6.12]{Petrat2014}
where $A\equiv \id $. For the first bound, we apply \prettyref{lem:diagonalization} with $g \equiv  1$ (the constant function) and $B_2 \equiv  g_2$. Hence, we can write $ (p_{2}^{A} )^{\ast}g_{2}p_{2}^{A} = \sum_{i=1} \lambda_i |\chi_i \rangle \langle \chi_i|_2$ with $\lambda_i = \langle A \chi_i, g A \chi_i\rangle \ge 0$ and $\sum_{i=1}^N \lambda_i = \| g \rho^A \|_1 $ by \eqref{eq:lambda:i:sum}. This gives 
\begin{align}
\langle\psi,\left(p_{2}^{A}\right)^{\ast} g_{2}p_{2}^{A}\psi\rangle =  \sum_{i=1}^N \lambda_i \langle \psi, | \chi _i \rangle \langle \chi_i |_2 \psi  \rangle  \le \sup_{i=1,\ldots ,N} \langle \psi, | \chi _i \rangle \langle \chi_i |_2 \psi  \rangle   \| g \rho^A \|_1.
\end{align}
Moreover,  using antisymmetry of $\psi$,
\begin{align}
\langle \psi, | \chi _i \rangle \langle \chi_i |_2 \psi  \rangle   = \frac{1}{N} \sum_{m=1}^N \langle \psi, | \chi _i \rangle \langle \chi_i |_m \psi  \rangle \le \frac1N \| \psi \|_2
\end{align}
where we used that $\sum_{m=1}^N | \chi _i \rangle \langle \chi_i |_m$ is an orthogonal projection on $L^2_{a}(\mathbb R^{3N})$. This proves the first statement.

We proceed similarly for the second statement, namely 
\begin{align}
\langle\psi,p_{2}^{A}h_{12}p_{2}^{A}\psi\rangle=\sum_{i=1}^{N}\lambda_{i}(x_{1})\langle\psi,|\chi_{i}^{1}\rangle\langle\chi_{i}^{1}|_{2}\psi\rangle\le \| g \ast \rho^A \|_\infty \sup_{i}  \langle\psi,|\chi_{i}^{1}\rangle\langle\chi_{i}^{1}|_{2}\psi\rangle
\end{align}
by means of \prettyref{lem:diagonalization}, $\lambda_{i}\geq 0$ for all $i=1,\ldots,N$ and \eqref{eq:remark:diagonalization}. Using antisymmetry of $\psi$ in $x_2,\ldots x_N$, we then bound 
\begin{align}
\langle\psi,|\chi_{i}^{1}\rangle\langle\chi_{i}^{1}|_{2}\psi\rangle  & =\frac{1}{N-1}\langle\psi,\sum_{m=2}^{N}|\chi_{i}^{1}\rangle\langle\chi_{i}^{1}|_{m}\psi\rangle\leq \frac{1}{N-1} \| \psi \|^2.
\end{align}
\end{proof}

Based on the previous two lemmas, we can now prove the following general statement that is used in the next section to estimate all relevant terms of the form \eqref{example}.

\begin{lem} \label{lemma:diaginalization:general}Under the same assumptions as in \prettyref{lem:diagonalization}, consider $h_{12}= \tfrac12 ( \id_2 \otimes B_{2})g_{12}+\tfrac12 g_{12}( \id_1 \otimes B_{2})$ and let $r_1 \in \{ q_1 , p_1 \}$. Then, there exists a constant $C>0$ such that for all $\psi ,\varphi \in L^2_{a}(\mathbb R^3)$
\begin{align}
& \left| \lsp\varphi,   p_1^A  \big((N-1)  p_2 h_{12 }  p_2 -    \tr_2(p_2 h_{12}p_2) \big) r_1 \psi\rsp\right|\notag \\
& \qquad \leq C N^{-\frac12}  \| |g| \ast  \rho^A \|_{\infty}^{1/2}  \| |g| \ast \rho^B \|_{\infty}^{1/2}   \|\varphi \|\left( N \|r_{1}q_{2}\psi\|^2 +\|r_{1}\psi\|^2 \right)^{1/2} . \label{eq:0q:general:bound:1}
\end{align}
Moreover, let $r_2\in \{q_2,p_2\}$, $g_{21}=g(x_2-x_1)$ and $g_{23}=g(x_2-x_3)$. There exists a constant $C>0$ such that all  $\psi ,\varphi \in L^2_{a}(\mathbb R^3)$
\begin{align}
&\left| \lsp\varphi ,p_{1}p_{2}g_{21} \left( (N-2)p_{3} g_{23}  p_{3} - \tr_3(p_3 g_{23} p_3) \right)r_1 r_2 \psi\rsp \right| \notag\\
&\qquad \qquad \qquad \qquad \le C N^{-1} \| \rho \|_1^{3/2}  \| \varphi\| \left( N \|r_{1}r_{2} q_3 \psi\|^2 +\|r_{1}r_2 \psi\|^2 \right)^{1/2} .\label{eq:0q:general:bound:2}
\end{align}
\end{lem}
\begin{proof} Using \prettyref{lem:diagonalization}, we can write for $m\ge 2$
\begin{align}
\sum_{m=2}^N p_m h_{1m} p_m - \tr_2(p_2 h_{12} p_2) =  \sum_{i=1}^{N}\lambda_{i}(x_{1}) \Big( \sum_{m=2}^N |\chi_{i}^{1}\rangle\langle\chi_{i}^{1}|_{m} - 1 \Big)\label{eq:force-diag-form}
\end{align}
with $\lambda_{i}(x_{1})= \langle \chi^1_i, ( B_2 g_{12} + g_{12}B_2) \chi_i^1 \rangle_2(x_1)$, where  $\{ \chi_i^1\}_{i=1}^N \subset L^2(\mathbb R^3)$ is an ($x_1$-dependent)  orthonormal set, defined as in \eqref{eq:def:chi:i}, that satisfies in particular $\sum_{i=1}^N|\chi_{i}^{1}\rangle\langle\chi_{i}^{1}|_{m} = p_m$.  After symmetrizing in $x_2$, we use the above identity and Cauchy-Schwarz in the scalar product as well as in the sum over $i$ to estimate
\begin{align}
 & \left| \lsp\varphi ,p_{1}^A\Big((N-1)p_{2} h_{12}p_{2}- \tr_2 (p_2 h_{12} p_2) \Big) r_{1}\psi\rsp\right|  \notag \\
 & = \left|\lsp\varphi ,p_{1}^A  \left(\sum_{m=2}^{N}p_{m} h_{1m} p_{m} - \tr_2 (p_2 h_{12} p_2) \right)r_{1}\psi\rsp\right|\nonumber \\
 & = \left|\lsp \varphi ,p_{1}^A \left(\sum_{i=1}^{N}\lambda_{i}(x_{1})\Big(\sum_{m=2}^{N}|\chi_{i}^{1}\rangle\langle\chi_{i}^{1}|_{m}-1\Big)\right)r_{1}\psi\rsp\right|\nonumber \\
 & \leq \left(\sum_{i=1}^{N}\||\lambda_{i}(x_{1})|p_{1}^{A}\varphi\|^{2}\right)^{1/2}\left(\sum_{i=1}^{N}\|\Big(\sum_{m=2}^{N}|\chi_{i}^{1}\rangle\langle\chi_{i}^{1}|_{m}-1\Big)r_{1}\psi\|^{2}\right)^{1/2}\label{eq:diag-estimate-0q-0}. 
\end{align}
Here we used in the last step that the operator $ q_{\neq 1}^{\chi_i^{1}}:=  1-\sum_{m=2}^{N}|\chi_{i}^{1}\rangle\langle\chi_{i}^{1}|_{m} $ satisfies
\begin{align}
\left( q_{\neq 1}^{\chi_i^{1}}  \right)^2 \psi^1 = q_{\neq 1}^{\chi^{1}_i} \psi^{1}
\end{align}
for all $\psi^{1}\in L^2(\mathbb R^{3N})$ that are antisymmetric in
all variables except $x_{1}$. This follows from the fact that the $\{\chi_{i}^1\}_{i=1}^{N}$ are orthonormal set.  Using $\sum_{i=1}^N |\chi_{i}^1 \rangle \langle \chi_i^1|_m  = p_m$ and denoting $\psi^1 = r_1 \psi$, the second factor in \eqref{eq:diag-estimate-0q-0} can be further estimated as
\begin{align}
\lsp\psi^{1},\sum_{i=1}^{N}  q_{\neq 1}^{\chi^{1}_i} \psi^{1}\rsp =\lsp\psi^{1},\Big(N-\sum_{m=2}^{N}p_{m}\Big)\psi^{1}\rsp & =\lsp\psi^{1},\Big(N-(N-1)p_{2}\Big)\psi^{1}\rsp\nonumber \\
 & =(N-1)\langle\psi^{1},q_{2}\psi^{1}\rangle+\langle\psi^{1},\psi^{1}\rangle.\label{eq:p-chi-bound} 
\end{align}
For the first factor,  we apply the first statement of \prettyref{lem:diagonalization-estimate} for $A=1$ and $g = |\lambda_i|^2$, 
\begin{align}
\| |\lambda_{i}(x_{1}) | p_1^A  \psi \|^{2} & \leq N^{-1}  \| | \lambda_i |^2 \rho^A \|_1,
\end{align}
and then use Cauchy-Schwarz to estimate
\begin{align}
|\lambda_i (x_{1})|^2 \le  \langle B \chi_i^1, |g_{12} | B \chi_i^1 \rangle_2(x_1)  \langle  \chi_i^1, |g_{12} |  \chi_i^1 \rangle_2(x_1) .
\end{align}
With this, we can proceed as follows: 
\begin{align}
\sum_{i=1}^N \| | \lambda_i|^2 \rho^A \|_1& = \int dx_1 \sum_{i=1}^N |\lambda_i(x_1)|^2 \rho^A (x_1)\notag\\
&  \le \int dx_1 \Big( \sum_{i=1}^N \langle B \chi_i^1,| g_{12} | B  \chi_i^1 \rangle_2(x_1) \Big) \Big(\max_{i}  \langle \chi_i^1, |g_{12}|  \chi_i^1 \rangle_2(x_1)   \Big) \rho^A (x_1)\notag\\
&\le \sup_{x_1}  \Big(\sum_{i=1}^N\langle B \chi_i^1,| g_{12} | B \chi_i^1 \rangle_2(x_1) \Big)  \max_{i} \int dx_1  \langle \chi_i^1, |g_{12}| \chi_i^1 \rangle_2(x_1)  \rho^A (x_1),
\end{align}
where we get
\begin{align}
 \sum_{i=1}^N\langle B \chi_i^1,| g_{12} |  B \chi_i^1 \rangle_2(x_1)  =    \sum_{i=1}^N\langle B \varphi_i ,| g_{12} |B \varphi_i \rangle_2(x_1)  = |g|\ast \rho^B(x_1),
\end{align}
as well as, with Fubini and $\| \chi_i^1 \|_2 = 1$,
\begin{align}
 \int dx_1  \langle \chi_i^1, |g_{12}| \chi_i^1 \rangle_2(x_1)  \rho^A(x_1) & \le \| |g| \ast \rho^A \|_\infty.
\end{align}
Combining the above, we thus find 
\begin{align}
\sum_{i=1}^N \| | \lambda_i(x_1)| p_1^A \psi \|^2 \le N^{-1} \| |g|\ast \rho^B \|_\infty \| |g| \ast \rho^A \|_\infty
\end{align}
This proves the bound in \eqref{eq:0q:general:bound:1}.

The proof of the second statement works similarly. After symmetrizing in $x_3$ we use Lemma \ref{lem:diagonalization} to write
\begin{align}
\sum_{m=3}^N p_{m} g_{2m}  p_{m} - \tr_3(p_3 g_{23} p_3) = \sum_{i=1}^N \mu_i(x_2) |\chi_i^2\rangle \langle \chi_i^2 |_m
\end{align}
with $\mu_i(x_2) = \langle \chi_i^2 , g_{23} \chi_i^2 \rangle_3(x_3)$.
and the follow the analogous steps to Inequality \eqref{eq:diag-estimate-0q-0}. Thus
\begin{align}
&\left| \lsp\varphi ,p_{1}p_{2}g_{21} \left( (N-2)p_{3} g_{23}  p_{3} - \tr_3(p_3 g_{23} p_3) \right) r_1 r_2  \psi\rsp \right| \notag\\
 & = \left|\lsp\varphi ,p_{1} p_2 g_{21}  \left(\sum_{m=3}^{N}p_{m} g_{2m} p_{m} - \tr_3 (p_3 h_{23} p_3) \right) r_{1} r_2 \psi\rsp\right|\nonumber \\
 & = \left|\lsp \varphi ,p_{1} p_2 g_{12}\left(\sum_{i=1}^{N} \mu_{i} (x_{2})\Big(\sum_{m=3}^{N} |\chi_{i}^{2}\rangle\langle\chi_{i}^{2}|_{m} - 1\Big)\right) r_{1} r_2 \psi\rsp\right|\nonumber \\
 & \le \left( \sum_{i=1}^N \| | \mu_i(x_2) |\, |g_{12}| p_1 p_2 \psi \|^2 \right)^{1/2} \left( (N-2) \| q_3 r_1 r_2 \psi \|^2 + \| r_1 r_2 \psi \|^2 \right)^{1/2} .
\end{align}
The first factor is further estimated using \prettyref{lem:diagonalization-estimate} for $p^A_2 = |\mu_i (x_2)|p_2$: 
\begin{align}
\| | \mu_i(x_2)|\, |g_{12}| p_1 p_2 \psi \|^2 & \le \frac{1}{N(N-1)} \| (|g|^2 \ast |\mu_i |^2 \rho)  \rho  \|_1 \le C N^{-2} \| | \mu_i |^2 \rho \|_1 \| \rho \|_1
\end{align}
as well as 
\begin{align}
\sum_{i=1}^N \| | \mu_i|^2 \rho \|_1 &\le \sup_{x_2}  \Big(\sum_{i=1}^N\langle  \chi_i^2, | g_{23} |  \chi_i^2 \rangle_3(x_2) \Big)  \max_{i} \int dx_2  \langle \chi_i^2, |g_{23}| \chi_i^2 \rangle_3(x_2)  \rho (x_2)
\end{align}
so that one finds
\begin{align}
\sum_{i=1}^N \| | \mu_i|^2 \rho \|_1\le  \| |g| \ast \rho \|_\infty^2.
\end{align}
Thus, combined
\begin{align}
\sum_{i=1}^N 
\| | \mu_i(x_2)|\, |g_{12}| p_1 p_2 \psi \|^2 \le C N^{-2}  \| \rho \|_1^3
\end{align}
which completes the proof of the second statement
\end{proof}

In addition, we can use the same strategy as in the proof of Lemma \ref{lemma:diaginalization:general}
to treat diagonalization estimates of one-body operators.
\begin{lem}
\label{lem:hg-0q-estimate} Let $h$ be a self-adjoint operator on $L^2(\mathbb R^3)$ with $\| h p \| < \infty$. There is a constant $C>0$ such that for
all $\psi\in\LaR$ 
\begin{align}
& \left|\lsp\psi,\big(N p_{1}  h_{1}p_{1}-\tr(p  h p )\big)\psi\rsp\right| \le C \| \rho^h \|_1^{1/2}  \, \left( N \| q_1 \psi \|^2 + \| \psi \|^2 \right)^{1/2} 
\end{align}
with $\rho^{h} : = \sum_{k=1}^N | h \varphi_k^t |^2$.
\end{lem}
\begin{proof} The proof works in complete analogy to the proof of Lemma \ref{lemma:diaginalization:general}.
\end{proof}

\subsection{$0q$- and $1q$-estimates with diagonalization}\label{sec:0:1:q:diagonalization}

We now apply Lemma~\ref{lemma:diaginalization:general} to inner products involving zero or one $q$-operator. Here, we exploit cancellations between the interaction operators that appear in the guaged Hamiltonian and the corresponding mean-field operators, defined as
\begin{align}\label{eq:R:W:tollbox}
R_1 := \tr_2 \left( p_2 (w_{\nabla f})_{12} p_2 \right), \qquad
W_1 := \frac{1}{2} \tr_{2,3} \left( p_2 p_3 (w_{ff})_{123} p_2 p_3 \right).
\end{align}
These operators are analogous to those in \eqref{eq:R-def}, but now defined with respect to the projections \eqref{eq:p:toolbox} for a general set of orbitals $\varphi_1, \ldots, \varphi_N$.

Throughout this and the following sections, many of the upper bounds are given in terms of the $L^1$-norms of the functions $\rho^\nabla=\sum_{j=1}^{N}|\nabla \varphi_j|^2$ and $\rho^\Delta=\sum_{j=1}^{N}|\Delta \varphi_j|^2$.

\begin{lem}
\label{lem:0q-diag-estimate} Let $r_1^{(0)} = p_1$ and $r_1^{(1)} = 2 q_1 $. There is a constant $C>0$ such that  
for $i\in \{0,1\}$ and all $\varphi,\psi\in L_{a}^{2}(\RR^{3N})$
\begin{align}
&  \left| \lsp\varphi,  \big((N-1) \P^{(0)}_{12}   (w_{\nabla f})_{12}^{\phantom{()}} \P^{(i)}_{12}  -p_{1}  R_{1} r_{1}^{(i)}\big)\psi\rsp\right|\notag \\
&\hspace{4cm}  \leq C  (N^\frac12 + \| \rho^{\nabla}\|_{1}^{1/2}) \|\varphi \|\left( N \|r_{1}^{(i)}q_{2}\psi\|^2 +\|r_{1}^{(i)} \psi\|^2 \right)^{1/2}\label{eq:0:q:estimates:w:nabla:f}
\end{align}
Moreover, let $s_1^{(0)} = p_1$ and $s_1^{(1)} = 3q_1$. There is a constant $C>0$ such that for $i\in \{0,1\}$ and all $\varphi,\psi\in L_{a}^{2}(\RR^{3N})$
\begin{align}
& \left|\lsp\varphi ,\big((N-1)(N-2)  \P^{(0)}_{123} (w_{ff})_{123}^{\phantom{()}}  \P^{(i)}_{123}-2 p_{1}W_{1}s_{1}^{(i)}\big)\psi\rsp \right| \notag\\
 &\hspace{5cm} \leq C N^{\frac{3}{2}} \| \varphi \|\left( N\| s^{(i)}_{1}q_{2}\psi\|^2 +\|s_{1}^{(i)}\psi\|^2 \right)^{1/2} .\label{eq:0:q:estimates:w:f:f} 
\end{align}
\end{lem}
\begin{proof} For the first statement, we recall $\P^{(0)}_{12} = p_1 p_2$, $\P^{(1)}_{12} = q_1 p_2 + p_1 q_2 $ and $(w_{\nabla f})_{12} =2  i\nabla_{1} \cdot f_{12}+  ( i\nabla_1 \cdot f_{12} ) + i\nabla_{2} \cdot f_{21}+f_{21}\cdot i\nabla_{2}$ to write
\begin{align}
&  \lsp\varphi,\left((N-1) \P^{(0)}_{12} (w_{\nabla f})_{12}^{\phantom{()}} \P^{(i)}_{12} - p_{1}R_{1}r_1^{(i)} \right)\psi\rsp \notag\\
 &\quad = \lsp\varphi,\left((N-1) p_1 p_2 (w_{\nabla f})_{12}  p_2 r_1^{(i)}  - p_{1}R_{1} r_1^{(i)} \right)\psi\rsp  = (\rm{i}) + (ii) + (iii) 
\end{align}
with
\begin{align}
  {\rm{(i)}}  &  = 2 \lsp\varphi,p_{1}\Big((N-1)p_{2}i\nabla_{1} \cdot f_{12}p_{2}-i\nabla_{1}\cdot \tr_2 \left( p_2 f_{12} p_2 \right) \Big)r^{(i)}_{1}\psi\rsp  \notag \\
  {\rm{(ii)}}  &   =  \lsp\varphi,p_{1}\Big((N-1)p_{2}( i\nabla_1 \cdot f_{12}) p_{2}-  \tr_2\left( p_2 (i \nabla_1 \cdot  f_{12} ) p_2 \right) \Big)r^{(i)}_{1}\psi\rsp \nonumber \\
   {\rm{(iii)}} &  =  \lsp\varphi,p_{1}\Big((N-1)p_{2}\left((i\nabla_{2})\cdot f_{21}+f_{21}\cdot(i\nabla_{2})\right)p_{2} -\tr_2 \left( p_2 ( i \nabla_2 \cdot f_{21} +  f_{21} \cdot i\nabla_2 ) p_2 \right)    \Big)r^{(i)}_{1}\psi\rsp \notag
\end{align}

Written this way, each term can be estimated using Lemma \ref{lemma:diaginalization:general}. For the first term, we choose $A \equiv i\nabla$ and $h_{12} \equiv  f_{12}$ (that is, we set $g_{12} \equiv f_{12}$ and $B \equiv \id$). Strictly speaking, we apply Lemma \ref{lemma:diaginalization:general} to $A \equiv i\nabla_1^{(\ell)}$ and $h_{12} \equiv  f_{12}^{(\ell)}$ for $\ell \in \{1,2,3\}$ separately, but we do not write this explicitly (also for the following terms). For the second term, we use $A \equiv \id$ and $h_{12} \equiv (i\nabla_1 \cdot f_{12})$ (that is, $g_{12} \equiv (i\nabla_1 \cdot f_{12})$ and $B \equiv \id$). For the third line, we employ Lemma \ref{lemma:diaginalization:general} with $A \equiv \id $ and $h_{12} \equiv  i\nabla_2 \cdot f_{21} + f_{21} \cdot i \nabla_2$ (that is, $g_{12} \equiv f_{21}$ and $B \equiv  i\nabla$).

This leads to
\begin{align}
| {\rm{(i)}}| + | {\rm{(iii)}} | \le C  N^{-1/2}  \|  \rho^\nabla \|_{1}^{1/2}  \| \rho \|_1^{1/2}   \|\varphi \|\left( N \|r^{(i)}_{1}q_{2}\psi\|^2 +\|r^{(i)}_{1}\psi\|^2 \right)^{1/2} ,
\end{align}
and similarly for $ {\rm{(ii)}}$ but with $\|  \rho^\nabla \|_{1}^{1/2}$ replaced by $\|  \rho \|_{1}^{1/2}$ on the right side. Invoking $\| \rho \|_1 = N$, this proves Inequality \eqref{eq:0:q:estimates:w:nabla:f}.

Recalling $(w_{ff})_{123} = 2f_{12}\cdot f_{13}+2f_{21}\cdot f_{23}+2f_{31}\cdot f_{32} $ and writing 
\begin{align}
W_1 & = \overline{f}(x_1) \cdot \overline{f}(x_1) + 2  \tr_{2} (p_2 f_{21}\cdot  \overline{f(x_2)}  p_2 ) 
\end{align}
where we abbreviate $\overline f(x_1) := \tr_2 (p_2 f_{12} p_2)$, we can decompose
\begin{align}
 & \lsp\psi,\left((N-1)(N-2) \P^{(0)}_{123}(w_{ff})_{123}^{\phantom{()}} \P^{(i)}_{123} -2p_{1}W_{1}s_1^{(i)} \right)\psi\rsp \notag\\
 & \quad = \lsp\psi,\left((N-1)(N-2)  p_1 p_2 p_3 (w_{ff})_{123}^{\phantom{()}}  p_2 p_3 s_1^{(i)} -2p_{1}W_{1}s_1^{(i)} \right)\psi\rsp  = \rm{(iv)} + (v) + (vi) +(vii)\notag
 \end{align}
 with 
  \begin{align}
 \rm{(iv)} & = 4 \lsp\psi,p_{1}p_{2}\left((N-1)(N-2)p_{3} ( f_{21}\cdot f_{23} ) p_{3}-(N-1)f_{21}\cdot  \overline{f}(x_{2})\right)p_{2} s_1^{(i)} \psi\rsp \nonumber \\
\rm{(v)} & = 4  \lsp\psi,p_{1}\left( (N-1)p_{2}f_{21}\cdot\overline{f}(x_{2})p_{2}-   \tr_{2} (p_2 f_{21}\cdot  \overline{f(x_2)}  p_2 )   \right) s_1^{(i)} \psi\rsp \nonumber \\
 \rm{(vi)} & =2 \lsp\psi,p_{1}p_{2}\left((N-1)(N-2)p_{3}\left(f_{12}\cdot f_{13}\right)p_{3}-(N-1)f_{12}\cdot\overline{f}(x_{1})\right)p_{2} s_1^{(i)} \psi\rsp \nonumber \\
 \rm{(vii)} & = 2 \lsp\psi,p_{1}\left((N-1) p_{2}f_{12} \cdot \overline{f}(x_1)  p_{2} -\overline{f}(x_{1})\cdot\overline{f}(x_{1})\right) s_1^{(i)} \psi\rsp \notag
\end{align}
% \begin{align}
% \rm{(iv)} & = 4 \lsp\psi,p_{1}p_{2}\left((N-1)(N-2)p_{3} ( f_{21}\cdot f_{23} ) p_{3}-(N-1)f_{21}\cdot\overline{f}(x_{2})\right)p_{2} s_1^{(i)} \psi\rsp \nonumber \\
%\rm{(v)} & = 4  \lsp\psi,p_{1}\left((N-1)p_{2}f_{21}\cdot\overline{f}(x_{2})p_{2}-\overline{f\cdot\overline{f}}(x_{1})\right) s_1^{(i)} \psi\rsp \nonumber \\
% \rm{(vi)} & =2 \lsp\psi,p_{1}p_{2}\left((N-1)(N-2)p_{3}\left(f_{12}\cdot f_{13}\right)p_{3}-(N-1)f_{12}\cdot\overline{f}(x_{1})\right)p_{2} s_1^{(i)} \psi\rsp \nonumber \\
% \rm{(vii)} & = 2 \lsp\psi,p_{1}\left((N-1) p_{2}f_{12} \cdot \overline{f}(x_1)  p_{2} -\overline{f}(x_{1})\cdot\overline{f}(x_{1})\right) s_1^{(i)} \psi\rsp \label{eq:w_ff-triangle}
%\end{align}
Note that here we added and subtracted the terms involving $f_{21} \cdot \overline{f}(x_2)$ and $ f_{12} \cdot \overline f (x_1)$.

For the second line we apply Lemma \ref{lemma:diaginalization:general} for $A \equiv \id$ and $h_{12} \equiv f_{12} \cdot \overline{f}(x_2)$ (i.e. $g_{12} \equiv f_{12}$ and $B_2 \equiv \overline f(x_2)$). 
For the last line we apply Lemma \ref{lemma:diaginalization:general} with $A \equiv \overline{f}$ and $h_{12} \equiv f_{21}$ (i.e. $g_{12} \equiv f_{21}$ and $B \equiv \id$. Since $\| \rho^{\overline f} \|_1 \le \| |f| \ast \rho \|_\infty  \| \rho \|_1\le \| \rho \|_1^2$, this gives
\begin{align}
|{\rm (v)} | + |{\rm (vii)} | \le C N^{-1/2} \| \rho \|_1^2 \| \varphi \| \left((N-1)\|q_{2} s^{(i)}_1 \psi\|^{2}+\| r_1 \psi\|^{2}\right)^{1/2}.
\end{align}
For the first line we apply \eqref{eq:0q:general:bound:2} with $g_{21}= f_{21}$ and $g_{23} = f_{23}$ so that 
\begin{align}
|{\rm (iv)}| \le  \| \rho \|_1^{3/2}  \| \varphi \| \left((N-1)\|q_{2} r_1 \psi\|^{2}+\| s^{(i)}_1 \psi\|^{2}\right)^{1/2}.
\end{align}
The remaining line ${\rm (vi)}$ is treated in analogy with the same upper bound. Invoking again $\| \rho \|_1= N$, this proves Inequality \eqref{eq:0:q:estimates:w:f:f}.
\end{proof}

\subsection{$0q$- and $1q$-estimates without diagonalization}

In this section, we prove estimates for terms containing zero or one $q$ operators, without using the diagonalization lemmas, i.e., without exploiting cancellations between microscopic terms and mean-field terms. 

\begin{lem}
\label{lem:1q-estimate} There is a constant $C>0$, such that for all $\varphi,\psi\in L_{as}^{2}(\RR^{3N})$
\begin{align}
\left |\lsp\varphi, \P^{(0)}_{12}(w_{f})_{12}^{\phantom{()}} \P^{(0)}_{12}\psi\rsp\right| & \leq C \|\varphi\|\ \|\psi\|, \\[1mm] 
\left|\lsp\varphi, \P^{(0)}_{12}(w_{f})_{12}^{\phantom{()}} \P^{(1)}_{12}\psi\rsp\right| & \leq C  \|\varphi\|\ \|q_1 \psi\|, \label{eq:P0:P1:bound:example}\\[1mm]
\left|\lsp\varphi, \P^{(0)}_{12}(w_{\nabla f})_{12}^{\phantom{()}} \P^{(1)}_{12}\psi\rsp\right| & \leq C   ( 1 + N^{-\frac12} \| \rho^\nabla_t \|_1^\frac12  )  \|\varphi\|\ \|q_{2}\psi\|, \\[1mm]
 \left |\lsp\varphi, \P^{(0)}_{123}(w_{ff})_{123}^{\phantom{()}} \P^{(1)}_{123} \psi \rsp \right |
 & \leq C \|\varphi\|\ \|q_{1}\psi\|.
\end{align}
\end{lem}
\begin{proof}
Except for the third line, the estimates follow directly from Cauchy-Schwarz and the assumption that $f$ is bounded.

In the third line, we use $(w_{\nabla f})_{21}=(w_{\nabla f})_{12}$ to write
\begin{align}
\lsp\varphi, \P^{(0)}_{12} (w_{\nabla f})_{12}^{\phantom{()}} \P^{(1)}_{12}\psi\rsp = \lsp\varphi, p_1 p_2 (w_{\nabla f})_{12} q_1 p_2 \psi\rsp .
\end{align}
Since we must avoid $i\nabla_1$ acting on $q_1$, we write $(w_{\nabla f})_{12} =  2 i\nabla_1 \cdot  f_{12} + ( i\nabla_1 \cdot f_{12} ) + (i\nabla_{2})\cdot f_{21}+f_{21}\cdot(i\nabla_{2})$. Applying \eqref{eq:diag:bound:pgp} for the constant function $g\equiv 1$ and $A\equiv i\nabla$, we obtain $ \langle \varphi , p_1 (-\Delta_1) p_1 \varphi \rangle \le N^{-1} \| \rho^\nabla \|_1 \| \varphi\|^2$ and $ \langle q_1 \psi , p_2 (-\Delta_2) p_2 q_1 \psi \rangle \le N^{-1} \| \rho^\nabla \|_1 \| q_1 \psi \|^2$, and it follows straightforwardly that
\begin{align}
\left| \lsp \varphi, p_1 p_2 (w_{\nabla f})_{12} q_1 p_2 \psi\rsp \right| \le C (1+ N^{-\frac12} \| \rho^\nabla \|_1^{\frac12} ) \| \varphi\| \| q_1 \psi \| 
\end{align}
which completes the proof of the lemma.
\end{proof}

\subsection{$nq$-estimates for $n\ge2$}
\label{sec:q:estimtes:n>2}
For the terms involving two $q$ operators, we state two different lemmas depending on whether the projections $\P^{(a)}_{12}$ resp. $\P^{(a)}_{123}$ appearing in the inner product are identical (symmetric terms) or distinct (asymmetric terms).

\begin{lem}
\label{lem:2q-estimate-sym} There is a $C>0$ such that for all $\varphi,\psi\in L_{a}^{2}(\RR^{3N})$ 
\begin{align}
\left|\lsp\varphi, \P^{(1)}_{12} (w_{\nabla f})_{12}^{\phantom{()}} \P^{(1)}_{12}\psi\rsp\right| & \leq C  N^{-\frac12} \| \rho^\nabla \|_1^\frac12   \|q_{2}\varphi\|\ \|q_{2}\psi\| \notag \\
 & \quad+ C \left(\|\nabla_{1}q_{1}\varphi\|\ \|q_{2}\psi\|+\|q_{2}\varphi\|\ \|\nabla_{1}q_{1}\psi\|\right), \\
 \left|\lsp\varphi, \P^{(1)}_{12}(w_{f})_{12}^{\phantom{()}} \P^{(1)}_{12} \psi\rsp\right| & \leq C \|q_{1}\varphi\|\ \|q_{2}\psi\|, \\
 \left|\lsp\varphi, \P^{(1)}_{123}(w_{ff})_{123}^{\phantom{()}} \P^{(1)}_{123} \psi\rsp \right| & \leq C \|q_{1}\varphi\|\ \|q_{2}\psi\|.
\end{align}
\end{lem}
\begin{proof} In the first line, we estimate
\begin{align}
 & \left|\lsp\varphi,\P^{(1)}_{12}(w_{\nabla f})_{12}^{\phantom{()}} \P^{(1)}_{12} \psi\rsp \right|\nonumber \\
 & = 2 \big|\lsp\varphi,(q_{1}p_{2}+p_{1}q_{2})\left((i\nabla_{1})\cdot f_{12}+f_{12}\cdot(i\nabla_{1})\right)(q_{2}p_{1}+p_{2}q_{1})\psi\rsp\big| \nonumber \\[1mm]
 & \le C ( \|  \nabla_1 q_1 \varphi \| + \| \nabla_1 p_1 q_2 \varphi \|) \| q_2 \psi \| + C \| q_2 \varphi \| ( \| \nabla_1 p_1  q_2 \psi \| + \| \nabla_1 q_1  \psi  \| ) \notag\\
 & \le C ( \|  \nabla_1 q_1 \varphi \| + N^{-\frac12} \| \rho^\nabla \|_1^\frac12  \| q_2 \varphi \| ) \| q_2 \psi \|  + C \| q_2 \varphi \|   ( N^{-\frac12} \| \rho^\nabla \|_1^\frac12 \| q_2 \psi \|   + \| \nabla_1 q_1  \psi  \| )
\end{align}
where we used $ \| \nabla_1 p_1  q_2 \psi \| \le N^{-1/2} \| \rho^\nabla \|_1^{1/2} \| q_2 \psi \|$ and the same for  $\varphi$, by \eqref{eq:diag:bound:pgp}. This proves the first inequality of the lemma. 
The second and third lines are straightforward.
\end{proof}

\begin{lem} 
\label{lem:2q-estimate-asym} There is a $C>0$ such that for all $\varphi,\psi\in L_{a}^{2}(\RR^{3N})$ 
\begin{align}
 \left|\lsp\varphi, \P^{(0)}_{12} (w_{\nabla f})_{12}^{\phantom{()}} \P^{(2)}_{12}\psi\rsp\right| 
& \leq C (  1 + N^{-\frac12} \| \rho ^{\nabla}\|_{1}^\frac12  ) \| q_1 \psi \| \big(  N^{-1} \|\varphi \|^{2}+  \|q_{1}\varphi\|^{2}\big)^{1/2}, \\[1mm]
 \left| \lsp\varphi, \P^{(0)}_{12}(w_{f})_{12}^{\phantom{()}} \P^{(2)}_{12} \psi\rsp \right|  & \leq C  \|q_{1}\psi\|\left( N^{-1} \|\varphi\|^{2}+ \|q_{2}\varphi\|^{2}\right)^{1/2},\\[1mm]
 \left|\lsp\varphi, \P^{(0)}_{123} (w_{ff})_{123}^{\phantom{()}} \P^{(2)}_{123} \psi\rsp\right|
 & \leq C   \|q_{1}\psi\| \left( N^{-1} \| \varphi \|^2 +  \|q_1 \varphi \|^2 \right)^{1/2}.
\end{align}
\end{lem}
\begin{proof} We start with the second estimate, where we insert $\hat \ell \hat{\ell^{-1}}$ and shift $\hat \ell$ via \prettyref{lem:Projection-shift} to the left:
\begin{align}
\left| \lsp \varphi, \P^{(0)}_{12}(w_{f})_{12}^{\phantom{()}} \P^{(2)}_{12} \psi\rsp \right| & = \left| \lsp  \hat \ell_{+2} \varphi, \P^{(0)}_{12}(w_{f})_{12}^{\phantom{()}} \P^{(2)}_{12} \hat{\ell}^{-1} \psi\rsp\right| \notag\\
& \le \| \hat \ell_{+2} \varphi \| \| \hat {\ell}^{-1} q_1 q_2 \psi \| \notag\\[0.5mm]
& \le C ( \| q_1 \varphi \|^2 + N^{-1} \| \varphi \|^2 )^{1/2} \| q_1 \psi \| \label{eq:2q:estimates:l:shift}
\end{align}
where we used $\| \hat \ell_{+2} \varphi \|^2 = \langle \varphi , \hat n_{+2} \varphi 
\rangle  \le C ( \| q_1 \varphi \|^2 + N^{-1})$ and $ \| \hat {\ell}^{-1} q_1 q_2 \psi \|  \le \sqrt 2 \| q_1 \varphi \|$, see \prettyref{lem:l-inverse conversion}. The third bound follows in the same way. For the first bound, we must ensure that $i\nabla_i$ does not act on $q_i$ for both $i=1,2$, so we rewrite $(w_{\nabla f})_{12} = 2 i \nabla_1 \cdot f_{12} + (i\nabla_1 \cdot f_{12}) +  2 i \nabla_2 \cdot f_{21} + ( i \nabla_2 \cdot f_{21})$. Applying \eqref{eq:diag:bound:pgp} with constant function $g\equiv 1$ and $A\equiv i\nabla$, we further obtain $\| \nabla_2 p_2 \varphi \|^2 \le C N^{-1} \| \rho^{\nabla}\|_1 \| \varphi \|^2 $. Using this, the first bound of the lemma then follows using the strategy from \eqref{eq:2q:estimates:l:shift}.
\end{proof}

The $nq$-estimates for $n\ge3$ are relevant for the kinetic energy estimate in Lemma 
\ref{lem:Bad-kinetic-energy-decomp} and the proof of the norm approximation
in Proposition \ref{prop:Norm-approx-result}. In the norm approximation it
is of importance that $\nabla_{1}q_{1}$ always occurs with the auxiliary
dynamics (which always appears on the right side of the interaction term). We first state the estimates for the terms involving three $q$ projections.
\begin{lem} 
\label{lem:3q-estimate}There is a $C>0$ such that for all $\varphi,\psi\in L_{a}^{2}(\RR^{3N})$
\begin{align}
\left|\lsp\varphi,\P_{12}^{(1)}(w_{\nabla f})_{12}\P_{12}^{(2)}\psi\rsp\right| & \le  C\left(  N^{-1/2} \|q_{1}\varphi\|  +\|q_{1}q_2 \psi\| \right)  \left(  \| \nabla _1 q_{1} \psi \|  + \| q_1 \psi \| \right)   ,\label{eq:3q:bound:12}\\
\left|\lsp\varphi, \P_{12}^{(2)} (w_{\nabla f})_{12} \P_{12}^{(1)}\psi\rsp\right| & \le  C \|q_{1}q_{2}\varphi\| \left(  \|\nabla_{1}q_{1}\psi\| +  N^{-\frac{1}{2}}\|\rho^{\nabla}\|_{1}^{\frac{1}{2}} \| q_1 \psi \| \right) \label{eq:3q:bound:21}\\
\left|\lsp\varphi, \P^{(1)}_{12}(w_{f})_{12}^{\phantom{()}} \P^{(2)}_{12} \psi\rsp\right| & \le  C  \| \varphi \| \left( N^{-1} \|q_{1}q_{2}\psi\| + \| q_1 q_2 q_3 \psi \|^2 \right)^{1/2} ,\label{eq:3q:bound:3a}\\
\left|\lsp\varphi, \P^{(1)}_{123} (w_{ff})_{123}^{\phantom{()}} \P^{(2)}_{123} \psi\rsp\right| & \le C \| \varphi \| \left( N^{-1} \|q_{1}q_{2}\psi\|^2 + \| q_1 q_2 q_3 \psi \|^2 \right)^{1/2} .\label{eq:3q:bound:3b}
\end{align}
\end{lem}

\begin{proof}
The proof of the first two lines work similarly as in \prettyref{lem:2q-estimate-sym}. In order to have the expression $\nabla_{1}q_{1}\psi$ in the
bound, one uses $i\nabla_{1}\cdot f_{12}=i(\nabla_{1}\cdot f_{12})+f_{12}\cdot(i\nabla_{1})$.
Additionally, one inserts $\psi= \hat{\ell}\hat{\ell}^{-1}\psi$ and shifts $\hat{\ell}$
towards $\varphi$ as described in \prettyref{lem:2q-estimate-asym}. This gives for the first line
\begin{align}
\left|\lsp \hat \ell _{+1} \varphi, \P^{(1)}_{12}(w_{\nabla f})_{12}^{\phantom{()}} \P^{(2)}_{12}\hat \ell^{-1} \psi\rsp\right| & \le C \| \hat \ell_{+1} q_{1} \varphi\| \left(  \| \nabla_{1}q_{1}\psi\| + \|q_1 \psi\| \right) \notag\\
& \le C \left( N^{-1/2} \| q_1 \varphi \| + \|q_1 q_2 \varphi \| \right)  \left(  \| \nabla_{1}q_{1}\psi\| + \|q_1 \psi \| \right) .
\end{align}
For the second line, we bring the gradient again to the right-hand side. If it acts on a $p$, then we use $\|  \nabla_2 p_2 q_1  \psi \|^2 \le N^{-1} \| \rho^\nabla \|_1 \| q_1  \psi\|^2$, by \eqref{eq:diag:bound:pgp}. The proofs for the estimates with $w_{f}$ and $w_{ff}$ are straightforward.
\end{proof}

The next lemma treats all operators with four or more $q$ projections.

\begin{lem}
\label{lem:4q-estimate}There is a $C>0$ such that for all $\varphi,\psi\in L_{a}^{2}(\RR^{3N})$
\begin{align}
\left|\lsp\varphi, \P^{(2)}_{12}(w_{\nabla f})_{12}^{\phantom{()}} \P^{(2)}_{12}\psi\rsp\right|\le & C \left(  N^{-1} \| q_1 q_2 \varphi \|^2 + \|q_{1}q_{2}q_{3}\varphi\|^2 \right)^{1/2} \|\nabla_{1}q_{1}\psi\| ,\\
\left|\lsp\varphi, \P^{(2)}_{12}(w_{f})_{12}^{\phantom{()}} \P^{(2)}_{12} \psi\rsp\right|\le & C\|  \varphi\| \left( N^{-1} \| q_1 q_2 \psi \|^2 + \|q_{1}q_{2}q_3 \psi\|^2 \right)^{1/2} ,\\
\left|\lsp\varphi, \P^{(2)}_{123} (w_{ff})_{123}^{\phantom{()}} \P^{(2)}_{123} \psi\rsp\right|\le & C \|  \varphi\| \left( N^{-1} \| q_1 q_2 \psi \|^2 + \|q_{1}q_{2}q_3 \psi\|^2 \right)^{1/2}   ,\\
   \left|\lsp\varphi,\P^{(2)}_{123}(w_{ff})_{123}^{\phantom{()}} \P^{(3)}_{123} \psi\rsp\right|\le &  C\|\varphi \|\|q_{1}q_{2}q_{3}\varphi\|  ,\\
      \left|\lsp\varphi, \P^{(3)}_{123} (w_{ff})_{123}^{\phantom{()}} \P^{(3)}_{123}\psi\rsp\right|\le & C\|\varphi \|\|q_{1}q_{2}q_{3}\psi\|.
    \end{align}
\end{lem}
\begin{proof}
The proof works analogously to that of the previous lemma. Note that for our purpose, we do not need to exploit all $q$ projections in these estimates.
\end{proof}
\begin{lem}[2q estimates of kinetic terms, asymmetric terms]
    \label{lem:2q-asym-kin-estimate} Let $h$ be a self-adjoint operator on $L^2(\mathbb R^3)$ with $\| h p \| < \infty$. There exists a $C>0$ such that for all $\psi\in L_{a}^{2}(\RR^{3N})$
    \begin{align*}
       \left|\lsp\psi,\big( h_1 p_{1}-p_{1}h_{1}\big)p_{2}(w_{\nabla f})_{12} \P^{(2)}_{12}\psi\rsp\right| & \leq C N^{-\frac12} \| \rho^h \|_1^\frac12 \left(\|q_{1}\psi\|^2 +\|\nabla_{1}q_{1}\psi\|^2 \right)^{1/2} \left(\|q_{2}\psi\|^{2}+N^{-1}\right)^{1/2},\notag \\
       \left|\lsp\psi,\big(h_{1}p_{1}-p_{1}h_{1}\big)p_{2}(w_{f})_{12} \P^{(2)}_{12}\psi\rsp\right|  & \leq C  N^{-\frac12} \| \rho^h \|_1^\frac12 \|q_{1}\psi\|\left(\|q_{2}\psi\|^{2}+N^{-1}\right)^{1/2},\\  
     \left|\lsp\psi,\big(h_{1}p_{1}-p_{1}h_{1}\big)p_{2}p_{3}(w_{ff})_{123} \P^{(2)}_{123}\psi\rsp\right| &  \leq C  N^{-\frac12} \| \rho^h \|_1^\frac12  \|q_{1}\psi\|\left(\|q_{2}\psi\|^{2}+N^{-1}\right)^{1/2}.
     \end{align*}
     where $\rho^h = \sum_{k=1}^N | h \varphi_k|^2$.
\end{lem}

\begin{proof}
  The proof works analogous to \prettyref{lem:2q-estimate-asym}, where we had $p_1$ instead of $p_1 h_1 - h_1 p_1$, and where we used $\|p_{1}\psi\|\leq \|\psi\|$. Here we use instead the bounds $\|h_{1}p_{1}\psi\| \leq\|h_{1}p_{1}\|_{\as}\|\psi\|$ and $\|p_{1}h_1 \psi\| \le \| p_1 h_1 \|_{\as} \| \psi\|$ with  operator norm
  \begin{equation}
  \|A \|_{\as}\coloneqq\sup_{\substack{\psi\in L_{a}^{2}(\RR^{3N}),\\
  \|\psi\|=1
  }
  }\| A \psi\|.
  \end{equation}
  We have $\|h_{1}p_{1}\|_{\as}=\|p_{1}h_{1}\|_{\as}$
  since $h_1p_{1}$ is bounded by assumption. Invoking \prettyref{lem:diagonalization-estimate}
  with $g\equiv 1$ and $A\equiv h$, it holds that
  \begin{equation}\label{eq:h:op:as:bound}
  \| h_{1}p_{1}\|_{\as}\leq C N^{-\frac{1}{2}}\|\rho^{h}\|_{1}^{1/2}.
\end{equation}
With this bound in place, the proof proceeds in direct analogy to that of Lemma \ref{lem:2q-estimate-asym}.
\end{proof}

\begin{lem}[2q estimates of kinetic terms, symmetric terms]
  \label{lem:2q-sym-kin-estimate} Let $h$ be a self-adjoint operator on $L^2(\mathbb R^3)$ with $\| h p \| < \infty$. There exists a $C>0$ such that for all $\psi\in L_{a}^{2}(\RR^{3N})$
  \begin{align*}
     \left|\lsp\psi,\big(h_{1}p_{1}-p_{1}h_{1}\big)q_{2}(w_{\nabla f})_{12} \P^{(1)}_{12}\psi\rsp\right|  & \leq CN^{-\frac12} \| \rho^h \|_1^\frac12 \left(N^{-\frac{1}{2}}\|\rho^{\nabla}\|_{1}^{1/2}\|q_{1}\psi\|+\|\nabla_{1}q_{1}\psi\|\right)\|q_{2}\psi\|,\\
    \left|\lsp\psi,\big( h_{1} p_{1}-p_{1}h_{1}\big)q_{2}(w_{f})_{12} \P^{(1)}_{12}\psi\rsp\right|&  \leq C N^{-\frac12} \| \rho^h \|_1^\frac12 \|q_{1}\psi\|\|q_{2}\psi\|,\\
   \left|\lsp\psi,\big(h_{1}p_{1}-p_{1} h_{1}\big)q_{2}p_{3}(w_{ff})_{123} \P^{(1)}_{123}\psi\rsp\right|  & \leq C N^{-\frac12} \| \rho^h \|_1^\frac12 \|q_{1}\psi\|\|q_{2}\psi\|.
   \end{align*}
\end{lem}
 \begin{proof} With \eqref{eq:h:op:as:bound}, the proof follows in close analogy to the proof of Lemma \ref{lem:2q-estimate-sym}.
\end{proof}

\appendix

\section{Analysis of the Hartree solutions}

The goal of this section is to prove Lemma \ref{lem:bound:D(t)}. To this end, we state and prove two lemmas. The first one provides estimates of the energy of the gauged solutions in terms of the non-gauged ones. In the second one, we propagate the energy for the non-gauged solutions. Adding both together then proves Lemma \ref{lem:bound:D(t)}.

Recall that $\rho_{t}^{A}=\sum_{k=1}^{N}|A\psi_{k}^{t}|^{2}$, where $A$ is an operator on $L^2(\mathbb{R}^3)$ and $\{\psi_k^t\}_{k=1}^N$ is a solution to~\eqref{eq:mf-eq}. Similarly, let $\rho_{\varphi_t}^{A}$ denote the corresponding quantity defined with respect to the ungauged solutions $\{\varphi_k^t\}_{k=1}^N$ of the rescaled mean-field equations~\eqref{eq:mf-eq:rescaled}. In the next statement, we show that any $L^p$-norm estimate for $p \in [1, \infty]$ on the solutions $\{\varphi_k^t\}_{k=1}^N$ to the original mean-field equations carries over to the solutions $\{\psi_k^t\}_{k=1}^N$ of the gauged mean-field equations~\eqref{eq:hg}, since the two differ only by a phase factor.

\begin{lem}
  \label{lem:psi-phi-conversion} Let $p\in [1,\infty]$. There exists a constant $C>0$ such that for all $t\geq0$ 
  \begin{align*}
  \|i\partial_{t}\psi_{k}^{t}\|_{p} & \leq C\left(t^{2}+t \eN^{5/2}\|\rho_{\varphi_{t}}^{\nabla}\|_{1}\right)\|\varphi_{k}^{t}\|_{p}+\eN  \|\Delta\varphi_{k}^{t}\|_{p},\\
  \eN^{1/2} \|\nabla\psi_{k}^{t}\|_{p} & \leq C\left(t\|\varphi_{k}^{t}\|_{p}+ \eN^{1/2} \|\nabla\varphi_{k}^{t}\|_{p}\right),\\
  \eN \|\Delta\psi_{k}^{t}\|_{p} & \leq C\left((1+t)^{2}\|\varphi_{k}^{t}\|_{p}+t \eN^{1/2} \|\nabla\varphi_{k}^{t}\|_{p}+ \eN \|\Delta\varphi_{k}^{t}\|_{p}\right).
  \end{align*}
  Moreover, with $\rho_{\varphi^{t}}^{\nabla}\coloneqq\sum_{k=1}^{N}|\nabla\varphi_{k}^{t}|^{2},\rho_{\varphi^{t}}^{\Delta}\coloneqq\sum_{k=1}^{N}|\Delta\varphi_{k}^{t}|^{2}$,
  there exists a constant $C>0$ such that
    \begin{align*}
  \|\eN \rho_{t}^{\nabla}\|_{p} & \leq C\left(\|\eN \rho_{\varphi_{t}}^{\nabla}\|_{p}+t^{2}\|\rho_{t}\|_{p}\right),\\
  \|\eN^{2}\rho_{t}^{\Delta}\|_{p} & \leq C\left(\|\eN^{2}\rho_{\varphi_{t}}^{\Delta}\|_{p}+t^{2}\|\eN\rho_{\varphi_{t}}^{\nabla}\|_{p}+(t^{2}+t^{4})\|\rho_{t}\|_{p}\right).
  \end{align*}
  \end{lem}
  \begin{proof}
  The first inequality is an immediate consequence of
  $\psi_{k}^{t}=e^{ i t \eN  (v\ast \rho_t )}\varphi_{k}^{t}$ and \eqref{eq:mf-eq}, yielding
  \begin{equation}
  i\partial_{t}\psi_{k}^{t}= \left( -t\eN \partial_{t}(v\ast \rho_t) -\eN \Delta \right) \varphi_{k}^{t}.
  \end{equation}
  Moreover, by Cauchy-Schwarz and using Assumption \ref{ass1}, we estimate 
  \begin{align}
  \|t\eN \partial_{t}(v\ast \rho_t)\|_\infty & \leq Ct\left(\eN^{2}\||f|\ast\rho_{\varphi_{t}}^{\nabla}\|_{\infty}^{1/2}\||f|\ast\rho_{t}\|_{\infty}^{1/2}+t\eN^{3}\||f|\ast\rho_{t}\|_{\infty}^{2}\right)\nonumber \\
   & \leq Ct\left(\eN^{2}\|\rho_{\varphi_{t}}^{\nabla}\|_{1}^{1/2}\|\rho_{t}\|_{1}^{1/2}+t\eN^{3}\|\rho_{t}\|_{1}^{2}\right).\label{eq:v_t-deriv}
  \end{align}
  Furthermore, it holds
  \begin{align}
  \nabla\psi_{k}^{t} & =e^{it \eN (v\ast \rho_t )  }\left( it\eN \overline{f}\varphi_{k}^{t}+\nabla\varphi_{k}^{t}\right),\\
  \Delta\psi_{k}^{t} & =e^{i t \eN (v\ast \rho_t )  }\left(-t^{2}\eN^2\overline{f}^{2}\varphi_{k}^{t}+ it\eN (\nabla\cdot\overline{f})\varphi_{k}^{t}+ it\eN ( \overline{f}\cdot\nabla) \varphi_{k}^{t}+\Delta\varphi_{k}^{t}\right).
  \end{align}
  We find the desired estimates using \eqref{eq:v_t-deriv} with $\|\nabla\cdot\overline{f}\|_{\infty},\|\overline{f}\|_{\infty}\leq CN$
  and $\eN \|\rho_{t}^{\nabla}\|_{1}\leq ND(t)^{2}$. In addition,
  it holds by direct computation that
  \begin{align}
  \|\eN \rho_{t}^{\nabla}\|_{p} & \leq C\left(\|\eN \rho_{\varphi_{t}}^{\nabla}\|_{\infty}+t^{2}\eN^{3}\|\overline{f}\|_{\infty}^{2}\|\rho_{t}\|_{p}\right),\\
  \|\eN^2 \rho_{t}^{\Delta}\|_{p} & \leq C\Big(\|\eN^{2}\rho_{\varphi_{t}}^{\Delta}\|_{\infty}+t^{2}\eN^{3}\|\overline{f}\|_{\infty}^{2}\|\eN^{2}\rho_{\varphi_{t}}^{\nabla}\|_{p} +\left(t^{4}\eN^{3}\|\overline{f}\|_{\infty}^{4}+t^{2}\eN^{4}\|\nabla\cdot\overline{f}\|_{\infty}^{2}\right)\|\rho_{t}\|_{p}\Big).
  \end{align}
  \end{proof}
  
Since the solutions $\{\varphi_{k}^{t}\}_{k=1}^{N}$
of \eqref{eq:mf-eq} satisfy energy conservation in the sense that
$E(t)\coloneqq\langle\varphi_{k}^{t},( -\Delta + \tfrac12 v\ast \rho_t ) \varphi_{k}^{t}\rangle = E(0)$ for all $t\in \mathbb R$, we directly obtain the bound
\begin{equation}
\sum_{k=1}^{N}\eN \|\nabla\varphi_{k}^{t}\|^{2}\leq E(t)+\eN \sum_{k=1}^{N}|\langle\varphi_{k}^{t}, (v\ast \rho_t)  \varphi_{k}^{t}\rangle|\leq\sum_{k=1}^{N}\eN \|\nabla\varphi_{k}^{0}\|^{2}+2CN^{\frac{4}{3}}.
\end{equation}
By Assumption \ref{ass1}, the first term on the right side is bounded by a constant times $N$. In the next proposition, we use a Grönwall argument to show that the upper bound can in fact be improved to be proportional to $N$. 

\begin{lem}
\label{prop:H1-H2-bd}Let $\{\varphi_{k}^{t}\}_{k=1}^{N}$ be the solutions to \eqref{eq:mf-eq} with initial conditions satisfying Assumption \ref{ass1}. Then, there is a constant $C>0$ such that for all $t\ge0$
\begin{align}
\eN \sum_{k=1}^N \|\nabla\varphi_{k}^{t}\|^{2} + \eN^2 \sum_{k=1}^N \|\Delta \varphi_{k}^{t}\|^{2} \leq C N e^{C(t+1)}.
\end{align}
\end{lem}

\begin{proof}
We neglect the implicit $t$-dependencies in the notation.
It holds by Cauchy-Schwarz that
\begin{align}
\left|\frac{\d}{\d t} \eN \|\nabla\varphi_{k}^{t}\|^{2}\right|  =\eN^2 \left|\langle\varphi_{k}^{t},\left[-\Delta, v\ast\rho_t \right]\varphi_{k}^{t}\rangle\right| & =\eN^{2}\left|\langle\varphi_{k}^{t},( \Delta v\ast\rho_{t})\varphi_{k}^{t}+2\left(\nabla v\ast\rho_{t}\right)\cdot \nabla\varphi_{k}^{t}\rangle\right|\nonumber \\
 & \leq2\eN^{2}\| \nabla v\ast\rho_{t}\|_{\infty}\|\eN\nabla\varphi_{k}^{t}\|+\eN^2 \| \Delta v\ast\rho_{t}\|_{\infty}\|\varphi_{k}^{t}\|\nonumber \\
 & \leq 2\eN^2 N \|\nabla v\|_{\infty}\|\nabla\varphi_{k}^{t}\|+ \eN^2 N \|\Delta v\|_{\infty}\nonumber \\
 & \leq C\left( \eN \| \nabla\varphi_{k}^{t}\|^{2} + 1\right).
\end{align}
Similarly, one shows that
\begin{equation}
\left|\frac{\d}{\d t}\left( \eN^2 \|\Delta\varphi_{k}^{t}\|^{2}+ \eN \| \nabla\varphi_{k}^{t}\|^{2}\right)\right|\leq C\left( \eN^2 \|\Delta\varphi_{k}^{t}\|^{2}+\eN \| \nabla\varphi_{k}^{t}\|^{2}+1 \right).
\end{equation}
Thus, with the aid of Grönwall's inequality and Assumption \ref{ass1}, we obtain the claimed bound.
\end{proof}

\begin{proof}[Proof of Lemma \ref{lem:bound:D(t)}.]
The proof is a direct consequence of Lemmas  \ref{lem:psi-phi-conversion} and \ref{prop:H1-H2-bd}.
\end{proof}

\bigskip

\noindent \textbf{Acknowledgements}. We thank Esteban C\'ardenas and Sören Petrat for valuable discussions. This work was funded by the Deutsche Forschungsgemeinschaft (DFG, German Research  
Foundation) -- TRR 352 -- Project-ID 470903074. \bigskip

\noindent \textbf{Data Availability}. Data sharing is not applicable to this article as no datasets were generated or analysed during the current study.\bigskip

\noindent \textbf{Conflict of interests}. The authors have no competing interests to declare that are relevant to the content of this article.
\end{spacing}

\bibliographystyle{amsplain}
\bibliography{Hartree-Fock-Fermions}

\end{document}